\documentclass[11pt,letterpaper]{article}
\ifdefined\ACM
\else
\usepackage{amsthm}
\fi

\usepackage[bookmarks=true,pdfstartview=FitH,colorlinks,linkcolor=darkred,filecolor=darkred,citecolor=darkred,urlcolor=darkred]{hyperref}

\usepackage{colortbl}
\usepackage{mdframed}
\usepackage{color,soul}
\usepackage{paralist}
\usepackage{pdfpages}
\usepackage{enumitem}
\usepackage{float}
\usepackage{booktabs}
\usepackage[override]{cmtt}
\usepackage{wrapfig}
\usepackage{boxedminipage}
\usepackage{url}
\usepackage{graphicx}
\usepackage{xspace}
\usepackage{multirow}
\usepackage{amsmath,amstext,amsfonts,latexsym}
\usepackage{amsbsy}
\usepackage{stmaryrd}
\usepackage{centernot}
\usepackage{pifont}
\usepackage{color,xcolor}
\usepackage{graphicx,color,eso-pic}
\usepackage[
    lambda,
    advantage,
    operators,
    sets,
    adversary,
    landau,
    probability,
    notions,
    logic,
    ff,
    mm,
    primitives,
    events,
    complexity,
    asymptotics,
    keys]
{cryptocode}
\usepackage{tikz}
\usetikzlibrary{fit, shapes, positioning, patterns, arrows,shadows}
\usetikzlibrary{graphs,graphs.standard,arrows}
\usetikzlibrary{decorations.markings}

\tikzset{
    master/.style={
        execute at end picture={
            \coordinate (lower right) at (current bounding box.south east);
            \coordinate (upper left) at (current bounding box.north west);
        }
    },
    slave/.style={
        execute at end picture={
            \pgfresetboundingbox
            \path (upper left) rectangle (lower right);
        }
    }
}

\sethlcolor{lightgray}

\newcommand{\truthful}{truthful}
\newcommand{\Truthful}{Truthful}

   \newcommand{\bfa}{\mathbf{a}}
\newcommand{\bfB}{\mathbf{B}}   \newcommand{\bfb}{\ensuremath{{\mathbf{b}}}\xspace}
\newcommand{\bfsig}{\boldsymbol{\sigma}}   
\newcommand{\bfC}{\mathbf{C}}

   \newcommand{\bff}{\mathbf{f}}

\newcommand{\bfI}{\mathbf{I}}

\newcommand{\bfP}{\mathbf{P}}   \newcommand{\bfp}{\ensuremath{{\mathbf{p}}}\xspace}
   
   \newcommand{\bfr}{\mathbf{r}}
   \newcommand{\bfs}{\mathbf{s}}

   \newcommand{\bfv}{\mathbf{v}}
   
   \newcommand{\bfx}{\mathbf{x}}

\definecolor{darkred}{rgb}{0.5, 0, 0}
\definecolor{darkgreen}{rgb}{0, 0.5, 0}
\definecolor{darkblue}{rgb}{0,0,0.5}

\newcommand\markx[2]{}

\renewcommand{\path}{\ensuremath{\mathsf{path}}\xspace}

\newcommand{\ignore}[1]{}

\newcommand{\mcal}[1]{\ensuremath{\mathcal {#1}}}

\newcommand{\R}{\ensuremath{\mathbb{R}}}

\newcommand{\E}{\mathbb{E}}

\definecolor{darkgreen}{rgb}{0,0.5,0}
\definecolor{lightblue}{RGB}{0,176,240}
\definecolor{darkblue}{RGB}{0,112,192}
\definecolor{lightpurple}{RGB}{124, 66, 168}
\definecolor{grey}{RGB}{139, 137, 137}
\definecolor{maroon}{RGB}{178, 34, 34}
\definecolor{green}{RGB}{34, 139, 34}
\definecolor{types}{RGB}{72, 61, 139}
\definecolor{gold}{rgb}{0.8, 0.33, 0.0}

\definecolor{darkgray}{gray}{0.3}

\newcounter{task}

\ifdefined\ACM
\else
\theoremstyle{plain}
\newtheorem{theorem}{Theorem}[section]
\newtheorem{claim}[theorem]{Claim}

\newtheorem{lemma}[theorem]{Lemma}
\newtheorem{corollary}[theorem]{Corollary}

\theoremstyle{definition}
\newtheorem{remark}[theorem]{Remark}
\newtheorem{definition}[theorem]{Definition}

\fi

\usepackage[capitalize,noabbrev]{cleveref} 
\crefname{lemma}{Lemma}{Lemmas}
\Crefname{lemma}{Lemma}{Lemmas}

\usepackage[margin=1in]{geometry}

\newcommand{\elaine}[1]{{\footnotesize\color{magenta}[Elaine: #1]}}
\newcommand{\Hao}[1]{{\footnotesize\color{blue}[Hao: #1]}}
\renewcommand{\elaine}[1]{}
\renewcommand{\Hao}[1]{}

\newcounter{cnt:challenge}

\begin{document}
%
\title{Collusion-Resilience in Transaction Fee Mechanism Design}
\author{
Hao Chung\thanks{Supported by NSF awards 2212746, 2044679, 1704788, a Packard Fellowship,
a generous gift from the late Nikolai Mushegian, a gift from Google, and an ACE center grant
from Algorand Foundation.} \\ Carnegie Mellon University \\ {\tt haochung@andrew.cmu.edu}
\and Tim Roughgarden\thanks{Author's research at Columbia University supported in part by NSF awards
CCF-2006737 and CNS-2212745, and research awards from the Briger Family Digital Finance Lab and the Center
for Digital Finance and Technologies.} \\ Columbia University and a16z crypto \\ {\tt tim.roughgarden@gmail.com}
\and Elaine Shi\footnotemark[1] \\ Carnegie Mellon University \\{\tt runting@cs.cmu.edu}
    }
\date{\vspace{-20pt}}

\maketitle
\thispagestyle{empty}

\begin{abstract}
Users bid in a transaction fee mechanism (TFM) to get their transactions included and confirmed by a blockchain protocol. 
Roughgarden (EC'21) initiated the formal treatment of TFMs and proposed three requirements: user incentive compatibility (UIC), miner incentive compatibility (MIC), and a form of collusion-resilience called OCA-proofness. 
Ethereum's EIP-1559 mechanism satisfies all three properties simultaneously when there is no contention between transactions, 
but loses the UIC property when there are too many eligible transactions to fit in a single block. 
Chung and Shi (SODA'23) considered an alternative notion of collusion-resilience, called $c$-side-contract-proofness ($c$-SCP), and showed that, when there is contention between transactions, no TFM can satisfy UIC, MIC, and $c$-SCP for any $c\geq 1$. 
OCA-proofness asserts that the users and a miner should not be able to ``steal from the protocol.''
On the other hand, the $c$-SCP condition requires that a coalition of a miner and a subset of users should not be able to profit through strategic deviations (whether at the expense of the protocol or of the users outside the coalition).

Our main result is the first proof that, when there is contention between transactions, no (possibly randomized) TFM in which users are expected to bid truthfully satisfies UIC, MIC, and OCA-proofness.
This result resolves the main open question in Roughgarden (EC'21). 
We also suggest several relaxations of the basic model that allow our impossibility result to be circumvented.


\end{abstract}

\newpage
\tableofcontents
\thispagestyle{empty}
\newpage
\setcounter{page}{1}

\section{Introduction}

Real estate on the blockchain is scarce,
and blockchain users bid in an auction called
the transaction fee mechanism (TFM)
to have their transactions included and confirmed on the blockchain.
The original Bitcoin protocol adopted a simple first-price
auction, where the top $k$ bids win and they each pay their bid. 
However, such first-price auctions are known to incentivize
untruthful bidding. Therefore, a line of subsequent
works~\cite{zoharfeemech,yaofeemech,functional-fee-market,eip1559,roughgardeneip1559,roughgardeneip1559-ec,dynamicpostedprice,foundation-tfm,crypto-tfm,LP-tfm,greedy-tfm,bayesian-tfm,active-miner-tfm,pos-tfm,tiered-tfm,sequencer,optimal-base-fee,optimal-base-fee2,optimal-base-fee3} explored what is the ``dream TFM'' for blockchains. 
Most works \cite{roughgardeneip1559,roughgardeneip1559-ec, foundation-tfm,crypto-tfm,LP-tfm,greedy-tfm,bayesian-tfm,active-miner-tfm,pos-tfm} agree on roughly the same 
set of desiderata, that is, a dream TFM should
provide incentive compatibility not just
for an individual user, but also for the miner of the block. 
Further, a dream TFM should provide 
resilience against miner-user collusion.

Roughgarden~\cite{roughgardeneip1559-ec}
was the first to formally define the aforementioned requirements
for a TFM, which he referred to as {\it user incentive
compatibility}\footnote{User incentive compatibility (UIC) is usually
called dominant-strategy incentive compatible (DSIC) in the mechanism
design literature. In general, we allow UIC TFMs to make use of
non-truthful (but dominant) bidding strategies (see \cref{def:UIC}).}, 
(myopic) {\it miner incentive compatibility}, and 
{\it OCA-proofness}, where OCA stands for
``off-chain agreement'' and refers to colluding strategies
between the miner and a set of users that allow offchain transfers.
Roughgarden~\cite{roughgardeneip1559-ec} also showed that 
the simple ``posted price auction with all fees burnt'' mechanism,
which corresponds to the behavior of Ethereum's EIP-1559
TFM~\cite{eip1559} when there is no congestion, satisfies
all three properties.
However, the posted price auction with all fees burnt
does not satisfy all three properties when there is congestion.
In practice, congestion does occur
when there are major events such as an NFT mint or price
fluctuations --- for example,  
in Ethereum, roughly 2.3\% of the blocks experience
congestion.\footnote{From Jan 1, 2024 to Feb 5, 2024, 256595 blocks
  have been produced on Ethereum, and 5840 blocks among them were full
  (meaning more than 99.9\% of the gas limit (30M) was used).}
When congestion arises, 
approximately speaking, 
Ethereum's EIP-1559 mechanism falls back
to the first-price auction,
violating user incentive compatibility.
Therefore, an interesting question
is whether we can design a dream TFM satisfying all three properties
for finite block sizes.

Chung and Shi~\cite{foundation-tfm} considered an alternative notion of collusion-resilience, called side-contract-proofness.
Unfortunately, they proved that no (even randomized)
TFM can simultaneously satisfy user incentive compatibility and side-contract-proofness.
Because side-contract-proofness is a more demanding property than OCA-proofness,
the question raised by Roughgarden \cite{roughgardeneip1559-ec}, of
whether there is a dream TFM satisfying all three properties under his
collusion-resilience notion, had remained open.

\paragraph{Two notions of miner-user collusion-resilience}
Multiple natural notions of collusion-resilience 
can and have been
studied in the context of TFM design. Here we clarify informally the
key differences between the notions proposed by Roughgarden
\cite{roughgardeneip1559-ec} and Chung and Shi~\cite{foundation-tfm}.
These notions are defined formally in
Definitions~\ref{def:cSCP}--\ref{def:OCAproof} (see
Section~\ref{s:defs}) and compared further via examples in Appendix~\ref{sec:comparison}.

\begin{itemize}
\item 
{\bf OCA-proofness}:
Roughgarden's notion, henceforth referred to as {OCA-proofness},
asserts that there should exist a ``reference strategy'' for a miner
and all users that is guaranteed to maximize their joint surplus.
In this reference strategy, the miner is expected to follow the
inclusion rule intended by the TFM. For users, the definition requires
only that users follow some fixed bidding strategy~$\sigma$ (i.e., a
mapping from a private user valuation to a user bid) that is
individually rational (i.e., $\sigma(v) \le v$ for all $v \ge 0$). In
particular, in the reference strategy, users are expected to bid
independently (with a user's bid independent of other users'
valuations and bids), and expected to submit a single bid (with no additional
fake bids injected). One example of such a bidding strategy is the
truth-telling strategy (with $\sigma(v) = v$).
Because Roughgarden~\cite{roughgardeneip1559-ec} wished to
discuss the OCA-proofness properties of non-UIC TFMs like first-price
auctions, the definition also allows the reference strategy to be
defined by a non-truthful bidding strategy (e.g., $\sigma(v) = v/2$).
As a consequence, to prove that a TFM is both UIC and OCA-proof, it is
sufficient to prove that it is UIC under one bidding strategy and
OCA-proof under a possibly different bidding strategy (as in the
example in Section~\ref{sec:non-direct-oca-counterexample}).

\item 
{\bf $c$-SCP}:
Chung and Shi's notion~\cite{foundation-tfm}, henceforth called
{$c$-SCP} (where SCP stands for  
side-contract-proofness), 
requires that the honest strategy
(i.e., all users follow the honest bidding rule and the miner honestly implements
the inclusion rule)
is the profit-maximizing strategy 
for any coalition consisting of the miner
of the present block and at most $c$ users.
For \truthful{} mechanisms, the honest bidding rule is the truthful one,
while for non-\truthful{} mechanisms, the bidding rule can be more general (see \cref{sec:tfm-def} for the formal definition).
Chung and Shi's notion
aligns with standard notions
used in a line of work at the intersection
of game theory and cryptography ~\cite{gtcrypto00,gtcrypto01,gtcrypto02,gtcrypto03,giladutilityindjournal,giladgtcrypto,rdp00,rdp01,rdp02,katzgametheory,gtcrypto06,seqrationalcrypto,gt-fair-cointoss,gt-fair-coin-complete,gt-leader-shi,fruitchain,logstar-gt-leader,credibleauction-comm00,credibleauction-comm01}.
\end{itemize}

\paragraph{Discussion}
The two notions of collusion-resilience
address different issues.
OCA-proofness 
captures the intuitive requirement 
that the users and miners {\it should   
not be able to steal from the protocol}
through strategic deviations --- for this reason,
it considers {\it only the global coalition}
consisting of the miner and all users.
By contrast, the $c$-SCP notion 
captures the intuitive idea that 
a miner-user coalition's best response
is to act honestly,  
and that 
no strategic deviations can 
allow the coalition to {\it steal from other users
or steal from the protocol}.
For further discussion, see the end of this section and \cref{sec:comparison}.


\subsection{Our Contributions}

As explained above, both Roughgarden's and Chung and Shi's collusion-resilience
notions capture meaningful 
incentive compatibility considerations. 
Recognizing their differences, 
one natural question is: 
does Chung and Shi's finite-block impossibility result
still hold if we adopt 
the original OCA-proofness notion of Roughgarden in lieu 
of $c$-SCP? 
Notably, no existing TFM construction~\cite{zoharfeemech,yaofeemech,functional-fee-market,eip1559,roughgardeneip1559,roughgardeneip1559-ec,dynamicpostedprice,foundation-tfm,crypto-tfm,LP-tfm,greedy-tfm,bayesian-tfm,active-miner-tfm,pos-tfm,tiered-tfm,sequencer,optimal-base-fee,optimal-base-fee2,optimal-base-fee3} 
simultaneously satisfies
user incentive compatibility, 
miner incentive compatibility, 
and OCA-proofness under finite block size. 

\paragraph{Main impossibility result}
In our work, we give an affirmative answer  
to the above question. We show that, indeed,
an analog of Chung and Shi's 
finite-block impossibility result
still holds when we replace the $c$-SCP requirement with OCA-proofness.
Specifically, we prove the following theorem. 

\begin{theorem} 
Suppose the block size is finite.
Then, no possibly randomized, \truthful{} TFM can simultaneously satisfy
user incentive compatibility (UIC), miner incentive compatibility (MIC),
and OCA-proofness.
Further, this impossibility holds even when
the globally optimal strategy $\sigma$ need
not be individually rational.
\label{thm:intromain}
\end{theorem}

In a \truthful{} TFM, a user is expected to bid truthfully,
so if the mechanism satisfies UIC, a user's utility is maximized when it just reports its true value.
However, OCA-proofness allows the global coalition to adopt a non-truthful bidding strategy $\sigma$ even for \truthful{} mechanisms.

\ignore{
As mentioned, since $c$-SCP for any $c$ implies
OCA-proofness,  \Cref{thm:intromain}
implies the impossibility of simultaneously
achieving UIC, MIC, and 1-SCP. 
While it may seem like 
our result is a strict generalization
of Chung and Shi's impossibility, there are actually  
interesting subtleties.
}
Our \Cref{thm:intromain} 
is intuitively stronger but technically incomparable in comparison with Chung and Shi's impossibility,  
which shows that no TFM can simultaneously satisfy UIC  
and 1-SCP for finite block sizes.
The reason is that Chung and Shi's impossibility
does not rely on MIC; however, 
MIC is necessary for our \Cref{thm:intromain} to hold. 
Specifically,  a simple second-price auction 
with no burning (see \cref{remark:second-price-no-burning})
satisfies both UIC and OCA-proofness, but does not satisfy
MIC since the miner may benefit by injecting
a fake $(t+1)$-th bid 
where $t$ is the number of confirmed bids, since the $(t+1)$-th bid
sets the price for confirmed bids.

\paragraph{Global SCP}
We suggest a simpler version of OCA-proofness
that we call {\bf global SCP}, which also intuitively captures 
the requirement that strategic users and miners 
cannot steal from the protocol, and is perhaps more appropriate when
focusing on UIC TFMs (as we do in this paper).
In our work, global SCP is not only a 
technical {\it stepping stone} towards proving 
\Cref{thm:intromain}, 
but also of {\it independent interest} as we explain below.
Specifically, 
{global SCP} is almost the same as OCA-proofness, except for 
requiring $\sigma$
to be the honest bidding strategy indicated by the mechanism (i.e.,
the same bidding strategy used to establish UIC).
In other words, a mechanism satisfies global SCP
if and only if the honest strategy is surplus-maximizing for the global
coalition.
It is easy to see that 
for a \truthful{} mechanism, 
$c$-SCP for any $c$ 
implies global SCP, which in turn implies OCA-proofness. 
To prove \Cref{thm:intromain},  
we first prove the following theorem:  

\begin{theorem} 
Suppose that the block size is finite.
Then no possibly randomized TFM can simultaneously satisfy
user incentive compatibility (UIC), miner incentive compatibility (MIC),
and global SCP.
Further, the impossibility holds even for non-\truthful{} mechanisms. 
\label{thm:introglobalscp}
\end{theorem}

We now explain why the global SCP notion is of independent interest.
One advantage of global SCP 
is that the {\it revelation principle} holds
for any TFM that satisfies UIC, MIC, and global SCP, which we formally
prove in~\cref{sec:revelation}. 
In other words, given any TFM that is UIC, MIC, and global SCP,
there is an equivalent \truthful{} mechanism
that simulates it.   
\elaine{comment on the subtlety later}
\ignore{
Note that although the direct revelation 
principle is straightforward for classical auctions  
where we care only about achieving UIC in the presence
of an honest auctioneer, 
whether the principle
still holds for TFMs is subtle, partly due
to the additional MIC and collusion-resilience requirements,
and partly due to the separation of the inclusion rule 
executed by the miner 
and the confirmation/payment rules executed by the blockchain.
}
For this reason, 
\Cref{thm:introglobalscp}
rules out 
even non-\truthful{} TFMs
that simultaneously satisfy 
UIC, MIC, and global SCP.\footnote{Simultaneously with and independently of this paper, Gafni and Yaish~\cite{barrier-tfm} proved, among other results, a version of \cref{thm:introglobalscp} for the special case of deterministic mechanisms and a block size of 1.}

By contrast, 
\Cref{thm:intromain}
holds only for \truthful{} mechanisms.
In particular, in \cref{sec:non-direct-oca-counterexample}, we show a 
non-\truthful{} 
mechanism that simultaneously satisfies
UIC, MIC, and OCA-proof.
The mechanism is contrived, 
but it demonstrates the subtlety and the technical challenges when modeling the notion of collusion-resilience.
\Hao{TODO: check whether this philosophical statement is proper.}
This also suggests that the revelation principle does not hold 
for mechanisms that satisfy UIC, MIC, and OCA-proofness, partly because
in such a mechanism, 
the bidding strategies used to establish UIC and OCA-proofness may be different.

\paragraph{Ways to circumvent the impossibilities}
We show in \cref{sec:bypass-impossibility}
that the impossibility of \Cref{thm:intromain} 
can be circumvented 
by allowing non-\truthful{} mechanisms or by
allowing users to coordinate in bidding 
in the globally optimal strategy $\sigma$. 
Additionally, in \cref{sec:BayesianIC}, we show that the impossibility
can also be circumvented by adopting a Bayesian notion of incentive compatibility.
Bayesian incentive compatibility is particularly justified when the strategic players must submit their bids 
without having seen other honest users' bids---for example,
when honest users' bids are encrypted 
and block production is performed by trusted hardware or through secure multi-party computation~\cite{crypto-tfm}.

\subsection{TFM Incentive-Compatibility Notions: A Cheat Sheet}

We gather here informal definitions and comparisons of the key
incentive-compatibility notions used in this paper.
First, a transaction fee mechanism specifies how a user is supposed to
bid (as a function of its private valuation), which transactions a miner
is supposed to include (as a function of the transactions it knows
about and their bids), and the resulting outcome (the subset of
included transactions that get confirmed, and the payments made by the
users and received by the miner). 
If the bidding strategy suggested by the TFM is the identity, 
then we additionally call the TFM {\em truthful}.
In this paper, as in the rest of the
TFM literature, we consider only static mechanisms.
\begin{itemize}

\item \textbf{UIC.} (\cref{def:UIC}) 
Provided that the miner follows the suggested inclusion rule, the bidding strategy suggested by the TFM is a
  dominant strategy for users. 

\item \textbf{MIC.} (\cref{def:MIC}) The inclusion rule suggested by the TFM
  is always revenue-maximizing for the miner regardless of users' bids; moreover, the miner
  cannot increase its revenue through the injection of fake transactions.

\item \textbf{Global SCP.} (\cref{def:globalSCP}) If the miner follows the inclusion rule
  suggested by the TFM and all users follow the bidding rule suggested
  by the TFM, then their joint surplus is at least as large as it would be from
  any coordinated deviation.

\item \textbf{$c$-SCP.} (\cref{def:cSCP}) 
  For every coalition of the miner and at most~$c$ users, 
  if the miner follows the inclusion rule suggested by the TFM 
  and the users in the coalition follow the bidding rule suggested by the TFM,
  then the joint surplus of the coalition is at least as large as it would be
  from any coordinated deviation (holding fixed the bids of users
  outside the coalition).

\item \textbf{OCA-proof.} (\cref{def:OCAproof}) If the miner follows
  the inclusion rule
  suggested by the TFM and all users follow a suitably chosen individually
  rational bidding rule~$\sigma$ (possibly different from the one
  suggested in the TFM description), then their joint surplus is as
  large as it would be from any coordinated deviation.

\end{itemize}
For example, in \cite{roughgardeneip1559-ec} it was shown that
Ethereum's EIP-1559 TFM and a variant called the ``tipless mechanism''
satisfy UIC, MIC, and OCA-proofness when there is no contention
between transactions; in fact, in this case, these TFMs satisfy the
$c$-SCP condition for every $c \ge 1$. When there is contention
between transactions for inclusion in a block, the EIP-1559 TFM loses
its UIC property and the tipless mechanism loses (all three notions
of) collusion-resilience.

As mentioned above:
\begin{itemize}

\item (\cref{thm:revelation}) A relevation principle holds for the
  global SCP and $c$-SCP
  notions: any UIC and MIC TFM that satisfies one of these properties
  can be simulated by a truthful UIC and MIC TFM that satisfies the same property.

\item A relevation principle does not in general hold for the
  OCA-proof notion: while there are non-truthful TFMs that satisfy
  UIC, MIC, and OCA-proofness
  (\cref{sec:non-direct-oca-counterexample}), there are no truthful
  TFMs that satisfy UIC, MIC, and OCA-proofness
  (\cref{thm:oca-random-impossibility}).

\end{itemize}

The main result in Chung and Shi~\cite{foundation-tfm} states that,
even among randomized TFMs, no TFM satisfies UIC and $c$-SCP for
any $c \ge 1$. 
Our \cref{thm:random-scp-impossible} proves that,
even among randomized TFMs, no TFM satisfies UIC, MIC, and global SCP.
(Due to the revelation principle mentioned above, these impossibility
results apply to both truthful and non-truthful TFMs.)
Our \cref{thm:oca-random-impossibility} proves the stronger statement that,
even among randomized TFMs, no truthful TFM satisfies UIC, MIC, and
OCA-proofness.

Reflecting on the competing notions of collusion-resilience, we can
observe the following. The $c$-SCP condition may be particularly
appropriate in scenarios where the primary concern is deviations by
small coalitions, or in scenarios where users may wish to deviate in
ways that exploit other users. The $c$-SCP condition is also notable
in that, together with the UIC condition, it already triggers the
impossibility result in~\cite{foundation-tfm} (without any appeal to
MIC). The OCA-proofness condition is distinguished by being the
weakest of the three notions (thus leading to the strongest
impossibility results) and by allowing the discussion of non-UIC
mechanisms.\footnote{For example, in a first-price auction,
  the ``reference outcome'' might be
  defined by a (non-truthful) bidding strategy that would constitute a
  Bayes-Nash equilibrium with respect to some prior over user
  valuations
(cf., Corollary~5.12 in \cite{roughgardeneip1559-ec}).}
For TFMs that are UIC and MIC, like those studied in this paper,
global SCP is arguably the ``right'' definition---capturing the spirit
of OCA-proofness, without any additional technical complications
arising from users using different bidding strategies to satisfy UIC
and collusion-resilience. 
Put differently, 
the UIC and MIC conditions imply that the miner and the users following their intended strategies constitutes a Nash equilibrium; 
the global SCP condition asserts that this Nash equilibrium is also robust to deviations by the grand coalition, 
while OCA-proofness only asserts such robustness for a possibly different strategy profile 
(defined by the intended inclusion rule and some individually rational bidding strategy).
From this vantage point, one might view 
\cref{thm:random-scp-impossible} as the main impossibility result in
this paper, with
\cref{thm:oca-random-impossibility} serving as a technically challenging
extension of the result under still weaker incentive-compatibility conditions.

\section{Definitions}\label{s:defs}


\subsection{Transaction Fee Mechanism in the Plain Model}
\label{sec:tfm-def}
In this work, we assume each user has a true value $v \in \mathbb{R}_{\geq 0}$ for its transaction being confirmed on the blockchain, where $\R_{\geq 0}$ denotes non-negative reals.
We assume all transactions have an equal size, and we define the \emph{block size} 
to be the maximal number of transactions that can be included in a block.
If there exists a finite number $k$ such that a block can include at most $k$ transactions, we say the block size is \emph{finite};\footnote{The finite block size regime in this work and \cite{foundation-tfm} corresponds to the case in \cite{roughgardeneip1559-ec} where the base fee in the EIP-1559 or tipless mechanisms is \emph{excessively low}, i.e.~the number of transactions willing to pay the base fee exceeds the maximum block size
(cf., Definition 5.6 in \cite{roughgardeneip1559-ec}).
}
otherwise, we say the block size is \emph{infinite}.
When the block size is finite, we always assume $k \geq 1$.
Additionally, we assume that the users' utilities are independent of the order of the transactions in a block.
In practice, transactions may have different sizes and the order of the transactions matters.
For example, the size is measured in ``gas'' in Ethereum, and the prices for trading may depend on the order of the transactions.
However, for our impossibility results, these assumptions only make them stronger.

A transaction fee mechanism (TFM) consists of 
the following possibly randomized algorithms:
\begin{itemize}
\item 
{\bf Bidding rule} (executed by the user):
takes in the user's true valuation $v \in \R_{\geq 0}$ for getting  
its transaction confirmed, 
and outputs a vector of a non-negative number of real-valued bids. 
Without loss of generality, we assume that 
at most one bid in the output vector 
may correspond to the user's actual transaction\footnote{The 
blockchain protocol can always suppress
conflicting or double-spending transactions.}
which has a true value of $v$ 
(henceforth called the {\it primary} bid), and 
the remaining bids are {\it fake bids} with a true value of $0$.
We say a mechanism is {\it \truthful{}} if the honest bidding rule
asks the user to submit a single primary bid that reports its true value.
Otherwise, if the honest bidding rule may ask the user to submit bids other than its true value, the mechanism is said to be non-\truthful{}.
Our formulation is more general and 
admits non-\truthful{} mechanisms, too.\footnote{Throughout the paper
  except \cref{sec:revelation}, we only focus on bidding rules that
  output a single bid. In \cref{sec:revelation}, we consider general
  bidding rules that may output multiple bids. 
}

\item 
{\bf Inclusion rule} (executed by the miner):
takes a 
bid vector $\bfb = (b_1,\dots,b_t) \in \R^t$ as input, 
and outputs 
a subset of $S \subseteq \bfb$ indicating which bids to be included in the block.
When the block size $k$ is \emph{finite}, 
it must be that $|S| \leq k$.
\item 
{\bf Confirmation rule} (executed by the blockchain):
takes as input a block $\bfB$, i.e., a vector of included bids, 
and outputs a subset of $\bfB$ to be confirmed.
In general, 
not all transactions included in the block must be confirmed, 
and only the confirmed transactions are executed.\footnote{Roughgarden \cite{roughgardeneip1559-ec} assumes that all included transactions are confirmed. However, Chung and Shi \cite{foundation-tfm} show that allowing unconfirmed transactions in a block enlarges the design space. For example, some mechanisms require a block  to contain some unconfirmed transactions (see Section 7 in \cite{foundation-tfm}).}
\item 
{\bf Payment rule} (executed by the blockchain):
given a block $\bfB$ as input, outputs the payment of each
bid.  
We require {\it individual rationality}, i.e., 
for any vector of true values $\bfv = (v_1, \ldots, v_t)$,  
if all users execute the honest 
bidding rule, and we then execute the honest inclusion, 
confirmation, and payment rules, then 
the following holds with probability $1$:\footnote{We can also relax
  the requirement such that individual rationality holds in
  expectation. Both the impossibility results
  (\cref{sec:deterministic,sec:randomglobalscp,sec:oca-impossibility})
  and the revelation principle result (\cref{sec:revelation}) continue
  to hold.}
for every user $i \in [t]$, if its primary 
bid is not confirmed, then its total payment is zero;
else if its primary bid is confirmed,
then its total payment is at most its true value $v_i$. 
Here, a user's total payment is the sum of the payment 
of all bids it has submitted.
\item 
{\bf Miner revenue rule} (executed by the blockchain):
given a block $\bfB$ as input, 
outputs how much revenue the miner gets. 
We require {\it budget feasibility}, i.e.,  
the miner's revenue 
is at most the total payment collected from all confirmed bids.
When the miner revenue is strictly smaller than the total payment, 
the difference between the total payment and the miner revenue is 
said to be \emph{burnt}, i.e.~eliminated from circulation.
\end{itemize}

We say a TFM is \emph{trivial} if the confirmation probability of all transactions is zero for any bid vector assuming the miner honestly follows the inclusion rule;
otherwise, it is called \emph{non-trivial}.

A strategic miner or miner-user coalition
may deviate from the honest inclusion rule. 
On the other hand, 
since the confirmation, payment, and miner revenue rules 
are executed by the blockchain, 
they are always 
implemented honestly.

We focus on mechanisms that are {\it weakly symmetric}, i.e.,
mechanisms that do not make use 
of the bidders' identities or other auxiliary 
information (e.g., timestamp, transaction metadata),
except for tie-breaking among equal bids.
\ignore{
A mechanism is said to be weakly symmetric, 
iff auxiliary information (e.g.,  
user identity, timestamp) 
is only used to break ties between multiple bids  
of the same amount.
}
More formally, we define weak symmetry as below.

\begin{definition}[Weak symmetry]
\label[definition]{def:symmetry}
A mechanism is called \emph{weakly symmetric} 
if the mechanism can always be equivalently described in the following manner: 
given a bid vector $\bfb$ where each bid may
carry some extra information such as identity or timestamp, 
the honest mechanism always sorts the vector $\bfb$ by the bid amount first. 
During the sorting step, if multiple bids have the same amount,
then arbitrary tie-breaking rules may be applied, and the tie-breaking 
can depend on extra information such as timestamp, identity, or random coins.
After this sorting step, the inclusion rule and the confirmation rules should depend only on the amount of the bids and their relative position in the sorted bid
vector.
\end{definition}

The above weak symmetry definition  
is equivalent to the following: 
for two bid vectors $\bfb$ and $\bfb'$ of length $n$ that are equivalent 
up to a permutation, 
then the distribution of the  
mechanism's outcomes on $\bfb$ and $\bfb'$ are
identically distributed --- assuming that the outcome   
is encoded in the form of 
$\left(\mu, {\sf sorted}(\{(v_i, x_i, p_i)\}_{i \in [n]})\right)$
sorted in descending order of $v_i$, 
where $\mu$ encodes the miner's revenue, each tuple $(v_i, x_i, p_i)$
indicates that some user with bid $v_i$ 
had the confirmation outcome $x_i \in \{0, 1\}$, and paid $p_i$.

If we require that the tie-breaking cannot depend on the extra information, 
we can define \emph{strong symmetry} as follows:
for any bid vector $\bfb = (b_1, \dots, b_t)$, and for any $i \neq j$ such that $b_i = b_j$, 
the random variables $(x_i, p_i)$ and $(x_j, p_j)$ must be identically distributed.
Throughout the paper, we focus on weakly symmetric mechanisms.
Since a strongly symmetric mechanism must also be weakly symmetric, 
our impossibility results also rule out the existence of strongly symmetric mechanisms.

\paragraph{Strategy space}
A strategic user can deviate
from the honest bidding rule and post an arbitrary bid vector 
with zero to multiple bids.
Without loss of generality, we may assume that in the strategic bid
vector, at most one bid can correspond to the user's actual transaction
which has a non-zero true value; all other bids must
be fake bids with zero true value.
A strategic miner 
can deviate from the  honest inclusion rule, and instead
create an arbitrary block (subject to the block size limit)
that includes any subset of the bid vector as well
as any number of fake bids
that it chooses to inject. 
A strategic miner-user coalition can adopt 
a combination of the above strategies. 

\paragraph{Utility and social welfare}

For a user with true value $v$, let $x \in \{0, 1\}$ be the indicator of
whether its primary bid is confirmed or not,  
let $p$ denote its total payment, then the user's utility
is $x \cdot v - p$.
The miner's utility is simply its revenue. 
The social welfare is defined to be the sum of the utilities of all
users and the miner (i.e., the total value of the confirmed
transactions, less any burned payments).

Notice that we allow the miner revenue to be smaller than the sum of users' payment, since the coins can be burnt.
When calculating the social welfare, the payments among the users and the miner are canceled out, so the social welfare is independent of the payment;
however, the amount of burnt coins decreases the social welfare.
For example, suppose there is only one user, and let $p$ be the user's payment and $q$ be the amount of burnt coins.
In this case, the user's utility is $x \cdot v - p$, the miner revenue is $p - q$, and the social welfare is $(x \cdot v - p) + (p - q) = x\cdot v - q$.
\Hao{Add a comment on only burning decreases the social welfare, but payment does not.}

\subsection{Incentive Compatibility Notions}
\label{sec:IC}

\begin{definition}[User incentive compatible (UIC)]
\label[definition]{def:UIC}
A TFM is said to be \emph{user incentive compatible (UIC)},
iff the following holds: 
for any user $i$, 
for any bid vector $\bfb_{-i}$ that corresponds to the 
bids posted by all other users, 
user $i$'s 
expected utility is always maximized 
when it follows the honest bidding rule (assuming that the miner executes
the inclusion rule honestly).
Notice that for a \truthful{} mechanism, the bidding rule is just truth-telling, so following the honest bidding rule simply means submitting the value truthfully;
for a non-\truthful{} mechanism, a bidding rule may ask a user to submit a bid other than its true value. 
\end{definition}

\begin{definition}[Miner incentive compatible (MIC)]
\label[definition]{def:MIC}
A TFM is said to be \emph{miner incentive compatible (MIC)}, 
iff given any bid vector $\bfb$, 
the miner's expected utility is maximized when 
the miner does not inject any fake bid and 
creates a block indicated by the honest inclusion rule.
\end{definition}

\begin{definition}[$c$-side-contract-proof ($c$-SCP)]
\label[definition]{def:cSCP}
A TFM is said to be 
$c$-side-contract-proof ($c$-SCP), 
iff for any coalition $C$ consisting of the miner and between $1$ and at most $c$ users,
and for any bid vector $\bfb_{-C}$ that corresponds to the bids posted by all users not in $C$, 
the coalition's joint utility is always maximized when 
it adopts the honest strategy, i.e., all users
in the coalition bid 
according to the honest 
bidding rule, 
and the miner follows the honest inclusion rule.
\end{definition}

\begin{definition}[Global side-contract-proof (global SCP)]
\label[definition]{def:globalSCP}
A TFM is said to be \emph{global side-contract-proof (global SCP)}, 
iff given any vector of true values $\bfv$, 
the expected social welfare is maximized when all the users bid 
according to the honest 
bidding rule, 
and the miner follows the honest inclusion rule,
where the maximization is taken over all the coordinated strategies that the coalition consisting of the miner and all users can adopt.
\end{definition}

\begin{definition}[OCA-proof]
\label[definition]{def:OCAproof}
A 
TFM is said to be \emph{OCA-proof} 
iff there exists an individually rational bidding 
strategy
$\sigma: \mathbb{R}_{\geq 0} 
\rightarrow \mathbb{R}$ such that the expected social welfare 
is maximized when all the users bid according to $\sigma$,
and the miner creates a block indicated by the inclusion rule,
where the maximization is taken over all the coordinated strategies that the coalition consisting of the miner and all users can adopt.
\end{definition}

In the definitions above, the expectation is taken over the randomness of the TFM.
More explicitly, in \cref{def:UIC}, the expectation is taken over the randomness of the inclusion/confirmation/payment rules;
in \cref{def:MIC,def:cSCP,def:globalSCP,def:OCAproof}, the expectation is taken over the randomness of the inclusion/confirmation/\\payment/miner revenue rules.

Note that in the OCA-proofness definition, 
$\sigma$ is required to output a single real-valued bid.
A canonical example of $\sigma$ is scaling; that is, $\sigma(v) = \gamma v$ for some $\gamma \in [0,1]$ (cf., Corollary 5.12 and 5.14 in \cite{roughgardeneip1559-ec}).

A detailed comparison between $c$-SCP, global SCP, and OCA-proofness
is given in \cref{sec:comparison}.


\section{Preliminary: Myerson's Lemma}


If a TFM satisfies UIC, then for any user, its confirmation probability and the payment must satisfy the famous Myerson's lemma~\cite{myerson}.
It is an important ingredient for proving our impossibilities, and we review it in this section. 
Formally, given a bid vector $\bfb = (b_1,\dots ,b_t)$, let $x_i(\bfb) \in [0,1]$ denote the probability that user $i$'s bid is confirmed,
and let $p_i(\bfb)$ denote user $i$'s expected payment assuming the miner implements the mechanism honestly.
For any vector $\bfb$, let $\bfb_{-i}$ denote a vector when we remove $i$-th coordinate in $\bfb$.
With these notations, Myerson's lemma can be stated as follows.

\begin{lemma}[Myerson's lemma~\cite{myerson}]
\label{lemma:myerson}
For any TFM that satisfies UIC if and only if
\begin{itemize}
\item 
For any user $i$, bid vector $\bfb_{-i}$, and $b_i,b'_i$ such that $b_i' > b_i$, it must be $x_i(\bfb_{-i}, b'_i) \geq x_i(\bfb_{-i}, b_i)$.
\item 
For any user $i$, bid vector $\bfb_{-i}$ from other users, and bid $b_i$ from user $i$, user $i$'s expected payment must be \[
p_i(\bfb_{-i}, b_i) = b_i \cdot x_i(\bfb_{-i}, b_i) - \int_0^{b_i} x_i(\bfb_{-i}, t) d t,
\]
with respect to the normalization condition: $p_i(\bfb_{-i}, 0) = 0$, i.e.~user $i$'s payment must be zero when $b_i = 0$.
\end{itemize}
\end{lemma}

When the mechanism is deterministic, the confirmation probability $x_i$ is either $0$ or $1$.
In this case, user $i$'s payment can be simplified as 
\[
p_i(\bfb_{-i}, b_i) = 
	\left\{\begin{matrix}
	\inf \{z \in [0, b_i]: x_i(\bfb_{-i}, z) = 1\},& \text{ if $x_i(\bfb_{-i}, b_i) = 1$}; \\ 
	0,& \text{ if $x_i(\bfb_{-i}, b_i) = 0$.}
	\end{matrix}\right.
\]
Conceptually, user $i$ must pay the minimal price which makes its bid
confirmed.

\section{Warmup: Impossibility of UIC + MIC + Global SCP for Deterministic Mechanisms}
\label{sec:deterministic}
As a warmup, we first show a finite-block impossibility
for UIC + MIC + global SCP for
{\it deterministic} mechanisms.  
Recall that a TFM is said to be trivial if everyone's confirmation probability is zero for any bid vector assuming the miner follows the inclusion rule.
In this case, everyone's utility is always zero in an honest execution.
We will show that no non-trivial mechanism can satisfy all three properties simultaneously.
Later in \cref{sec:randomglobalscp}, we extend 
the impossibility to randomized mechanisms.
Due to the revelation principle that
we prove in \cref{sec:revelation}, 
if we can prove the impossibility
for \truthful{} mechanisms, the impossibility
immediately extends to non-\truthful{} mechanisms as well.
Therefore, in this section, we shall assume
\truthful{} mechanisms.

\begin{lemma}
For any global SCP mechanism, the confirmed bids must correspond
to the highest bids. 
\label[lemma]{lem:ordered}
\end{lemma}
\begin{proof}
Suppose in some scenario, 
Alice bids her true value $b$ and Bob bids his true value $b' < b$;
however, Bob's bid is confirmed and Alice's is not.
Now, we can have Alice and Bob swap their bids. 
The miner 
creates the same block as before in which 
the position originally corresponding to 
Bob now has Alice's bid of $b'$.
Since the mechanism is weakly symmetric (\cref{def:symmetry}), Alice's bid is confirmed.
This way, the social welfare increases by $b - b'$ in comparison 
with the honest case, and this violates global SCP.
\end{proof}

\begin{lemma}
For any global SCP mechanism, 
the amount of burnt coins depends only on the number of confirmed bids.
\label[lemma]{lem:burn}
\end{lemma}
\begin{proof}
Suppose in two different scenarios, 
when everyone acts honestly, 
the blocks made are ${\bf B}$ and ${\bf B}'$ respectively, 
the confirmed bids are $\bfb \subseteq {\bf B}$
and $\bfb' \subseteq {\bf B}'$ respectively where $\bfb$ and $\bfb'$ are of the same
length,
and the burnt amount in the two scenarios are $q$ and $q'$ respectively,
where $q < q'$.
Now, suppose we are actually in the second scenario.
A global coalition can adopt the following strategy:  create a block identical to ${\bf B}$
in which the confirmed bids correspond to the users 
with the highest true values 
and the rest can be fake bids.
Observe 
that the social welfare is the sum of the true values of all confirmed
bids (where fake bids have a true value of $0$) minus 
the total coins burnt.
Therefore, the above strategy achieves strictly higher social welfare
than the honest case.
\end{proof}

\begin{theorem}
No non-trivial deterministic TFM can simultaneously satisfy UIC, MIC, and global SCP
when the block size is finite.
\label{thm:impossible}
\end{theorem}
\begin{proof}
Since the block size is finite, there must exist a scenario
that maximizes the number of confirmed bids.
Suppose in this scenario, the block created is ${\bf B}$
and the confirmed bids contained in ${\bf B}$ are 
$\bfb = (b_1, \ldots, b_m)$ sorted from high to low.
Since the mechanism is non-trivial, we have $m \geq 1$.
Given \Cref{lem:burn}, we can use $q_i$ to denote the amount of burnt 
coins when exactly $i$ bids are confirmed.

First, we show that in any scenario with at least $m$
users whose true values are sufficiently large w.r.t. $q_0, q_1, \ldots, q_m$, 
if everyone acts honestly, 
the mechanism must confirm $m$ bids --- by
\Cref{lem:ordered}, they must also be the $m$ highest bids.
Suppose this is not the case, say, 
$m' < m$ bidders are confirmed 
--- by \Cref{lem:ordered}, they 
must be the $m'$ highest bids.
Now, the global coalition 
can create the block ${\bf B}$
where the $m$ confirmed bids correspond to
the users with the highest true values, 
and the rest can just be fake bids.
In comparison with the honest case, the increase
in social welfare
is $\sum_{i = m'+1}^m \left(v_i - (q_i - q_{i-1})\right)$
which is positive as long as users' true values $v_i$'s are sufficiently large.
This violates global SCP.

Imagine a scenario with at least $m+1$ users
whose true values are sufficiently high, and $v_1 \geq \cdots \geq v_m > v_{m+1}$.
By \cref{lem:ordered}, it must be $v_1,\dots,v_m$ are confirmed.
Because the mechanism is UIC, by Myerson's lemma, all $m$ confirmed 
bids must pay $v_{m+1}$.
However, such a mechanism violates MIC, because the miner can 
pretend that there is a fake bid $v' \in (v_{m+1}, v_m)$.
Because the burnt amount is fixed to $q_m$,  
the miner will get more revenue  
if every confirmed bid pays
$v'$ instead of $v_{m+1}$.
\end{proof}

\section{Impossibility of UIC + MIC + Global SCP 
for Randomized Mechanisms}
\label{sec:randomglobalscp}
In this section, we extend the finite-block impossibility of UIC + MIC
+ global SCP
to even randomized mechanisms.
Recall that a TFM consists of five rules as defined in \cref{sec:tfm-def}, and a randomized TFM may use randomness in any of the five rules.
Since the confirmation, the payment, and the miner revenue rules are executed by the blockchain,
the strategic players can only bias the randomness in and deviate from the bidding rule and the inclusion rule.
Again, due to the revelation principle
proven in \cref{sec:revelation},  
it suffices to consider \truthful{} mechanisms.

\subsection{Proof Roadmap}
A new challenge in the randomized setting is that whether a bid is confirmed becomes probabilistic,
and the arguments in \cref{sec:deterministic} no longer hold.
Here we give a brief roadmap of our new proof.
The key idea of the proof
is captured in
 \Cref{lem:zero-user-utility}, which says that if there are $i$ equal bids 
$(\underbrace{b, \ldots, b}_i)$, followed 
by other bids $(b_{i+1}, \ldots, b_{k+1})$ 
that are strictly smaller than $b$, 
then $b_{i+1}, \ldots, b_{k+1}$ 
must be  
unconfirmed
with probability $1$, and moreover, 
all the $b$-bids must pay $b$ when confirmed, i.e.,
the total user utility is $0$ --- throughout, we use $k$ to denote the block size.
Since \Cref{lem:zero-user-utility}
imposes a very strong constraint on the TFM, using
it to lead to the 
final impossibility is not too hard (see \cref{thm:random-scp-impossible}).

We therefore 
provide some intuition for the proof of \Cref{lem:zero-user-utility}.
Specifically, we use an inductive proof.
The base case is 
the scenario $(\underbrace{b, \ldots, b}_k, b_{k+1})$ where $b_{k+1} < b$. 
Using global SCP and the fact that the block
size is only $k$, we show that the bid $b_{k+1}$ must be
unconfirmed with probability $1$.
Then, we show that the $k$ bids at $b$ 
must pay $b$ if they are confirmed. 
To show this, we first argue that 
in a TFM that is MIC and global SCP, 
the total user utility is equivalent if 
$b_{k+1}$ is anything below $b$ (see \cref{lem:freeknob}).
We then consider a scenario where the $(k+1)$-th bid
$b_{k+1} < b$ is arbitrarily close to $b$, 
and consider lowering
one of the first $k$ bids (say, the first bid) to   
$b_{k+1} -\delta$ --- in this case, we argue that 
the first bid will become unconfirmed using global SCP.
Finally, by Myerson's lemma, we conclude
that in the 
original scenario
$(\underbrace{b, \ldots, b}_k, b_{k+1} < b)$, 
the first bid must pay $b$
when confirmed.

Once we prove the base case, we 
proceed by induction for $i = k-1, k-2, \ldots, 1$.
The induction step is similar in structure
to the base case, except that we first need to use
the induction hypothesis that in the scenario 
$(\underbrace{b, \ldots, b}_{i}, b_{i+1} < b, \ldots)$, 
all the $b$ bids pay $b$ when confirmed, 
to conclude that in the scenario  
$(\underbrace{b, \ldots, b}_{i-1}, b_{i} < b, \ldots)$, 
$b_i$ is unconfirmed with probability $1$ --- this follows
due to Myerson's lemma.
From this point on, the rest of the proof for the induction step
is similar to the base case. 

\ignore{
First, observe that when the number of users is more than the block size, at least one user's bid must be unconfirmed, and we show that it must be the one with the lowest true value (\cref{lem:top-k}).
If a strategic player knows that injecting a fake bid $b$ must be unconfirmed for certain, this fake bid is cost-free.
Thus, to prevent the strategic player from gaining profit by injecting fake bids, its utility must be independent of any bid with zero confirmation probability (\cref{lem:freeknob}).
Next, intuitively, Myerson's lemma requires that the payment to be the minimum amount so that the confirmation probability is non-zero.
By taking the unconfirmed bid $b$ to be arbitrarily close to other users' true values, we have that other users' utilities must be zero (\cref{lem:zero-user-utility}),
and surprisingly, we show zero user utility also holds even if there is only one user (\cref{lem:only-user})!
Finally, as long as the mechanism is non-trivial, there must exist a scenario where the only user has a positive utility, which leads to a contradiction (\cref{thm:random-scp-impossible}).
}

\Hao{Added preamble.}

\ignore{
We use the notation
${\bf x}(\bfb)$  
to denote the vector of confirmation probabilities
of all transactions under the bid vector $\bfb$.
Similarly, we use $x_i(\bfb)$ to denote the confirmation
probability of the $i$-th bid.
}

\subsection{Formal Proofs}
In the rest of this section, we present the formal proofs.

\begin{lemma}
Suppose the mechanism satisfies global SCP.
Then, for any bid vector $\bfb = (b_1,\dots,b_t)$, 
if $b_i > b_j$, it must be that 
$x_i(\bfb) \geq x_j(\bfb)$
where 
$x_i(\bfb)$ and $x_j(\bfb)$
denote user $i$ and $j$'s confirmation probabilities, respectively.
\label[lemma]{lem:random-ordered}
\end{lemma}
\begin{proof}
For the sake of contradiction, suppose there exists two integers $i,j$ such that $b_i > b_j$ and $x_i(\bfb) < x_j(\bfb)$.
The global coalition can swap user $i$'s and user $j$'s bids.
Concretely, user $i$ should bid $b_j$ and user $j$ should bid $b_i$.
Since the bid vector is the same as before (though $b_i$ represents user $j$'s bid, and $b_j$ represents user $i$'s bid),
the expected burning is the same as before.
However, since the mechanism is weakly symmetric, the social welfare increases by $\left(b_i \cdot x_j(\bfb) + b_j \cdot x_i(\bfb)\right) -\left(b_i \cdot x_i(\bfb) + b_j \cdot x_j(\bfb)\right) > 0$, which violates global SCP.
\end{proof}

\begin{lemma}
Suppose the mechanism satisfies global SCP and let the block size be $k$.
Consider any  
 bid vector satisfying $b_1 \geq b_2 \geq \ldots \geq b_k > b_{k+1}$.
Then, 
it must be that $b_{k+1}$
has 0 probability of confirmation.
\label[lemma]{lem:top-k}
\end{lemma}
\begin{proof}
Suppose for the sake of contradiction
that $b_{k+1}$
has non-zero probability of confirmation.
This means it must be included in the block
with some non-zero probability $p$, and 
conditioned on being included, it has a non-zero probability $p'$ of
confirmation.
A global coalition whose true values
are $b_1, \ldots, b_{k+1}$ can play the following strategy.
First, run the honest mechanism in its head including flipping any coins required by the inclusion rule
assuming everyone bids honestly, 
and let ${\bf B}$ be the resulting block.
If $b_{k+1} \notin {\bf B}$, just have everyone bid truthfully
and build the block ${\bf B}$.
Else, if $b_{k+1} \in {\bf B}$, then let $b_i > b_{k+1}$
be one bid that is left out of the block --- such a bid
must exist because the block size is only $k$.
The coalition will have the $i$-th user 
actually bid $b_k$ and build the same block ${\bf B}$
where the $b_{k+1}$
actually corresponds to user $i$'s strategic bid.
With this strategy, 
the coalition's expected gain in utility is $p \cdot p' \cdot (b_i - b_{k+1}) > 0$,
which violates global SCP.
\end{proof}

\ignore{
\begin{lemma}
Suppose the mechanism is global SCP and the block size is $k$.
Let $\bfb = (b_1,\dots,b_t)$ be any vector in the descending order such that $t > k$.
For any $i > k$, if the bid vector  is $\bfb$ and $b_i < b_k$, then it must be $x_i(\bfb) = 0$.
In other words, $b_i$ is either not be included, or it can never be confirmed even when it is included.
\label[lemma]{lem:top-k}
\end{lemma}
\begin{proof}
For the sake of contradiction, suppose there exists a scenario where all users' valuations are $\bfb = (b_1,\dots,b_t)$, and exist an integer $i > k$ such that $b_i < b_k$ and $x_i(\bfb) > 0$.
Since $x_i(\bfb) > 0$, when the miner follows the inclusion rule, $b_i$ is in the created block with non-zero probability.
However, since the block size is $k$ and $i > k$, when $b_i$ is in the block, there must exist another user $j$ whose true value $b_j$ is larger than $b_i$ but $b_j$ is not included.
Thus, whenever $b_i$ is in the created block indicated by the honest inclusion rule, the global coalition can swap user $i$ and $j$'s bids, so that user $j$ bids $b_i$ instead and the miner places $b_i$ as before while it is submitted by user $j$.
Since $x_i(\bfb) > 0$, user $j$'s bid will be confirmed with non-zero probability.
Since the created block is the same as the honest case except that $b_i$ is now submitted by user $j$ who has higher true value than user $i$, the joint utility of the global coalition increases, which violates global SCP.
\end{proof}
}


\begin{lemma}
Suppose the mechanism satisfies MIC and global SCP.
Let ${\bf a}$ be an arbitrary 
vector of positive length, 
and let $\bfb, \bfb'$
arbitrary vectors --- $\bfb$ and $\bfb'$ may or may not be of the same length,
and their lengths are allowed to be $0$.
Consider two scenarios  
with bid vectors $({\bf a}, \bfb)$
and $({\bf a}, \bfb')$ respectively. 
Suppose that $\bfb$ and $\bfb'$ have 0 confirmation probability
in each of the two scenarios, respectively.
Then, both scenarios enjoy the same expected
miner utility, total social welfare, and total user utility.
\label[lemma]{lem:freeknob}
\end{lemma}
\begin{proof}
We first prove that expected miner utility is the same
in both scenarios.
Suppose this is not true, and without loss
of generality, suppose
expected miner utility is higher in scenario 1.
Then, the miner can ignore the bids $\bfb$,
inject the fake bids $\bfb'$, pretend that the bid vector
is $(\bfa, \bfb')$, and run the honest mechanism.
Since the confirmation probability of $\bfb'$ is $0$, the
miner need not pay any cost for the fake bids.   
Therefore, the miner gets higher expected utility 
by taking the above strategy which violates MIC.

The proof of total social welfare is similar. 
Suppose, without loss
of generality, that the expected total social welfare in scenario 1 is higher.
Then, the global coalition can inject
fake bids $\bfb'$ and pretend that the bid vector is $(\bfa, \bfb')$, 
thus allowing it to increase its expected social welfare. This violates global
SCP.

The equivalence in total user utility follows directly from the above, since 
total user utility is the difference between the social welfare and 
the miner utility. 
\end{proof}

\begin{lemma}
Suppose the mechanism satisfies UIC, MIC, and global SCP, 
and the block size is $k$.
Let $\bfb = (b_1,\dots,b_{k+1})$ be any vector where $b_1 = b_2 = \cdots = b_i > b_{i+1} > \cdots > b_{k+1}$. 
Then, 
under the bid vector $\bfb$, 
the following hold: 1) 
users $i+1, \ldots, k+1$ have $0$ confirmation probability; 
and 2) the total user utility is $0$ assuming 
every one is bidding their true value.
\label[lemma]{lem:zero-user-utility}
\end{lemma}
\begin{proof}
We prove both statements by induction.
Henceforth, given a bid vector $\bfb = (b_1, \ldots, b_{k+1})$,
we use $x_i(\bfb)$
to denote the confirmation probability of the $i$-th user
under $\bfb$.

\paragraph{Base case.}
The base case is when $i = k$;
that is, there are $k+1$ users with valuations $\bfb = (b_1,\dots,b_{k+1})$ where $b_1 = \cdots = b_k > b_{k+1}$.
By \cref{lem:top-k}, we have $x_{k+1}(\bfb) = 0$.
Next, we show that the joint utility of all users must be zero.
By \cref{lem:freeknob} and \cref{lem:top-k}, the total user utility
under $\bfb = (\underbrace{b, \ldots, b}_k, b_{k+1})$ 
is the same as 
under $\bfb' = (\underbrace{b, b, \ldots,  b}_k, b - \delta)$ for 
any arbitrarily small $\delta > 0$ --- 
in both cases, the last user must have 0 confirmation
probability.
Now, consider the 
scenario 
$\bfb' = (\underbrace{b, b, \ldots,  b}_k, b - \delta)$. 
If one of the $b$-bidders denoted $i \in [k]$ lowered 
its bid to less than $b-\delta$ (and everyone else's bids
stay the same), its confirmation 
probability would become $0$ 
by \cref{lem:top-k}.
Since the mechanism is UIC, by Myerson's lemma, 
under $(b, b, \ldots,  b, b - \delta)$, 
user $i$'s expected payment must be at least  
$(b-\delta) \cdot x_i(\bfb')$. Therefore,
user $i$'s 
expected utility 
is at most $\delta \cdot x_i(\bfb')$.
Therefore, the expected total user utility under $\bfb'$ 
is at most $k \cdot \delta \cdot x_i(\bfb') \leq k \cdot \delta$.
This means that under $\bfb = (b_1, \ldots, b_{k+1})$, 
the expected total user utility is at most 
$k \cdot \delta$ for any arbitrarily small $\delta > 0$, i.e.,
the expected total user utility is at most $0$.
Finally, by the individual rationality of the payment rule, each user's utility is non-negative when it bids truthfully.
Therefore, the expected total user utility is $0$.

\paragraph{Inductive step.}
Fix any $i \in \{2,\dots,k\}$, suppose the lemma statement holds.
We will show that the lemma also holds 
for any vector $\bfb = (b_1,\dots,b_{k+1})$ where $b_1 = b_2 = \cdots = b_{i-1} > b_i > b_{i+1} > \cdots > b_{k+1}$. 

We first show that all users $i, \ldots, k+1$ are unconfirmed. 
Consider the 
vector $\underbrace{b, b, \ldots, b}_i, b_{i+1}, \ldots, b_{k+1}$. 
By our induction hypothesis, 
$i$ users can have non-zero  
confirmation probability and moreover, the total user utility is $0$.
Because of individual rationality, each user's utility must be non-negative,
so every user's utility is also $0$.
In this case, whenever a bid is confirmed, the payment is equal to the bid (and hence the true value).
By UIC and Myerson's lemma, 
it means that if the $i$-th user lowers its bid to $b_{i} < b$ (which
becomes the scenario $\bfb$ we care about), then
user $i$'s confirmation probability becomes $0$.
By 
\cref{lem:random-ordered}, users $i+1, \ldots, k+1$ must
have 0 confirmation probability under $\bfb$ too.

\ignore{Let $\bfb' = (\bfb_{-i}, b_1)$; that is, $\bfb'$ is the same as $\bfb$ except that the $i$-th coordinate is replaced with $b_1 = b_2 = \cdots = b_{i-1}$.
Consider the scenario where $k+1$ users' valuation is $\bfb'$.
Since the first $i$ coordinate are the same, by the inductive assumption, the joint utility of all users is zero when they bid $\bfb'$.
In particular, user $i$'s utility is zero.
By the Myerson's lemma, if a user with the true value $v$ gets zero utility when it bids truthfully, the allocation probability must be zero when it underbids.
Thus, we have $x_i(\bfb) = 0$.
By \cref{lem:random-ordered}, we also have $x_{i+1}(\bfb) = \cdots = x_{k+1}(\bfb) = 0$.
}

Next, we show that the joint utility of all users is zero 
under $\bfb$ and assuming everyone bids truthfully.
By \cref{lem:freeknob}, 
the total user utility under $\bfb$ is the same
as under 
the 
scenario 
$\bfb' = (\underbrace{b, b, \ldots, b}_{i-1}, b-\delta, b-2\delta, \ldots, 
b-(k+2-i)\delta)$ for an arbitrarily small $\delta > 0$ --- 
since in both scenarios,
users $i, \ldots, k+1$ have 0 confirmation probability.
Under $\bfb'$, if any $b$-bidders denoted $j$ lowers
its bid to less than $b-(k+2-i)\delta$, its confirmation
probability  
becomes $0$ by \cref{lem:top-k}.
By UIC and Myerson's lemma, 
 $j$'s expected payment under $\bfb'$ is at least
$(b-(k+2-i)\delta) \cdot x_j(\bfb')$.
This means that the total user utility under
$\bfb'$ is at most 
$(k+2-i)\delta \cdot x_j(\bfb') \cdot (i-1) \leq k^2 \delta$.
Therefore, under $\bfb$, the total user utility
is at most $k^2\delta$ for any arbitrarily small $\delta > 0$.
This means that the total user utility under $\bfb$ is at most $0$.
Finally, by the individual rationality of the payment rule, each user's utility is non-negative when it bids truthfully.
Therefore, the expected total user utility is $0$.

\end{proof}


\begin{lemma}
Suppose the mechanism satisfies UIC, MIC, and global SCP, 
and the block size is $k$.
Let $a$ be any positive real number.
Consider a scenario with only one bid $a$. Then, the only user's utility
is zero assuming it bids its true value.
\label[lemma]{lem:only-user}
\end{lemma}
\begin{proof}
Fix $a_1 > 0$.
Consider any vector $\bfa = (a_1, a_2, \dots, a_{k+1})$ such that $a_1 > a_2 > \cdots > a_{k+1}$.
By \cref{lem:zero-user-utility}, 
under $\bfa$, all users except the 
first user must have 0 confirmation probability, and moreover,
the first user has utility $0$.
By \cref{lem:freeknob}, 
the total expected user utility under 
$\bfa$ and  under a single bid $a_1$
are the same.
Therefore, 
under a single bid $a_1$, the user's utility is also $0$.
\ignore{
no matter we only have one user with the valuation $a_1$ or we have $k+1$ users with valuations $\bfa$, the joint utility of all users in both scenarios are the same.
Thus, if there is only one user in the bid vector with the valuation $a_1$, its utility is zero when it bids truthfully.
}
\end{proof}

\begin{theorem}
No non-trivial, possibly  randomized TFM can simultaneously satisfy UIC, MIC, and global SCP
when the block size is finite.
\label{thm:random-scp-impossible}
\end{theorem}
\begin{proof}
We will show that under any sufficiently large $a$, 
the confirmation probability under a single bid $a$ is non-zero. 
If we can show this, 
then we can show a contradiction to UIC.
Specifically, consider $b > a$ and both sufficiently large.
By \cref{lem:only-user}, if there is only one user with true value $b$,
its utility is zero when it bids truthfully.
However, the user can underbid $a$.
Since the confirmation probability is non-zero and the payment is at most $a$, the user enjoys positive utility, which violates UIC.

We now show that the confirmation probability under a single bid $a$ 
for a sufficiently large $a$ is non-zero assuming the miner honestly follows the inclusion rule.
Suppose this is not true for the sake of a contradiction;
that is, when the bid vector is just $a$ and the miner is honest, no one has positive confirmation probability, so the social welfare is zero.
Since the mechanism is non-trivial, there must exist 
a scenario $\bfb = (b_1,\dots,b_t)$ such that 
some user $i$ has positive confirmation probability $p(\bfb)$.
Let $a > \frac{1}{p(\bfb)} \cdot \sum_{j=1}^t b_j$.
Then, imagine that there is only one user with true value $a$.
The global coalition can replace the user's primary bid with $b_i$, inject the fake bids  
$\bfb_{-i}$, and pretend that the bid vector is $\bfb$.  
This way, the user would get positive confirmation probability $p(\bfb)$.
Since its payment is at most $\sum_{j=1}^t b_j$, 
the user's expected utility (and thus the social welfare) is positive.
This violates global SCP.
\end{proof}

\section{Feasibility and Impossibility of UIC + MIC + OCA-Proof}
We can generalize the proof in \Cref{sec:randomglobalscp}, and 
rule out UIC, MIC, and OCA-proof (rather than global SCP) for 
\truthful{} mechanisms.
Recall that for a \truthful{} mechanism, the difference between  
OCA-proof and global SCP is that global SCP insists that 
the optimal strategy of the global coalition is the truthful strategy, 
whereas OCA-proofness allows it to be 
some other strategy in which each user acts independently and bids  
the outcome of some function $\sigma(\cdot)$.

Interestingly, if we allow the bidding rule to be not truth-telling, i.e.~considering non-\truthful{} mechanisms,
we can have a mechanism that satisfies UIC, MIC, and OCA-proof.
We present the feasibility for non-\truthful{} mechanisms in \cref{sec:non-direct-oca-counterexample},
and we prove the impossibility of UIC + MIC + OCA-proof for \truthful{} mechanisms in \cref{sec:oca-impossibility}.
Notice that because of the feasibility in \cref{sec:non-direct-oca-counterexample}, we must require the bidding rule to be truth-telling to reach an impossibility in \cref{sec:oca-impossibility}.

\subsection{A Non-\Truthful{} Mechanism with UIC + MIC + OCA-Proof}
\label{sec:non-direct-oca-counterexample}

The rationale of the design is to signal to the mechanism when everyone is adopting
the globally optimal strategy $\sigma$ (as opposed to the bidding rule used to establish UIC).
When the mechanism detects that everyone is behaving according
to $\sigma$, it adopts a different behavior
to optimize social welfare.
We use the range $[0, 1)$ to encode
the actual bid, and use the range $[1, \infty)$
for signalling.
While the resulting mechanism is somewhat contrived and not necessarily meaningful from a practical point of view, 
it clarifies which notions of collusion-resilience most accurately capture the intended modeling goals
and illustrates some technical challenges involved in the proof in \cref{sec:oca-impossibility}.
Consider the following TFM:
\begin{itemize}
\item 
{\bf Globally optimal strategy $\sigma(v)$}:
Given a true value $v$, output a bid $v+1$.
\item 
{\bf Bidding rule}: 
Given a true value $v$, output a bid $1/(v+2)$.
\item 
{\bf Inclusion rule}: 
Let $S$ be the set of all pending bids that are in $[0,1)$.
If $|S| > k$, then randomly select $k$ bids from $S$ to include.
If $1 \leq |S| \leq k$, then include all bids in $S$.
If $|S| = 0$, choose the top up to $k$ bids to include. 
\item 
{\bf Confirmation, payment, and miner revenue rules}:
All included bids are confirmed.
Each confirmed bid pays nothing, and the miner gets nothing.
\end{itemize}

Obviously, this mechanism is non-trivial.

\begin{claim}
The above mechanism satisfies UIC, MIC, and OCA-proofness.
\end{claim}
\begin{proof}
For UIC, notice that if a user follows the bidding rule, its bid is always in $[0,1)$.
If there is no bid in $[0,1)$ before a user submits its bid, then bidding $1/(v+2)$ always guarantees user's bid to be included and confirmed, where $v$ denote the true value.
If there is already some bids in $[0,1)$ before a user submits its bid, then bidding $1/(v+2)$ is a dominant strategy since it guarantees the user's bid is added to $S$, the set of all bids in $[0,1)$, which is the best a user can do.
Next, MIC holds since the miner revenue is always zero.
Finally, if all users follow the globally optimal strategy $\sigma$, everyone's bid is at least $1$.
The honest inclusion rule will include the top up to $k$ bids, which maximizes the social welfare.
Thus, OCA-proofness holds.
\end{proof}

\begin{remark}
	\label{remark:revelation-OCA}
	We can try to apply revelation principle, and bake the bidding rule into the mechanism so that the resulting mechanism is \truthful{}.
	For example, whenever seeing a bid $b$, the miner and the mechanism view it as $1/(b+2)$.
	The modified mechanism, however, does not satisfy OCA-proofness anymore when the number of users is larger than the block size,
	since the miner should choose $k$ users with the highest true values instead of the random selection as indicated by the inclusion rule.
	This is not a coincidence: in the next section, we show that it is impossible to have a non-trivial \truthful{} mechanism satisfying UIC, MIC, and OCA-proofness.
\end{remark}

\subsection{Impossibility of UIC + MIC + OCA-Proof for \Truthful{} Mechanisms}
\label{sec:oca-impossibility}
The structure of the proof in this section is similar to the proof in \cref{sec:randomglobalscp}.
In fact, \cref{lem:random-oca-ordered,lem:random-oca-topk,lem:random-oca-freeknob,lem:random-oca-0util,lem:random-only-user,thm:oca-random-impossibility} are the analogs of \cref{lem:random-ordered,lem:top-k,lem:freeknob,lem:zero-user-utility,lem:only-user,thm:random-scp-impossible}, respectively,
except that we need to work on the images of $\sigma$ in order to apply OCA-proofness.
Before diving into the proof, we make a few remarks.
\begin{itemize}
\item 
A key step in \cref{sec:randomglobalscp} is to prove a user with true value $v$ must have zero utility, and we prove it by making an unconfirmed bid arbitrarily close to $v$ (see \cref{lem:zero-user-utility}).
To extend the similar idea to OCA-proofness, we require $\sigma$ to be strictly increasing and continuous.
In fact, we show that being strictly increasing is a consequence of UIC and OCA-proofness (\cref{lemma:sigma-monotone}),
and being strictly increasing implies that there exists a point at which $\sigma$ is continuous (\cref{cor:continuity}).
\item 
As we have seen a feasibility for non-\truthful{} mechanisms in \cref{sec:non-direct-oca-counterexample}, we must require the bidding rule to be truth-telling to reach an impossibility.
Indeed, the proofs of \cref{lemma:sigma-monotone,lem:random-oca-0util,lem:random-only-user,thm:oca-random-impossibility} rely on Myerson's lemma, and Myerson's lemma requires the mechanism to be \truthful{}.
Notice that the mechanism in \cref{sec:non-direct-oca-counterexample} does not satisfy the requirements of Myerson's lemma.
If we apply the revelation principle to make the mechanism become \truthful{} so that Myerson's lemma could apply, it breaks OCA-proofness (see \cref{remark:revelation-OCA}).
\item 
Notice that OCA-proofness requires $\sigma$ to output a single real-valued bid.
In fact, if we allow $\sigma$ to output multiple bids, we can have a (somewhat contrived) mechanism that satisfies UIC, MIC, and OCA-proof (see \cref{sec:counterexample}).
In particular, \cref{lemma:sigma-monotone} requires $\sigma$ to only output a single real number.
\item 
The original definition of OCA-proofness in~\cite{roughgardeneip1559-ec}
also requires $\sigma$ to be individually rational, but this 
will not be needed for the impossibility proof.
\end{itemize}

Given any vector $\bfb = (b_1,\dots,b_t)$ of the length $t$, 
let $\sigma(\bfb)$ denote the element-wise application of $\sigma$, i.e.~$\sigma(\bfb) = \left(\sigma(b_1),\dots,\sigma(b_t)\right)$.

\begin{lemma}
Suppose the mechanism satisfies OCA-proofness with the globally optimal strategy $\sigma$.
Then, for any bid vector 
$\sigma(\bfb) = (\sigma(b_1), \dots, \sigma(b_t))$, 
if $b_i > b_j$, it must be that 
the confirmation probability of the bid $\sigma(b_i)$
is at least as high as that of $\sigma(b_j)$.
\label[lemma]{lem:random-oca-ordered}
\end{lemma}
\begin{proof}
Imagine that $t$ players each have true values $b_1, \ldots, b_t$,
and they all bid according to the globally optimal strategy $\sigma$.
Let $x_i, x_j$ denote the confirmation 
probabilities of $\sigma(b_i)$ and $\sigma(b_j)$ respectively.
Suppose for the sake of contradiction that $b_i > b_j$, and $x_j > x_i$.
The global coalition can engage in the following OCA: 
have users $i$ and $j$ swap their bids, that is, user $i$ who has
true value $b_i$ bids $\sigma(b_j)$, and user $j$ 
who has true value $b_j$ bids $\sigma(b_i)$.
The set of bids is the 
same as before, therefore, the expected burning amount does not change.
The expected social welfare thus 
increases by $b_i x_j + b_j x_i - (b_i x_i + b_j x_j) = (b_i - b_j) (x_j - x_i) > 0$.
This violates OCA-proofness and the fact that $\sigma$ is the social-welfare-maximizing-strategy
of the global coalition.
\end{proof}

\begin{lemma}
Suppose the mechanism satisfies OCA-proofness with the globally optimal strategy $\sigma$, and let the block size be $k$.
Given $b_1 \geq b_2 \geq \ldots \geq b_k > b_{k+1}$,
then, under the bid vector
$\sigma(b_1), \sigma(b_2), \ldots, \sigma(b_k), \sigma(b_{k+1})$, 
 it must be that 
$\sigma(b_{k+1})$ has 0 probability of confirmation.
\label[lemma]{lem:random-oca-topk}
\end{lemma}
\begin{proof}
Imagine that $k+1$ players each have true values $b_1, \ldots, b_{k+1}$,
and they all bid according to the globally optimal strategy $\sigma$.
It suffices to show the following: 
suppose $\sigma(b_{k+1})$
has a non-zero probability $p$ of being included
in the block, then conditioned on being included, its probability of confirmation
must be $0$. 
Suppose this is not true, and conditioned on being 
included, $\sigma(b_{k+1})$
has probability $p' > 0$ of being confirmed.
Then, the global coalition can adopt the following OCA:
run the honest inclusion rule in its head, including flipping any coins required by the inclusion rule, 
and let ${\bf B}$ denote the resulting block.
If $\sigma(b_{k+1})$ is not included in ${\bf B}$, 
simply output the block ${\bf B}$.
Else if $\sigma(b_{k+1})$ is included in ${\bf B}$, 
then let $\sigma(b_i)$ be some bid that is left out --- such a bid must
exist because the block size is only $k$. 
Now, have the $i$-th user bid $\sigma(b_{k+1})$
and make an actual block ${\bf B}$ where the bid $\sigma(b_{k+1})$ 
corresponds to user $i$'s bid. 
The expected 
increase in social welfare relative to everyone adopting $\sigma$
is $p\cdot p' \cdot (b_i - b_{k+1}) > 0$, which violates
OCA-proofness and the fact that $\sigma$ is the social-welfare-maximizing strategy.
\end{proof}

\begin{lemma}
Suppose a \truthful{} mechanism satisfies UIC and OCA-proofness with the globally optimal strategy $\sigma$, and let the block size be $k$.
Then, there exists a constant $c^*$ such that for any real numbers $z,z'$ satisfying $z > z' > c^*$, it must be $\sigma(z) > \sigma(z')$.
Conceptually, $\sigma$ must be strictly increasing for large inputs.
\label[lemma]{lemma:sigma-monotone}
\end{lemma}
\begin{proof}
Since the mechanism is non-trivial, there exists a bid vector $(b_1,\dots,b_t)$ such that the total confirmation probability is positive.
Let $b_i$ be the bid with the highest confirmation probability among $b_1,\dots,b_t$.
Let $x$ be the confirmation probability of $b_i$ and let $c^*$ be any real number larger than $\sum_{j=1}^t b_j / x$.
We will show that for any $z > z' > c^*$, 
$\sigma(z) > \sigma(z')$.

Consider a scenario where there are $k+1$ users all with the true value $z$.
Suppose they form a global coalition and follow the bidding strategy $\sigma$ so that the bid vector is $\bfs = (\underbrace{\sigma(z), \dots, \sigma(z)}_{k+1})$.
There are two possible cases.

First, there exists a user $j$ with positive confirmation probability assuming the miner honestly follows the inclusion rule and select from the bid vector $\bfs$.
Because the mechanism satisfies UIC, by Myerson's lemma, if user $j$ increases its bid, the confirmation probability must be non-decreasing.
Thus, if $\sigma(z) \leq \sigma(z')$, and user $j$ increases its bid from $\sigma(z)$ to $\sigma(z')$, its confirmation probability is still positive.
However, if user $j$ bids $\sigma(z')$, the resulting bid vector will become $(\underbrace{\sigma(z), \dots, \sigma(z)}_{k},\sigma(z'))$.
By \cref{lem:random-oca-topk}, the bid $\sigma(z')$ has zero confirmation probability, which leads to a contradiction.
Consequently, it must be $\sigma(z) > \sigma(z')$.

Second, if everyone's confirmation probability is zero assuming the miner honestly follows the inclusion rule and select from the bid vector $\bfs$,
following $\sigma$ will lead to zero social welfare.
We argue that this case is impossible 
by showing 
an alternative strategy of the global coalition 
that achieves positive social welfare, contradicting
the fact that $\sigma$ is social-welfare-maximizing.
Specifically, the alternative 
strategy works as follows: 
the global coalition pretends that the bid vector is $(b_1,\dots,b_t)$ 
where the bid $b_i$ corresponds to the primary bid of any user, say, the first user.
Recall that the first user's true value is $z$.
Thus, social welfare is at least $z\cdot x - p$, where $p$ is 
the expected payment of all users.
Since $b_i$ has the highest confirmation probability among $b_1,\dots,b_t$, $p$ is at most $x\cdot \sum_{j=1}^t b_j$.
Because $z > c^* > \sum_{j=1}^t b_j / x$, the social welfare is positive.





\end{proof}

\begin{corollary}
\label{cor:continuity}
Suppose a \truthful{} mechanism satisfies UIC and OCA-proofness with the globally optimal strategy $\sigma$, and let the block size be $k$.
Then, there exists a constant $c^*$ such that for any real number $a >
c^*$, there exists a real number $a^* \geq a$ at which $\sigma$ is continuous.
\end{corollary}
\begin{proof}
	A well-known fact is that any function that is monotone
	must have only countably many discontinuities 
	within any interval $[x_1, x_2]$ 
	where $x_2 > x_1$; further, all such discontinuities 
	must be jump discontinuities.
	By \cref{lemma:sigma-monotone}, there exists a constant $c^*$ such that $\sigma$ is monotone for all inputs larger than $c^*$.
	Thus, for any $a > c^*$, within the interval $[a,a+1]$, $\sigma$ is monotone and there must exist a point $a^*$ such that $\sigma$ is continuous at $a^*$.
\end{proof}

\begin{lemma}
Suppose the mechanism satisfies MIC and OCA-proofness with the globally optimal strategy $\sigma$.
Let ${\bf a}, \bfb, \bfb'$
be arbitrary vectors. 
Consider two scenarios  
$(\sigma({\bf a}), \sigma(\bfb))$ 
and $(\sigma({\bf a}), \sigma(\bfb'))$. 
Suppose that 
$\sigma(\bfb)$ and $\sigma(\bfb')$ have 0 confirmation
probability in each of the two scenarios, respectively. 
Then, both scenarios enjoy the same expected
miner utility, social welfare, and total user utility.
\label[lemma]{lem:random-oca-freeknob}
\end{lemma}
\begin{proof}
The equivalence in expected miner utility follows directly from MIC
and the fact that the miner can change from $\bfb$
to $\bfb'$ for free, or vice versa, since these bids are unconfirmed.

The equivalence 
in social welfare 
follows from OCA-proofness, and the fact that $\sigma$ maximizes social welfare.
Suppose, for example, the social welfare under
$(\sigma(\bfa), \sigma(\bfb))$
is strictly greater than
$(\sigma(\bfa), \sigma(\bfb'))$. Then, under the scenario
where the users' true values are $(\bfa, \bfb')$, everyone following $\sigma$ will
not maximize social welfare, since there is an OCA that have  
the users bid $(\sigma(\bfa), \sigma(\bfb))$ instead which strictly increases
the social welfare.

The equivalence in total user utility follows directly from the above, since 
total user utility is the difference between the social welfare and 
the miner utility. 
\end{proof}


\begin{lemma}
Suppose a \truthful{} mechanism satisfies UIC, MIC and OCA-proofness with the globally optimal strategy $\sigma$, 
and the block size is $k$.
Then, there exists a constant $c^*$ such that for any $i \in [k]$, the following holds.
\begin{itemize}
\item 
Consider a scenario with the arbitrary bid vector
$\bfs = \underbrace{\sigma(b), \ldots, \sigma(b)}_i, \sigma(b_{i+1}), \ldots, \sigma(b_{k+1})$,
where $b > b_{i+1} > b_{i+2} > \ldots > b_{k+1} > c^*$.
Then, if $\sigma$ is continuous at $b$, the bids $\sigma(b_{i+1}), \ldots, \sigma(b_{k+1})$ have 0 confirmation probability.
\item 
Moreover, if there are $k+1$ users with the true value $\bfs$ and $\sigma$ is continuous at $b$, the expected total user utility is $0$ when all users bid truthfully.
\end{itemize}


\label[lemma]{lem:random-oca-0util}
\end{lemma}
\begin{proof}
By \cref{lemma:sigma-monotone}, there exists a constant $c^*$ such that for any real numbers $z,z'$ satisfying $z > z' > c^*$, it must be $\sigma(z) > \sigma(z')$.
Throughout this proof, we assume $b_{k+1} > c^*$ so that $\sigma$ is strictly increasing for any input larger than $b_{k+1}$.
By \cref{cor:continuity}, there exists a real number $b > b_{k+1}$ such that $\sigma$ is continuous at $b$, and we choose $b$ such that this condition holds.
We will prove the lemma by induction.

\paragraph{Base case.}
The base case is when $i = k$;
that is, there are $k+1$ users with the true values $\bfs = (\underbrace{\sigma(b), \ldots, \sigma(b)}_k, \sigma(b_{k+1}))$.
Notice the true value is $\sigma(b)$ instead of $b$ for the first $k$ users.
By \cref{lem:random-oca-topk}, 
$\sigma(b_{k+1})$ has 0 confirmation probability.
By \cref{lem:random-oca-freeknob},
the total user utility under 
$\bfs$ is the same
as under $\bfs' = (\underbrace{\sigma(b), \ldots, \sigma(b)}_k, 
\sigma(b-\delta))$ where $\delta > 0$ can be arbitrarily small 
--- since in both scenarios, the last bid
is unconfirmed.
Now, suppose the bid vector is $\bfs'$.
Fix any user with the true value $\sigma(b)$, say the first user.
If it lowered its bid to  
anything less than $\sigma(b - \delta)$,
its confirmation probability
must be $0$ by \cref{lem:random-oca-topk}.
Since the mechanism is UIC, by Myerson's lemma, 
the payment of the first user is at least
$\sigma(b - \delta) \cdot x_1(\bfs')$.
In this case, the first user's utility is at most $(\sigma(b) - \sigma(b - \delta)) \cdot x_1(\bfs') \leq \sigma(b) - \sigma(b - \delta)$.
Therefore, the total user utility under $\bfs'$ (and hence under $\bfs$) is at most $k \cdot (\sigma(b) - \sigma(b - \delta))$.
 Since $\sigma$ is continuous at $b$, and the above holds for any $\delta > 0$, 
the total user utility under $\bfs$ is at most $0$. 
By the individual rationality of the payment rule, each user's utility is non-negative when it bids truthfully.
Therefore, the expected total user utility is $0$.

\paragraph{Inductive step.}
Fix any $i \in \{2,\dots,k\}$, and 
suppose the lemma holds.
We want to show that the lemma also holds for any $\bfs_{i-1} = (\underbrace{\sigma(b), \ldots, \sigma(b)}_{i-1}, \sigma(b_i), \sigma(b_{i+1}), \ldots, \sigma(b_{k+1}))$.
Start from the scenario 
$\bfs_i = (\underbrace{\sigma(b), \ldots, \sigma(b)}_i, \sigma(b_{i+1}), \ldots, \sigma(b_{k+1}))$, 
and imagine that the $i$-th bid lowers
from $\sigma(b)$ to $\sigma(b_i)$.
Since the $i$-th user has 0 utility under $\bfs_i$, 
by UIC and Myerson's lemma, 
its confirmation probability must become 0 when
it lowers its bid to $\sigma(b_i)< \sigma(b)$.
By \cref{lem:random-oca-ordered}, the bids $\sigma(b_{i+1}), \ldots, \sigma(b_{k+1})$ must
have 0 confirmation probability too.

We next show that under $\bfs_{i-1}$, total user utility is $0$.
Consider another bid vector\\ $\bfs' := (\underbrace{\sigma(b), \ldots, \sigma(b)}_{i-1}, 
\sigma(b-\delta), \ldots, \sigma(b- (k+2 - i)\delta))$ for some $\delta > 0$.
By the argument above, the bids numbered $i+1, \ldots, k+1$ are not confirmed.
By \cref{lem:random-oca-freeknob}, the total user utility under $\bfs_{i-1}$
is the same as under the scenario $\bfs'$.
In scenario $\bfs'$, 
if any of the $\sigma(b)$ bidders, say the first bidder,
lowers 
its bid from $\sigma(b)$ to 
anything less than $\sigma(b- (k+2 - i)\delta)$, 
its confirmation probability becomes $0$ by \cref{lem:random-oca-topk}.
By UIC and Myerson, 
the first user's payment
under $\bfs'$
is at least $\sigma(b- (k+2 - i)\delta) \cdot x_1(\bfs')$,
and its utility is at most 
$(\sigma(b) - \sigma(b- (k+2 - i)\delta)) \cdot x_1(\bfs')
\leq \sigma(b) - \sigma(b- (k+2 - i))\delta$.
Therefore, the total user utility under $\bfs'$ (and also under $\bfs_{i-1}$)
is upper bounded by 
$k \cdot (\sigma(b) - \sigma(b- (k+2 - i)\delta))$. 
 Since $\sigma$ is continuous at $b$ and the above holds
 for any arbitrarily small $\delta > 0$, it means
that the total user utility under 
$\bfs_{i-1}$ is at most $0$.
By the individual rationality of the payment rule, each user's utility is non-negative when it bids truthfully.
Therefore, the expected total user utility is $0$.
\end{proof}

\begin{lemma}
Suppose a \truthful{} mechanism satisfies UIC, MIC and OCA-proofness with the globally optimal strategy $\sigma$,
and the block size is $k$.
Then, there exists a constant $c^*$ such that for any $a^* > c^*$,
if there is only a single user with the true value $\sigma(a^*)$ and $\sigma$ is continuous at $a^*$,
its utility is zero when it bids truthfully.
\label[lemma]{lem:random-only-user}
\end{lemma}
\begin{proof}
By \cref{lem:random-oca-0util}, 
there exists a constant $c^*$ such that 
under the scenario $\bfa := \left(\sigma(a^*), \sigma(a_2), \ldots, 
\sigma(a_{k+1})\right)$ where $a^* > a_2 > \ldots > a_{k+1} > c^*$, the
utility of the first user is $0$ assuming it is bidding truthfully.
Moreover, all users except the first user must have 0 confirmation probability.
By \cref{lem:random-oca-freeknob}, the total expected user utility under $\bfa$ and under a single bid $\sigma(a^*)$ are the same. 
Therefore, under a single bid $\sigma(a^*)$, the user's utility is also 0.
\end{proof}

\begin{theorem}
No non-trivial, possibly randomized \truthful{}
TFM can simultaneously satisfy UIC, MIC, 
and OCA-proofness 
when the block size is finite.
\label{thm:oca-random-impossibility}
\end{theorem}


\begin{proof}
We will show that under any sufficiently large $a$ at which $\sigma$ is continuous, 
the confirmation probability under 
a single bid $\sigma(a)$ is non-zero. 
If we can show this, 
then we can show a contradiction to UIC and the Myerson's lemma. 
Specifically, consider $b > a$ which both are sufficiently large and
at which $\sigma$ is continuous (\cref{cor:continuity} guarantees such $a,b$ exist).
By \cref{lem:random-only-user},   
under both scenarios, the user utility is $0$ (assuming
the bid, $\sigma(a)$ or $\sigma(b)$, is the same as the user's true value). 
By \cref{lemma:sigma-monotone}, for sufficiently large $a$,
we have $\sigma(b) > \sigma(a)$.
However,
by Myerson's lemma, if the confirmation probabilities
under a single $\sigma(a)$ and under a single $\sigma(b)$ 
are both non-zero, it cannot be  
that in both cases the user's utility is $0$.
Otherwise, the user with the true value $\sigma(b)$ can underbid $\sigma(a)$.
Since the confirmation probability is non-zero and the payment is at most $\sigma(a)$, the user enjoys positive utility, which violates UIC.

We now show that the confirmation probability under a single bid $\sigma(a)$ 
for a sufficiently large $a$ is non-zero assuming the miner honestly follows the inclusion rule.
Suppose this is not true for the sake of a contradiction;
that is, when the bid vector is just $\sigma(a)$ and the miner is honest, no one has positive confirmation probability, so the social welfare is zero.
Since the mechanism is non-trivial,
there exists a scenario $\bfb = (b_1, \ldots, b_t)$
such that the total confirmation probability
is non-zero, and let $x > 0$
be the confirmation probability 
of some bid $b_i$ in $\bfb$.
Consider $a > \sum_{j \in [n]} b_j/x$, and suppose
there is only one user with the true value $a$.
We will show that the strategy $\sigma(a)$ does not maximize social welfare 
which contradicts to the fact that $\sigma$ is social-welfare-maximizing-strategy.
Specifically, 
the user can post the bids $(b_1, \ldots, b_t)$ instead
where position $i$ corresponds to 
the user's primary bid and all others are fake bids.
Note that under $\bfb$, the 
total user payment 
is at most $\sum_{j \in [n]} b_j$. 
Since the user's true value is greater 
than $\sum_{j \in [n]} b_j/x$ and its confirmation probability
is $x$, 
the total user utility (and thus the social welfare) 
is positive under this new bidding 
strategy.

\ignore{
We reach a contradiction with \elaine{refer}
by arguing that for a sufficiently large $a$,
the expected user utility  
under a single bid $\sigma(a)$ 
must be positive assuming the bid reflects the user's true value.
\elaine{TO FIX: this is not true because it could be that the output
of sigma is small}
}
\end{proof}



\begin{remark}
\label{remark:second-price-no-burning}
Notice that Chung and Shi~\cite{foundation-tfm} showed that no TFM can simultaneously satisfy UIC and $1$-SCP when the block size is finite.
Their impossibility does not rely on MIC, while MIC is necessary for proving \cref{thm:oca-random-impossibility}.
Specifically, consider \emph{the second-price auction without burning} defined as follows.
Given the block size $k$, the miner includes the top up to $k$ bids in the block.
If the number of included bids is less than $k$, then all included bids are confirmed, and pay nothing.
If the number of included bids is $k$, then the top $k-1$ bids are confirmed, while the $k$-th bid is unconfirmed (break ties arbitrarily).
All confirmed bids pay the $k$-th bid, and all the payments go to the miner.
This mechanism satisfies UIC since the confirmation and the payment satisfy the requirements of Myerson's lemma.
The mechanism also satisfies OCA-proofness when $\sigma$ is just bidding truthfully since the confirmed bids are always the ones with the highest true values, which maximizes the social welfare.
However, it does not satisfy MIC, since the miner is incentivized to inject a fake bid to increase the $k$-th bid.
\end{remark}

\section{How to Circumvent the Impossibilities}
\label{sec:bypass-impossibility}
In this section, we discuss possible ways to circumvent the impossibilities.
Recall that in the definition of OCA-proofness, we require all users act independently in the globally optimal strategy $\sigma$, and $\sigma$ should only output a single bid.
If we remove either of the two requirements, we can have mechanisms that satisfies UIC, MIC and OCA-proofness, where we present in \cref{sec:posted-price-burning,sec:counterexample}.
These mechanisms are somewhat contrived
and not obviously relevant for real-world blockchain protocols;
the primary point of these mechanisms is to ``stress-test'' the definitions, and these examples highlight the subtlety of the modeling and help us to have a better understanding of the trade-offs between various notions of collusion-resilience.
In \cref{sec:inclusion-rule-respecting}, we consider an even weaker definition, called \emph{inclusion-rule-respecting}, which might be the weakest but still meaningful definition that captures the notion of ``no way to steal from the protocol.'' 
Later in \cref{sec:BayesianIC}, 
we introduce a new model that captures the use of trusted hardware 
and formalize a Bayesian notion of incentive compatibility.
We show that, in this model, it is possible to construct a TFM that satisfies Bayesian incentive compatibility.

\subsection{Allowing the Globally Optimal Strategy to Coordinate}
\label{sec:posted-price-burning}
OCA-proofness requires that all users act independently in 
the globally optimal strategy $\sigma$. 
However, if we remove this restriction, then 
the following auction would simultaneously satisfy UIC, MIC, and OCA-proof (with the aforementioned relaxation).

\ignore{
the above auction 
has the following strategy $\sigma$ which 
maximizes the social welfare: have the up to $k$ users
with the highest true values bid above $r$, and everyone
else bids $0$. 
}

\paragraph{Posted price auction with random selection and total burning.}
The mechanism is parametrized with some reserve price $r>0$ and block size $k$.
Any bid that is at least $r$ 
is considered eligible. Randomly select 
up to $k$ bids to include in the block. All included
bids are confirmed and they each pay $r$. All payment is burnt and the miner
receives nothing.
It is not hard to see that the above mechanism satisfies UIC and MIC. 
It also satisfies OCA-proofness with the aforementioned relaxation
since the global coalition can just have 
the users with the top $k$ values bid more than $r$, while everyone else bids $0$.  

Note that allowing all users to coordinate bids does not trivialize
the definition, as this relaxed OCA-proofness definition still
requires that the miner follow the TFM's inclusion rule honestly.  For
example, in the auction above, the miner has to select $k$ bids
uniformly at random among all eligible bids.

\subsection{Allowing the Globally Optimal Strategy to Output Multiple Bids}
\label{sec:counterexample}


We next show that when we allow the globally optimal
strategy $\bfsig$ to output multiple bids, 
then we can actually construct a \truthful{} mechanism
that satisfies UIC, MIC, and OCA-proofness.
Similar to \cref{sec:non-direct-oca-counterexample}, we try to signal to the mechanism when everyone is adopting
the globally optimal strategy $\sigma$. 
In this mechanism, since $\sigma$ can output multiple bids, users can signal through fake bids. 
We assume the miner and the mechanism can distinguish whether a set of bids come from the same user.
In practice, this can be implemented by checking whether these bids are signed by the same public key.
Specifically, consider the following mechanism
parametrized by a reserve price $r > 1$.
\begin{itemize}
\item 
{\bf Globally optimal strategy $\bfsig(v)$}: 
If the user's true value $v \geq r$, 
$\bfsig(v)$ outputs $(r, r/v)$ 
where $r/v \in (0, 1]$  is a uniquely decodable encoding
of the user's true value $v$.
Otherwise, if $v < r$, $\bfsig(v)$
simply outputs $\emptyset$, i.e., 
the user drops its bid.
\item 
{\bf Inclusion rule.} 
\begin{itemize}
\item 
{\bf Case 1}:  Check if it is the case that every user $i \in [n]$ submitted
a message of the form $(r, r/v_i)$ where $r/v_i \in (0, 1]$
for all $i \in [n]$.
If so, decode the true value vector $(v_1, \ldots, v_n)$.
Choose the 
$\min(n, k)$ bids with the highest true values, and include
their corresponding $r$ bids in the block.
\item 
{\bf Case 2}:
Else, if the number of bids is even, henceforth denoted $2n$, 
check whether there are exactly $n$ bids at $r$
and $n$ number of bids that lie in the range $(0, 1]$.  
If so, interpret the bid
vector as $(r, r/v_1),   \ldots, (r, r/v_{n})$, as if 
each pair $(r, r/v_i)$
comes from the $i$-th user for $i \in [n]$.
Now, apply the algorithm of Case 1.
\ignore{
suppose every user's bid is either of the form  
$(r, \bfb_{\rm aux})$
or $0$.  
Let $n$ be the number of bids that are of the form 
$(r, \bfb_{\rm aux})$. For these $n$ bids, 
decode each user's true value from 
the auxiliary information vector $\bfb_{\rm aux}$.
Choose the $\min(n, k)$ bids with the highest
true values to include in the block. 
}
\item 
{\bf Case 3}: 
Else, if the number of bids is odd, henceforth denoted $2n+1$, check if
there are $n$ bids at $r$ and $n$ bids in the range of $(0, 1]$,
and one last bid denoted $r'$.
If $r' > r$, then choose $r'$ to include
in the block, and among the remaining $n$ bids 
at $r$, randomly choose $\min(k-1, n)$
of them to include in the block; else if $r' \leq r$, 
randomly choose 
$\min(k, n)$ bids 
among all the $r$-bids
to include in the block.
\item 
{\bf Case 4}:
In all other cases, 
randomly choose 
$\min(k, n)$ bids among all bids that are \emph{strictly larger than} $r$ to include
in the block.
\Hao{I changed ``at least'' to ``strictly larger''}

\ignore{
not everyone's bid is of the form 
suppose every user's bid is either of the form 
$(r, \bfb_{\rm aux})$ or $0$.  
Let $n$ be the number of bids that bid at least $r$.
Among these $n$ bids,
randomly select $\min(k,n)$  
of them to include in the block.
}
\end{itemize}
\item 
{\bf Confirmation, payment, and miner revenue rules}:
All included bids that are at least $r$ are confirmed.
Each confirmed bid pays $r$, and the miner gets nothing.
\end{itemize}


\begin{claim}
The above mechanism satisfies UIC, MIC, and 
a relaxed notion of OCA-proofness 
in which $\sigma$ is allowed to output multiple bids.
\end{claim}
\begin{proof}
MIC is obvious since the miner always gets 0 revenue.
For OCA-proofness, observe that if the global coalition
all take the strategy $\bfsig$, then 
case 1 of the inclusion rule will be triggered,
and thus the social welfare is maximized. 

Proving UIC is the most complicated part.
Fix any user $i$.
Notice that each confirmed bid pays $r$, so if user $i$'s true value is at most $r$, it cannot get positive utility regardless its bid.
Thus, bidding truthfully is a dominant strategy.
Henceforth, we assume user $i$'s true value is strictly larger than $r$.
There are two possibilities.
\begin{enumerate}
\item 
First, before user $i$ submits its bid, the bid vector contains exactly $n$ bids at $r$ and $n$ bids lie in the range $(0,1]$ for some $n$.
Then, bidding truthfully will lead to Case 3, where user $i$'s bid is guaranteed to be included and confirmed, and thus maximizes the utility.
\item 
Second, before user $i$ submits its bid, the bid vector does not contain exactly $n$ bids at $r$ and $n$ bids lie in the range $(0,1]$ for some $n$.
Suppose there is no bid strictly larger than $r$ in the bid vector before user $i$ submits its bid.
Then, if user $i$ bids truthfully, its bid will be the only bid larger than $r$, and the inclusion rule goes to Case 4. 
Since $k \geq 1$, user $i$'s bid is guaranteed to be included and confirmed, and thus maximizes the utility.
On the other hand, suppose there are $t \geq 1$ bids strictly larger than $r$ in the bid vector before user $i$ submits its bid.
Then, the inclusion can never go to Case 1 or Case 2.
If user $i$ bids truthfully, the inclusion rule must go to Case 4, and user $i$'s bid is included with probability $\min(k,t+1) / (t+1)$.
However, if the strategic bidding leads to Case 3, it must be $t = 1$.
Then, user $i$'s inclusion probability has dropped from $1/2$ to $0$ if $k=1$, or from $1$ to at most 1 if $k\geq 2$.
On the other hand, if the strategic bidding leads to Case 4, 
any untruthful bidding $ > r$ leads to the same inclusion probability and the same payment, so the utility does not change,
while any untruthful bidding $\leq r$ leads to the zero inclusion probability which leads to zero utility. 
Finally, notice that injecting fake bids never helps in Case 4.
Thus, in all cases, the utility is better than bidding truthfully.
\end{enumerate}
\end{proof}

Recall that weakly symmetric mechanisms do not make use of the bidders' identities or other
auxiliary information except for tie-breaking among equal bids. Because the mechanism above
relies on the fact that the inclusion rule can distinguish whether a set of bids come from the same
user, it may not be obvious that the mechanism satisfies weak symmetry. The next claim proves this formally.

\begin{claim}
The above mechanism satisfies weak symmetry. 
\end{claim}
\begin{proof}
To show weak symmetry, we give an equivalent description as follows.
Given a bid vector $\bfb$, the sorting algorithm checks whether $\bfb$ satisfies the condition of Case 1.
If so, the sorting algorithm places the multiple $r$ bids corresponding to the true value vector $(v_1,\dots,v_n)$ (defined in Case 1 above) from high to low, and sorts the rest in the descending order.
Otherwise, if $\bfb$ does not satisfy the condition of Case 1, the sorting algorithm sorts $\bfb$ in the descending order and breaks tie arbitrarily.

In the original mechanism, the only case that depends on the extra information (identity, in this case) is Case 1.
By applying the sorting algorithm described above, the inclusion rule can implement Case 1 by only depending on the amount of the bids and their relative position in the sorted bid vector.

\end{proof}

\subsection{Inclusion-Rule-Respecting}
\label{sec:inclusion-rule-respecting}
It is also natural to consider 
a further relaxation of OCA-proofness henceforth 
called ``inclusion-rule-respecting''.
In comparison with OCA-proofness, 
inclusion-rule-respecting 
only requires
that the globally optimal strategy
$\sigma$ not involve
altering the inclusion rule,
and all other constraints on $\sigma$ are removed.
Inclusion-rule-respecting is strictly weaker
than the notions in \Cref{sec:posted-price-burning}
and
\Cref{sec:counterexample}.
Therefore, the mechanisms mentioned in \Cref{sec:posted-price-burning}
and 
\Cref{sec:counterexample}
simultaneously satisfy UIC, MIC, and inclusion-rule-respecting too.

The ``inclusion-rule-respecting'' notion has an intuitive interpretation:
it discourages miners from altering the intended protocol implementation; in other words,
all profiting strategies can be implemented
as bidding strategies 
above the protocol layer.
In this sense, ``inclusion-rule-respecting'' also captures
the intuition of ``no way to steal from the protocol''.
One can also view it as follows: 
global SCP 
captures ``not stealing from the protocol'' where the protocol  
is the union of the miner's inclusion rule as well as users' honest
bidding strategies;  and 
inclusion-rule-respecting captures the same notion but where
the protocol is only the underlying blockchain protocol.
OCA-proofness is somewhere in between.



\section{Transaction Fee Mechanisms in the Trusted Hardware Model}
\label{sec:BayesianIC}

In this section, we introduce a new model, the trusted-hardware model, 
which also helps us bypass the impossibility results.
Trusted hardware, like Intel SGX, refers to a secure computation environment 
whose execution is enforced and whose internal state remains hidden from strategic players. 
In a transaction fee mechanism, trusted hardware can enforce the inclusion rule and 
ensure that strategic players cannot view or tamper with honest users' bids during block construction.

We first define the model formally in \cref{sec:trusted-hardware}, 
then introduce the Bayesian variants of incentive compatibility properties in \cref{sec:BayesianIC-definitions}.
In \cref{sec:second-price-in-MPC}, we show that the second-price auction with reserve satisfies UIC, Bayesian MIC, and global SCP
demonstrating that we can bypass the impossibility 
by adopting the Bayesian notion of incentive compatibility.
Finally, in \cref{sec:Bayesian-impossibility}, we show that 
there remains a fundamental trade-off between global SCP and $1$-SCP
even if we only aim for Bayesian incentive compatibility.

\subsection{The Trusted Hardware Model}
\label{sec:trusted-hardware}
Recall that in the plain model (\cref{sec:tfm-def}), 
the miner observes all honest users' bids and may deviate from the inclusion rule. 
In contrast, the trusted-hardware model assumes that the inclusion rule is enforced 
within a trusted execution environment (e.g., Intel SGX), 
which users interact with directly via encrypted bids.
In this setting, when designing a mechanism,
we only need to specify the \emph{allocation rule}
which is a combination of the inclusion and confirmation rules.
We say a TFM is \emph{trivial} if the allocation probability of all transactions is always zero;
otherwise, it is called \emph{non-trivial}.

Below, we define the game that captures the using of trusted hardware in the transaction fee mechanism design.
To formalize the fact that the bids are hidden from the miner,
we assume the existence of an ideal trusted functionality $\mathsf{T}$
who the honest users interact with to submit their bids
and execute the mechanism rules.
When the strategic player is an individual user,
the execution of a transaction fee mechanism in the trusted-hardware model is described by the following game:
\begin{enumerate}
\item 
Each honest user submits a bid to $\mathsf{T}$.
The strategic user can submit a non-negative number of bids to $\mathsf{T}$.
\item 
The trusted party $\mathsf{T}$ collects all bids and outputs the outcomes of all players
according to the allocation rule, the payment rule, and the miner revenue rule.
\end{enumerate}
When the strategic player, denoted as $\mcal{C}$, is the miner or a coalition of the miner and some users,
the execution of a transaction fee mechanism in the trusted-hardware model is described by the following game:
\begin{enumerate}
\item 
Each honest user $i$ submits its identity $i$ to the miner, 
and its bid $b_i$ to $\mathsf{T}$.
\item
Let $\mcal{H}$ denote the set of honest identities received by the miner.
The coalition $\mcal{C}$ chooses a subset $\mcal{H}' \subseteq \mcal{H}$,
and a vector of bids $\bfb_{\mcal{C}}$ of arbitrary length.
Then, $\mcal{C}$ submits $\mcal{H}'$ and $\bfb_{\mcal{C}}$ to $\mathsf{T}$.
\item
The trusted party $\mathsf{T}$ outputs the outcomes of all players 
according to the allocation rule, payment rule, and the miner revenue rule
by using $\{b_i\}_{i \in \mcal{H}'} \cup \bfb_{\mcal{C}}$ as the input.
\end{enumerate}

In the games above, 
an honest user always takes on a single identity, and submits a bid according to the bidding rule.
However, a strategic user can take a non-negative number of identities,
and submit bids on behalf of these identities which possibly deviate from the bidding rule.
An honest miner does not take on any identity or submit any bid,
and it always forwards all users' identities to $\mathsf{T}$.
A strategic miner can take a non-negative number of identities and submit fake bids.
Moreover, we allow the strategic miner to choose a subset $\mcal{H}' \subseteq \mcal{H}$ of 
users' identities.
This captures the fact that the trusted hardware relies on the miner to relay the bids,
so the miner can censor some honest users' bids by not forwarding them to the trusted hardware.
A miner-user coalition can adopt a combination of the above strategies.



\begin{remark}
Shi, Chung, and Wu~\cite{crypto-tfm} introduced the \emph{MPC-assisted model}, 
which further restricts the miner's strategy space.
In this model, the mechanism is executed by a multiparty computation (MPC) protocol
run by multiple miners.
The MPC-assisted model is appropriate when the strategic players do not control all the parties participating in the MPC.
However, global SCP and OCA-proofness focus on the global coalition of all users and the miners.
If such a global coalition corrupts all MPC participants, the protocol's security is compromised. Therefore, the MPC-assisted model is incompatible with global SCP and OCA-proofness.

In terms of strategy space, the trusted hardware model lies between the plain model and the MPC-assisted model. Both the trusted hardware and MPC-assisted models prevent the strategic miner from learning the honest users' bids. However, unlike in the MPC-assisted model, 
the trusted hardware model still allows the strategic miner to censor some bids.
\end{remark}

\subsection{Defining Bayesian Incentive Compatibility}
\label{sec:BayesianIC-definitions}

The properties, UIC, MIC, and SCP, defined in \cref{sec:IC}
require that the honest behavior maximizes the strategic players' utilities
\emph{for any} bid vector submitted by the honest players.
In contrast, Bayesian incentive compatibility assumes the honest players' values are drawn from a known distribution,
and it only requires that the honest behavior maximizes 
the strategic players' expected utilities \emph{over the distribution of the honest players' values}
assuming the honest players follow the bidding rule.
This notion is particularly justified when the strategic players have to submit their bids 
without having seen other honest users' bids.

Let $\mcal{D}$ denote a distribution representing the true values of honest users,
which is known to the strategic players.
Let $\beta$ denote the bidding rule,
and we define $\beta(\mcal{D})$ as the distribution of bids 
when a user follows the bidding rule $\beta$ and its true value is drawn from $\mcal{D}$.
The Bayesian incentive compatibility properties are defined as follows.

\begin{definition}[Bayesian user incentive compatible]
\label{def:Bayesian-UIC}
We say a TFM satisfies \emph{Bayesian user incentive compatible (Bayesian UIC)}
iff the following holds: 
for any number $n$ of users,
for any user $i$, for any true value $v_i$ of user $i$,
and for any bid vector $\bfb_{i}$ submitted by user $i$,
it holds that \[
\underset{{\bfb}_{-\mcal{C}}\sim \beta(\mcal{D})^{n-1}}{\E}
\left[{\sf util}_i({\bfb}_{-\mcal{C}}, v_i)\right] \geq
\underset{{\bfb}_{-\mcal{C}}\sim \beta(\mcal{D})^{n-1}}{\E} 
\left[{\sf util}_i({\bfb}_{-\mcal{C}}, {\bf b}_i)\right],
\]
where ${\sf util}_i({\bf b})$ denotes user $i$'s expected utility 
(taken over the random coins of the TFM)
assuming the miner implements the TFM honestly.
\end{definition}

\begin{definition}[Bayesian miner incentive compatible]
\label{def:Bayesian-MIC}
We say a TFM satisfies \emph{Bayesian miner incentive compatible (Bayesian MIC)}
iff the following holds: 
for any set $\mcal{H}$ of users with arbitrary size $n$,
for any subset $\mcal{H}' \subseteq \mcal{H}$,
and for any bid vector $\bfb_{\mcal{M}}$ submitted by the miner,
it holds that \[
\underset{{\bfb}_{\mcal{H}}\sim \beta(\mcal{D})^{n}}{\E}
\left[{\sf util}_{\mcal{M}}({\bfb}_{\mcal{H}})\right] \geq
\underset{{\bfb}_{\mcal{H}}\sim \beta(\mcal{D})^{n}}{\E} 
\left[{\sf util}_{\mcal{M}}({\bfb}_{\mcal{H}'}, \bfb_{\mcal{M}})\right],
\]
where ${\bfb}_{\mcal{H}'}$ denotes the subvector of ${\bfb}_{\mcal{H}}$ corresponding to the users in $\mcal{H}'$, and
${\sf util}_{\mcal{M}}({\bf b})$ denotes the miner's expected utility 
(taken over the random coins of the TFM)
given the input to the allocation rule is $\bfb$.
\end{definition}

\begin{definition}[Bayesian side-contract-proof]
\label{def:Bayesian-SCP}
We say a TFM satisfies \emph{Bayesian $c$-side-contract-proof (Bayesian $c$-SCP)}
iff the following holds: 
for any miner-user coalition $\mcal{C}$ consisting of the miner
and at most $c$ users,
for any true values $\bfv_{\mcal{C}}$ of the users in $\mcal{C}$,
for any bid vector $\bfb_{\mcal{C}}$ submitted by $\mcal{C}$,
for any set $\mcal{H}$ of honest users with arbitrary size $n$,
and for any subset $\mcal{H}' \subseteq \mcal{H}$,
it holds that \[
\underset{{\bfb}_{\mcal{H}}\sim \beta(\mcal{D})^{n}}{\E}
\left[{\sf util}_{\mcal{C}}({\bfb}_{\mcal{H}}, \bfv_{\mcal{C}})\right] \geq
\underset{{\bfb}_{\mcal{H}}\sim \beta(\mcal{D})^{n}}{\E} 
\left[{\sf util}_{\mcal{C}}({\bfb}_{\mcal{H}'}, \bfb_{\mcal{C}})\right],
\]
where ${\bfb}_{\mcal{H}'}$ denotes the subvector of ${\bfb}_{\mcal{H}}$ corresponding to the users in $\mcal{H}'$, and
${\sf util}_{\mcal{C}}({\bf b})$ denotes the coalition's expected utility 
(taken over the random coins of the TFM)
given the input to the allocation rule is $\bfb$.
\end{definition}

Because the strategic players considered in global SCP and OCA-proofness are the coalition of all users and the miner,
and there are no honest users outside the coalition,
Bayesian notions of global SCP and OCA-proofness are equivalent to their \emph{ex post} counterparts,
i.e., \cref{def:globalSCP} and \cref{def:OCAproof}.

\subsection{Second-Price Auction Achieves UIC + Bayesian MIC + Global SCP}
\label{sec:second-price-in-MPC}
In this section, we assume each user's true value is independently drawn from a distribution~$\mathcal{D}$,
and we assume the probability density function of $\mathcal{D}$ exists. 
Let $F$ and $f$ denote the cumulative distribution function and probability density function of~$\mathcal{D}$, respectively.
We define the \emph{virtual value} function as $\phi(v) = v - \frac{1 - F(v)}{f(v)}$.
We assume $\mathcal{D}$ is \emph{regular}, i.e., $\phi(v)$ is strictly increasing and $f$ is strictly positive over the support.
We set the reserve price to be $r = \phi^{-1}(0)$.
Note that for any bounded distribution over $[0, v_{\max}]$, we have $\phi(v_{\max}) = v_{\max} \geq 0$, so $\phi^{-1}(0)$ exists.

\begin{mdframed}
	\begin{center}
		{\bf Second-price auction with reserve}
	\end{center}

\noindent {\bf Parameters}: the reserve $r$ and the block size $k$.

\noindent {\bf Inputs}: a bid vector $(b_1,\dots,b_N)$ where $b_1 \geq \cdots \geq b_N$ (break tie arbitrarily).

\begin{itemize}
\item 
{\bf Bidding rule}: 
Each user bids its true value.
\item 
{\bf Allocation rule and payment rule}: 
All bids among $b_1,\dots,b_k$ that are at least $r$ are confirmed.
Each confirmed bid pays $\max(r, b_{k+1})$.
\item 
{\bf Miner revenue rules}:
For each confirmed bid, the reserve is burnt, and the miner gets the rest.
\end{itemize}
\end{mdframed}

\begin{theorem}
The second-price auction with reserve above satisfies the following properties:
\begin{itemize}
\item 
It satisfies UIC, and global SCP for users' values regardless of distribution $\mathcal{D}$.
\item 
If each buyer's true value is i.i.d.~sampled from a bounded regular distribution $\mcal{D}$,
it additionally satisfies Bayesian MIC.
\end{itemize}
\end{theorem}
\begin{proof}
We start with the proof of UIC and global SCP.
The proof of Bayesian MIC is more involved, so we will postpone it to the end.

\paragraph{UIC.}
Since the allocation and payment rules align with the conditions of Myerson's lemma, 
a user cannot increase their utility by bidding untruthfully. 
Additionally, injecting fake bids cannot reduce the payment required. 
Thus, the mechanism satisfies UIC. 
This property holds even if users observe others' bids before submitting their own.

\paragraph{Global SCP.}
Each confirmed bid burns an amount $r$, so confirming a transaction with a value~$< r$ reduces total social welfare.
On the other hand, confirming a transaction with value~$v \geq r$ increases social welfare by $v - r$, regardless of the exact payment.
Because the mechanism always confirms the highest bids that are at least $r$,
everyone submits the true value maximizes the social welfare.
Thus, it satisfies global SCP.

\paragraph{Bayesian MIC.}
In the second-price auction, censoring honest users' bids never increases the payment from confirmed users
nor the number of confirmed users.
Hence, we assume the miner always forwards all honest users' bids to the trusted hardware,
while it may inject fake bids.

Let $\Pi$ denote the second-price auction with reserve~$r$ described above.
Define $\Pi'$ as an identical mechanism, except that no payment is burned: the miner receives the full payment.
We will show the following:
\begin{enumerate}
    \item $\Pi'$ satisfies Bayesian MIC.\footnote{Although $\Pi'$ satisfies Bayesian MIC, it does not satisfy global SCP; see \cref{remark:burning-reserve}.}
    \item If the miner can profit from deviating in $\Pi$, 
	it can also profit from deviating in $\Pi'$.
\end{enumerate}
This implies that $\Pi$ must also satisfy Bayesian MIC.

We first show that $\Pi'$ satisfies Bayesian MIC.
Let $d$ be the payment from a confirmed user in an honest execution,
so the miner receives $d$ per confirmed bid.
Now, suppose the miner deviates from the inclusion rule or
injects fake bids so that the payment becomes $d'$. 
Since only confirmed bids generate revenue, the deviation is profitable only if
\begin{equation}
\label{eq:second-price}
    d' \cdot \Pr[v \geq d' \mid v \geq d] > d,
\end{equation}
or equivalently, $d' \cdot (1 - F(d')) > d \cdot (1 - F(d))$.

Since \cref{eq:second-price} only holds if $d' > d$,
we assume $d' > d$ henceforth.
We define $h(v) = v \cdot (1 - F(v))$,
and we have $h'(v) = -\phi(v) \cdot f(v)$.
Because $\phi(v)$ is strictly increasing, $h(v)$ is non-increasing for $v \geq \phi^{-1}(0)$.
Since $d' > d \geq r = \phi^{-1}(0)$, \cref{eq:second-price} cannot hold.
Thus, fake bid injection does not increase the miner's utility in~$\Pi'$.

Now we are ready to show that $\Pi$ also satisfies Bayesian MIC.
Given any bid vector $\bfb$ from honest users,
let $P$ denote the total payment collected from the confirmed users,
and let $t$ be the number of confirmed bids
in an honest execution of $\Pi$.
In this case, the miner's utility is $P - t \cdot r$.
On the other hand, 
let $P'$ denote the total payment collected from the honest confirmed users (not including the miner's payment),
and let $t'$ be the number of confirmed bids (including fake bids)
if the miner injects some fake bids $\bff$ in $\Pi$.
In this case, the miner's utility is $P' - t' \cdot r$.
Note that $t,t', P, P'$ are the functions of $\bfb$ and $\bff$.

For the sake of contradiction, suppose $\Pi$ does not satisfy Bayesian MIC.
That is, there exists a number $n$ of users such that the miner's expected utility
increases after injecting fake bids $\bff$
when $n$ users' true values are drawn from $\mcal{D}$.
Thus, we have $
\mathop{\E}_{\bfb \sim \mcal{D}^n} \left[P' - t' \cdot r\right]
> 
\mathop{\E}_{\bfb \sim \mcal{D}^n} \left[P - t \cdot r\right]
$.
After inject fake bids, the number of confirmed bids (including fake bids) never decreases,
so we have $t' \geq t$.
Thus, we have $
	\mathop{\E}_{\bfb \sim \mcal{D}^n} \left[P - t \cdot r\right]
	\geq 
	\mathop{\E}_{\bfb \sim \mcal{D}^n} \left[P - t' \cdot r\right]
$.
Combining the two inequalities, we have \[
	\mathop{\E}_{\bfb \sim \mcal{D}^n} \left[P'\right]
	> 
	\mathop{\E}_{\bfb \sim \mcal{D}^n} \left[P\right].
\]
Because the allocation rule and the payment rule of $\Pi$ and $\Pi'$ are identical,
for any $\bfb$ and $\bff$,
the number of confirmed users and the corresponding payments are identical in $\Pi$ and $\Pi'$.
Consequently, the miner can profit by injecting fake bids $\bff$ in $\Pi'$,
and it contradicts the fact that $\Pi'$ satisfies Bayesian MIC.
\end{proof}

\begin{remark}[On the Necessity of Burning the Reserve]
	\label{remark:burning-reserve}
	Burning the reserve is essential for ensuring global SCP.
	If the reserve is not burned and the miner receives the full payment, then in situations where fewer than $k$ users have values $\geq r$, low-value users are incentivized to overbid to occupy block space.
	Without burning, the coalition's utility increases with these overbids.
	Conversely, if all payments are burned, the global coalition is incentivized to reduce payments, e.g., by having all but the top $k$ bidders bid~$0$ to keep the price at~$r$.
\end{remark}

\begin{remark}[First-price auction satisfies Bayesian UIC, MIC, and global SCP in the permissioned setting]
We have seen that the second-price auction can bypass the impossibility result by appealing to the Bayesian notion of MIC. 
Interestingly, the first-price auction can also bypass the impossibility---this time via the Bayesian notion of UIC---provided 
we are in a \emph{permissioned setting}, 
where the number of users $n$ is fixed \emph{a priori} and 
is common knowledge to all players and the mechanism.

Under the assumption that other users' bids are i.i.d.~draws from a known distribution $\mcal{D}$, 
a bidder with true value $v$ in a first-price auction has a Bayesian optimal strategy to bid $\gamma(n) \cdot v$, 
where $\gamma(n) \in [0,1]$ depends on $n$.\footnote{For an educational derivation, see, e.g., \url{https://cs.brown.edu/courses/cs1951k/lectures/2020/first_price_auctions.pdf}}  
When $n$ is known \emph{a priori}, this equilibrium bidding rule ensures that the first-price auction satisfies Bayesian UIC.  
The first-price auction also satisfies MIC, as shown in \cite{roughgardeneip1559}, 
and it satisfies global SCP since it always selects the bidders with the highest values.

However, Bayesian UIC here crucially relies on the permissioned setting.  
In designing a TFM for blockchains, 
we require a mechanism that works for any number $n$ of users.
In such a setting, no single dominant bidding strategy exists that works for all possible $n$.
\end{remark}

\begin{remark}[Smaller Reserve in the Permissioned Setting]
If the number $n$ of users is known in advance (e.g., in a permissioned setting), the reserve can be set below $\phi^{-1}(0)$. This reduces the burned payment and increases both miner revenue and social welfare.

Formally, given $n$, reserve $r$, and a fake bid vector $\bff$ injected by the miner, let:
\begin{itemize}
    \item $P(r, n, \bff)$ be the expected total payment from honest users (excluding fake bids),
    \item $B(r, n, \bff)$ be the expected total amount burned (including fake bids).
\end{itemize}
Bayesian MIC holds if
\begin{equation}
\label{eq:reserve-MIC}
P(r, n, \emptyset) - B(r, n, \emptyset) \geq P(r, n, \bff) - B(r, n, \bff),
\end{equation}
where the LHS is the miner's expected utility under honest execution, and the RHS under fake bids~$\bff$.

For example, if values are i.i.d.~drawn from the uniform distribution over~$[0, 1]$, 
then $\phi(v) = 2v - 1$ and $\phi^{-1}(0) = 1/2$. 
The smallest $r$ satisfying \cref{eq:reserve-MIC} for various $n$ is as follows:

\begin{table}[h!]
	\centering
	\begin{tabular}{l|llllllll}
		$n$	 & 1     & 2     & 3     & 4     & 5     & 10    & 100   & 1000  \\ \hline
		$r$  & 0.250 & 0.305 & 0.335 & 0.355 & 0.370 & 0.410 & 0.480 & 0.497 \\ 
	\end{tabular}
\end{table}

As $n$ grows, $r \to \phi^{-1}(0) = 1/2$. 
Intuitively, for large $n$, the block is almost surely full, so we have 
$B(r, n, \bff) \approx B(r, n, \emptyset)$, and \cref{eq:reserve-MIC} reduces to
\[
P(r, n, \emptyset) \geq P(r, n, \bff),
\]
matching the condition in classical optimal mechanism design.
\end{remark}

\subsection{Impossibility of Bayesian UIC + Bayesian MIC + Bayesian $1$-SCP + Global SCP}
\label{sec:Bayesian-impossibility}

In this section, we will show that for any mechanism satisfying Bayesian UIC,
Bayesian MIC, Bayesian $1$-SCP and global SCP,
the social welfare must be zero.
Notice that if users' true values are sampled from a bounded distribution over $[0, M]$,
we can run the posted-price auction with a price $M$:
only bids at least~$M$ are confirmed, each pays~$M$, and all payments are burned.
Although some transactions are confirmed,
the mechanism is still undesired because everyone's utility is zero.
Our impossibility result in this section demonstrates that it is a consequence of simultaneously
achieving all four incentive compatibility properties.

We begin by recalling the following lemma from~\cite{crypto-tfm}.

\begin{lemma}[Theorem C.5 in \cite{crypto-tfm}]
\label{lem:zero-miner-revenue}
Suppose each user's true value is i.i.d.~sampled from a distribution $\mcal{D}$.
For any (possibly randomized) TFM 
satisfying Bayesian UIC, Bayesian MIC, and Bayesian $1$-SCP, 
and for any number $n$ of users, it must be \[
\underset{\bfb\sim\mcal{D}^{n}}{\E} [\mu(\bfb)] = 0.
\]	
\end{lemma}

\begin{lemma}
	\label{lem:util-not-change}
	Suppose each user's true value is i.i.d.~sampled from a distribution $\mcal{D}$.
	Consider any (possibly randomized) TFM satisfying Bayesian UIC and global SCP.
	Then, 
	\Hao{Double check the expectation.}
	for any $n\geq 1$, for any user $i$,
	and for any $b_i, b'_i$, 
	it holds that \[
		\underset{\bfb_{-i}\sim\mcal{D}^{n}}{\E}[\mathsf{SW}_{-i}(\bfb_{-i}, b_i)]
		=
		\underset{\bfb_{-i}\sim\mcal{D}^{n}}{\E}[\mathsf{SW}_{-i}(\bfb_{-i}, b'_i)].
	\] 
\end{lemma}
\begin{proof}
The proof is similar to the proof of \cref{lem:util-not-change-expost},
while the mechanism only satisfies Bayesian UIC instead of ex post UIC.
We define \[
\overline{x}_i(\cdot) = \underset{\bfb_{-i}\sim\mathcal{D}^{n}}{\E} [x_i(\bfb_{-i}, \cdot)], \quad 
\overline{p}_i(\cdot) = \underset{\bfb_{-i}\sim\mathcal{D}^{n}}{\E} [p_i(\bfb_{-i}, \cdot)].
\]
By Bayesian UIC, for any user $i$, when user $i$'s true value is $b_i$,
it must be 
\begin{align*}
\underset{\bfb_{-i}\sim\mathcal{D}^{n}}{\E} \left[\mathsf{util}_i(\bfb_{-i}, b_i) - \mathsf{util}_i(\bfb_{-i}, b'_i)\right]
& =
[b_i \cdot \overline{x}_i(b_i) - \overline{p}_i(b_i)] - [b_i \cdot \overline{x}_i(b'_i) - \overline{p}_i(b'_i)]\nonumber\\
& \leq
[b'_i - b_i][\overline{x}_i(b'_i) - \overline{x}_i(b_i)].
\end{align*}
By global SCP, for any bid vector $\bfb_{-i}$, when user $i$'s true value is $b_i$, it must be \[
	\mathsf{SW}_{-i}(\bfb_{-i}, b'_i) - \mathsf{SW}_{-i}(\bfb_{-i}, b_i)
	\leq \mathsf{util}_i(\bfb_{-i}, b_i) - \mathsf{util}_i(\bfb_{-i}, b'_i).
\]
Thus, we have \[
	\underset{\bfb_{-i}\sim\mathcal{D}^{n}}{\E}
	\left[\mathsf{SW}_{-i}(\bfb_{-i}, b'_i) - \mathsf{SW}_{-i}(\bfb_{-i}, b_i)\right]
	\leq [b'_i - b_i][\overline{x}_i(b'_i) - \overline{x}_i(b_i)].
\]

Now, for any $b_i$ and $b'_i$,
we can divide the interval $[b_i, b'_i]$ into $L$ equally sized segments $b_i^{(0)},\dots,b_i^{(L)}$ 
where $b_i^{(0)} = b_i$ and $b_i^{(L)} = b'_i$.
Then, we have 
\begin{align*}
	\underset{\bfb_{-i}\sim\mathcal{D}^{n}}{\E}
	\left[\mathsf{SW}_{-i}(\bfb_{-i}, b'_i) - \mathsf{SW}_{-i}(\bfb_{-i}, b_i)\right]
	&= \sum_{j=0}^{L-1} \underset{\bfb_{-i}\sim\mathcal{D}^{n}}{\E}
	\left[\mathsf{SW}_{-i}(\bfb_{-i}, b_i^{(j+1)}) - \mathsf{SW}_{-i}(\bfb_{-i}, b_i^{(j)})\right]\\
	&\leq \sum_{j=0}^{L-1}(b_j^{(j+1)} - b_j^{(j)}) \cdot
	\left[\overline{x}_i(b_i^{(j+1)}) - \overline{x}_i(b_i^{(j)})\right]\\
	&=\frac{b'_i - b_i}{L} \cdot 
	\left[\overline{x}_i(b'_i) - \overline{x}_i(b_i)\right].
\end{align*}
By taking the limit for $L \rightarrow \infty$, we have 
$\underset{\bfb_{-i}\sim\mathcal{D}^{n}}{\E}\left[\mathsf{SW}_{-i}(\bfb_{-i}, b_i)\right]
\leq
\underset{\bfb_{-i}\sim\mathcal{D}^{n}}{\E}	\left[\mathsf{SW}_{-i}(\bfb_{-i}, b'_i)\right]$.
Because the argument above holds for any $b_i$ and $b'_i$,
it must be $\underset{\bfb_{-i}\sim\mathcal{D}^{n}}{\E}\left[\mathsf{SW}_{-i}(\bfb_{-i}, b_i)\right]
=
\underset{\bfb_{-i}\sim\mathcal{D}^{n}}{\E}	\left[\mathsf{SW}_{-i}(\bfb_{-i}, b'_i)\right]$.
\end{proof}

\begin{lemma}
\label{lem:add-user}
Suppose each user's true value is i.i.d.~sampled from a distribution $\mcal{D}$.
Consider any (possibly randomized) TFM satisfying Bayesian UIC and global SCP.
Then, 
for any $n\geq 1$, for any user $i$,
and for any $b_i$,
it holds that \[
	\underset{\bfb_{-i}\sim\mcal{D}^{n}}{\E}[\mathsf{SW}_{-i}(\bfb_{-i})]
	=
	\underset{\bfb_{-i}\sim\mcal{D}^{n}}{\E}[\mathsf{SW}_{-i}(\bfb_{-i}, b_i)].
\] 
\end{lemma}
\begin{proof}
First, we show that $\underset{\bfb_{-i}\sim\mcal{D}^{n}}{\E}[\mathsf{SW}_{-i}(\bfb_{-i})]
=
\underset{\bfb_{-i}\sim\mcal{D}^{n}}{\E}[\mathsf{SW}_{-i}(\bfb_{-i}, 0_i)]$.
To see this, if $\underset{\bfb_{-i}\sim\mcal{D}^{n}}{\E}[\mathsf{SW}_{-i}(\bfb_{-i})]
<
\underset{\bfb_{-i}\sim\mcal{D}^{n}}{\E}[\mathsf{SW}_{-i}(\bfb_{-i}, 0_i)]$,
then consider a world consists of $n$ users.
The global coalition can inject a fake bid $0_i$, and the social welfare increases.
On the other hand, if $\underset{\bfb_{-i}\sim\mcal{D}^{n}}{\E}[\mathsf{SW}_{-i}(\bfb_{-i})]
>
\underset{\bfb_{-i}\sim\mcal{D}^{n}}{\E}[\mathsf{SW}_{-i}(\bfb_{-i}, 0_i)]$,
then consider a world consists of $n+1$ users where user $i$'s true value is zero.
Because user $i$'s utility is zero under an honest execution, 
we have $\underset{\bfb_{-i}\sim\mcal{D}^{n}}{\E}[\mathsf{SW}_{-i}(\bfb_{-i}, 0_i)] 
= \underset{\bfb_{-i}\sim\mcal{D}^{n}}{\E}[\mathsf{SW}(\bfb_{-i}, 0_i)]$.
If user $i$ drops out, the social welfare becomes $\underset{\bfb_{-i}\sim\mcal{D}^{n}}{\E}[\mathsf{SW}_{-i}(\bfb_{-i})]$,
which is larger than the honest execution.
Thus, it must be $\underset{\bfb_{-i}\sim\mcal{D}^{n}}{\E}[\mathsf{SW}_{-i}(\bfb_{-i})]
=
\underset{\bfb_{-i}\sim\mcal{D}^{n}}{\E}[\mathsf{SW}_{-i}(\bfb_{-i}, 0_i)]$.
The lemma directly follows from the argument above and \cref{lem:util-not-change}.
\end{proof}

\begin{lemma}
\label{lem:user-same-utility}
Suppose each user's true value is i.i.d.~sampled from a distribution $\mcal{D}$,
and let $u = \underset{b\sim\mcal{D}^{n}}{\E}[\mathsf{util}(b)]$ 
denote the expected utility of the lone user if only one user submits a bid drawn from $\mcal{D}$.
Then, for any possibly randomized TFM that simultaneously satisfy Bayesian UIC,
Bayesian MIC, Bayesian $1$-SCP,
and global SCP,
and for any $n$, it must be \[
\underset{\bfb\sim\mcal{D}^{n}}{\E}[\mathsf{SW}(\bfb)] = n\cdot u.
\]
\end{lemma}
\begin{proof}
We show the lemma by induction on the number of users $n$.
By \cref{lem:zero-miner-revenue},
the miner revenue must be zero.
Thus, $\underset{b\sim\mcal{D}}{\E}[\mathsf{SW}(b)] = \underset{b\sim\mcal{D}}{\E}[\mathsf{util}(b)] = u$,
and the base case $n=1$ holds.

Next, we show the induction step.
Suppose the lemma holds for $k$ users, i.e., $\underset{\bfb\sim\mcal{D}^{k}}{\E}[\mathsf{SW}(\bfb)] = k\cdot u$ for an integer $k$.
We first show that each user's expected utility is the same regardless of the identity.
By weak symmetry, for any $n$,
the social welfare $\underset{\bfb\sim\mcal{D}^{n}}{\E}[\mathsf{SW}(\bfb)]$ is independent of the user identities.
By \cref{lem:add-user},
for any user identity $i$, 
we have $\underset{\bfb\sim\mcal{D}^{n}}{\E}[\mathsf{SW}(\bfb)] 
= \underset{\bfb\sim\mcal{D}^{n+1}}{\E}[\mathsf{SW}_{-i}(\bfb)]$.
For any bid vector $\bfb$, we have $\mathsf{SW}(\bfb) = \mathsf{SW}_{-i}(\bfb) + \mathsf{util}_i(\bfb)$.
Thus, we have $\underset{\bfb\sim\mcal{D}^{n+1}}{\E}[\mathsf{util}_i(\bfb)]
= \underset{\bfb\sim\mcal{D}^{n+1}}{\E}[\mathsf{SW}(\bfb)]  - \underset{\bfb\sim\mcal{D}^{n}}{\E}[\mathsf{SW}(\bfb)] $,
which implies that each user's expected utility is the same regardless of the identity
for any $n$.

By assumption, we have $\underset{\bfb\sim\mcal{D}^{k}}{\E}[\mathsf{SW}(\bfb)] = k\cdot u$.
Because the miner revenue is always zero and $\underset{\bfb\sim\mcal{D}^{k}}{\E}[\mathsf{SW}(\bfb)] 
= \underset{\bfb\sim\mcal{D}^{k+1}}{\E}[\mathsf{SW}_{-i}(\bfb)]$,
we have $\underset{\bfb\sim\mcal{D}^{k+1}}{\E}[\mathsf{util}_i(\bfb)] = u$ for any user $i$.
Finally, $\underset{\bfb\sim\mcal{D}^{k+1}}{\E}[\mathsf{util}_i(\bfb)] = 
\underset{\bfb\sim\mcal{D}^{k+1}}{\E}[\mathsf{util}_i(\bfb)] + \underset{\bfb\sim\mcal{D}^{k}}{\E}[\mathsf{SW}(\bfb)]
=  k\cdot u + u = (k+1) \cdot u$.
Thus, the lemma holds for $k+1$ users, and we conclude the proof.
\end{proof}


\begin{lemma}
\label{lem:largest-is-small}
Let $\mcal{D}$ be any non-negative distribution with bounded mean.
Given any natural numbers $n$,
let $X_1,\dots, X_n$ be i.i.d.~random variables drawn from $\mcal{D}$,
and let $Y_n = \max(X_1,\dots, X_n)$.
Then, it holds that \[
\lim_{n \rightarrow \infty}\E[Y_n]/n = 0.
\]
\end{lemma}
\begin{proof}
Given any real number $t > 0$,
for all $i \in [n]$, let $A^{t}_i$ and $B^{t}_i$ be the random variables defined as follows:
\[
A^{t}_i =
\begin{cases}
X_i, & \text{if } X_i \leq t,\\
0, & \text{otherwise};\\
\end{cases}
\text{ and }
B^{t}_i =
\begin{cases}
X_i, & \text{if } X_i > t,\\
0, & \text{otherwise}.\\
\end{cases}
\]
By definition, we have $X_i = A^{t}_i + B^{t}_i$.
Thus, we have \[
Y_n = \max(X_1,\dots, X_n) \leq \max(A^{t}_1,\dots, A^{t}_n) + \max(B^{t}_1,\dots, B^{t}_n)
\leq t + \max(B^{t}_1,\dots, B^{t}_n).
\]

Next, fix any $\epsilon > 0$.
Because the mean of $\mcal{D}$ is bounded,
there exists a threshold $t$ such that $\E[B^{t}_1] \leq \epsilon$.
Because $\max(B^{t}_1,\dots, B^{t}_n) \leq \sum_{i=1}^n B^{t}_i$,
we have \[
\E[\max(B^{t}_1,\dots, B^{t}_n)] \leq \sum_{i=1}^n \E[B^{t}_i] = n \cdot \E[B^{t}_1] \leq n \cdot \epsilon.
\]

Combining the arguments above, we have \[
\E[Y_n] \leq t + \E[\max(B^{t}_1,\dots, B^{t}_n)] \leq t + n \cdot \epsilon.
\]
Thus, we have \[
	\lim_{n \rightarrow \infty}\frac{\E[Y_n]}{n} \leq \epsilon.
\]
Because the argument holds for any $\epsilon > 0$,
we conclude that $\lim_{n \rightarrow \infty}\E[Y_n]/n = 0$.
\end{proof}

\begin{theorem}
Let $\mcal{D}$ be any non-negative distribution with bounded mean.
Suppose each user's true value is i.i.d.~sampled from $\mcal{D}$,
and suppose the block size is finite.
Then, for any possibly randomized TFM that simultaneously satisfy Bayesian UIC,
Bayesian MIC, Bayesian $1$-SCP,
and global SCP,
the expected social welfare must be zero.
\end{theorem}
\begin{proof}
Let $k$ be the block size, and let $m$ be the mean of $\mcal{D}$.
For the sake of contradiction, suppose there exists a number $n$ of users
such that the expected social welfare $U \coloneqq \underset{\bfb\sim\mcal{D}^{n}}{\E}[\mathsf{SW}(\bfb)]$ is positive.
By \cref{lem:user-same-utility}, we have $u \coloneqq \underset{b\sim\mcal{D}^{n}}{\E}[\mathsf{util}(b)] = U/n > 0$.

\cref{lem:largest-is-small} implies that for any $\epsilon > 0$,
and for sufficiently large $N$,
when there are $N$ users whose values are i.i.d.~sampled from $\mcal{D}$,
the expectation of the largest user's value is strictly smaller than $N \epsilon$.
Because only at most $k$ users can be confirmed,
the expected social welfare is strictly smaller than $Nk\epsilon$.
Let $\epsilon < u / k$.
By \cref{lem:user-same-utility}, when there are $N$ users,
the expected social welfare is $N u > Nk\epsilon$,
which leads to a contradiction.
\end{proof}

\section{Static Revelation Principle for Transaction Fee Mechanisms}
\label{sec:revelation}
Informally speaking, the {\it revelation principle} says that any mechanism
can be simulated by an equivalent \truthful{} mechanism.
Whenever the revelation principle holds, without loss
of generality, we may assume that the 
users' honest bidding strategy is simply truth-telling.
We focus here on a revelation principle for ``static'' mechanisms in
which users interact with a mechanism simply by submitting a bid.%
%
\footnote{More generally, typical revelation principles in
  auction theory transform non-direct mechanisms (in which the
  action space may be different from the valuation space, for example
  due to multiple rounds of user-mechanism interaction) into a direct mechanisms.}

TFMs differ from traditional auctions in that they additionally
need to satisfy MIC and collusion-resilience, 
and also the separation of the inclusion rule (executed
by the miner) and the confirmation/payment rules (executed
by the blockchain).
Therefore, revelation principles are not ``automatic.''
In this section, we prove that for any fixed $c$, the revelation
principle holds for TFMs that must satisfy 
UIC, MIC, and $c$-SCP.
As a direct corollary, the revelation principle
holds for TFMs that satisfy 
UIC, MIC, and global SCP.
The main subtlety 
in the proof is that when we bake the user's non-truth-telling
bidding rule $\beta$ into the mechanism itself, 
not only do we need to 
have the miner's inclusion rule 
execute $\beta$, 
the blockchain's confirmation/payment rules
must also double-check the correct enforcement of $\beta$ again.
Otherwise, a strategic miner 
may not honestly enforce $\beta$.

Notice that we focus on the mechanisms that satisfy $c$-SCP instead of OCA-proofness in this section.
$c$-SCP requires that following the honest bidding rule and the honest inclusion rule is a dominant strategy for a coalition consisting of the miner and at most $c$ users.
For a non-\truthful{} mechanism with the bidding rule $\beta$, $\beta$ is used by users to satisfy both the UIC and the $c$-SCP conditions.
In \cref{sec:revelation-single}, we assume the honest bidding rule $\beta$ only outputs a single bid.
Later in \cref{sec:revelation-multiple}, we generalize the proof so that $\beta$ can output any non-negative number of bids.

\subsection{Static Revelation Principle: Bidding Rules That Output Single Bid}
\label{sec:revelation-single}
The outcome of a TFM is defined as a tuple $(\bfx, \bfp, \mu)$
where $\bfx \in \{0, 1\}^t$ is a bit-vector indicating 
whether each bid is confirmed or not, $\bfp \in \R_{\geq 0}^t$
is the vector of payments for all bids, and 
$\mu \in \R_{\geq 0}$
is the miner's revenue.

\begin{theorem}[Static revelation principle]
\label{thm:revelation}
Let $c$ be any natural number.
Suppose $\Pi$ is a non-\truthful{} TFM for block size $k$ 
that is UIC, MIC, and $c$-SCP with an individually rational bidding rule $\beta$.
If $\beta$ always outputs a single bid,
then there exists a \truthful{} TFM $\Pi'$ for block size $k$
that is UIC, MIC, and $c$-SCP such that 
1) the honest bidding rule is truth-telling, 
and 2) given any vector of true values, the outcome under 
an honest execution of $\Pi$ is identically distributed as the 
outcome under an honest execution of $\Pi'$.
\end{theorem}
\begin{proof}
Suppose we are given a TFM 
$\Pi = (\beta, \bfI, \bfC, \bfP, \mu)$
where $\beta: \R_{\geq 0} \rightarrow \R$ 
denotes the user's honest bidding rule, 
$\bfI$ denotes the inclusion rule,
$\bfC$ denotes the confirmation rule, 
$\bfP$ denotes the payment rule, 
and $\mu$ denotes the miner revenue rule.
Given a vector $\bfb = (b_1,\dots,b_t)$, let $\beta(\bfb)$ denote the element-wise application of $\beta$;
that is, $\beta(\bfb) = (\beta(b_1),\dots,\beta(b_t))$.
We construct a \truthful{} mechanism
$\Pi' = (\bfI', \bfC', \bfP', \mu')$
as follows: 
\begin{itemize}
\item 
Inclusion rule $\bfI'$: Given a bid vector $\bfb$, the miner selects the bids by using $\bfI$ as if the bid vector is $\beta(\bfb)$.
That is, the miner selects the bids from $\bfb$ if the corresponding bids in $\beta(\bfb)$ are selected by $\bfI$.
\item 
Confirmation rule $\bfC'$, payment rule $\bfP'$ and miner revenue rule $\mu'$:
Given a created block $\bfB$, the mechanism applies $\bfC$ and $\bfP$ to $\beta(\bfB)$.
Then, a bid $b$ in the original block $\bfB$ is confirmed if and only if its corresponding bid $\beta(b)$ is confirmed when applying $\bfC$ to $\beta(\bfB)$, and its payment is the same as its corresponding bid when applying $\bfP$ to $\beta(\bfB)$.
The miner revenue is $\mu(\beta(\bfB))$; that is, applying $\mu$ as if the block is $\beta(\bfB)$.
\end{itemize}

Notice that $\beta$ is individually rational under $\Pi$,
so the payment corresponding to the bid $\beta(v)$ never exceeds $v$.
Thus, under the induced mechanism $\Pi'$, the payment never exceeds the bid, so the individual rationality of the payment rule is respected.
\Hao{Here we require IR to hold with probability $1$ instead of in expectation.}

Next, by directly checking the syntax, the user's confirmation probability and expected payment when submitting a bid $b$ under $\Pi'$ is identically distributed to submitting a bid $\beta(b)$ under $\Pi$ assuming the miner follows the inclusion rule.
Similarly, the miner's revenue when creating a block $\bfB$ under $\Pi'$is identically distributed to creating a block $\beta(\bfB)$ under $\Pi$.
Thus, the outcome under an honest execution of $\Pi$ is identically distributed as the outcome under an honest execution of $\Pi'$.

Finally, we show that the induced mechanism $\Pi'$ satisfies UIC, MIC, and global SCP.
Since submitting a bid $b$ and creating a block $\bfB$ under $\Pi'$ is equivalent to submitting a bid $\beta(b)$ and creating a block $\beta(\bfB)$ under $\Pi$, respectively,
the fact that $\Pi'$ satisfies UIC and global SCP directly follows from the fact that $\Pi$ satisfies UIC and global SCP under the user's honest bidding rule $\beta$.

For proving $\Pi'$ satisfies MIC, suppose $\Pi'$ is not MIC for the sake of contradiction.
In this case, there exists a bid vector $\bfb$ and a block $\tilde{\bfB}$ such that if the miner creates the block $\tilde{\bfB}$ instead of $\bfI'(\bfb)$, the miner's utility increases.
Then, consider another miner under the mechanism $\Pi$, and imagine that the bid vector is $\beta(\bfb)$.
Notice that $\mu(\beta(\tilde{\bfB})) = \mu'(\tilde{\bfB})$ and $\mu(\bfI(\beta(\bfb))) = \mu'(\bfI'(\bfb))$.
Thus, if the miner creates the block $\beta(\tilde{\bfB})$ instead of $\bfI(\beta(\bfb))$, the miner's utility increases.
It violates the fact that $\Pi$ is MIC.
Thus, $\Pi'$ must satisfy MIC.

\end{proof}

\subsection{Static Revelation Principle: Allowing Bidding Rules that Output Multiple Bids}
\label{sec:revelation-multiple}

In this section, we extend \cref{thm:revelation} to allow the bidding rule to output multiple bids.
The basic idea for proving the revelation principle is still the same --- trying to bake the non-truth-telling
bidding rule $\beta$ into the mechanism itself.
However, when $\beta$ may output multiple bids, a user $i$ will submit a vector $\bfb_i$ of bids instead of a single real number, and the inclusion rule may ask the miner to only select a subset of the bids from $\bfb_i$.
Therefore, the new challenge in this case is how to make the blockchain's confirmation/payment rules
also double-check the correct enforcement of $\beta$ again, given that the blockchain can only see a subset of the bids from $\bfb_i$.

To handle this difficulty, we need to relax the syntax of the inclusion rule.
Earlier in our definitions (\cref{sec:tfm-def}), we require
that the inclusion rule outputs a subset of the input bid vector.
In this section, we slightly relax this syntax requirement, and allow
the honest inclusion rule to  
append some metadata for each bid.\footnote{We assume the metadata do not occupy the block space. See \cref{remark:metadata} for more discussion.}
The purpose of the metadata in the honest execution is to encode the auxiliary information for the blockchain's confirmation/payment rules.
This relaxation in the syntax is non-essential for the impossibility results in this paper,
since all the proofs for our impossibility results (\cref{sec:deterministic,sec:randomglobalscp,sec:oca-impossibility}) still hold
even if the honest inclusion rule can append the metadata.
Now, we are ready to present the revelation principle.

\begin{theorem}
\label{thm:revelation-multiple}
Let $c$ be any natural number.
Given any non-\truthful{} TFM $\Pi$ for block size $k$ 
that is UIC, MIC, and $c$-SCP with an individually rational bidding rule, there exists a \truthful{} 
TFM $\Pi'$ for block size $k$ \Hao{double check}
that is UIC, MIC, and $c$-SCP such that 
1) the honest bidding rule is truth-telling, 
and 2) given any vector of true values, the outcome under 
an honest execution of $\Pi$ is identically distributed as the 
outcome under an honest execution of $\Pi'$.
\end{theorem}

\paragraph{Proof of \Cref{thm:revelation-multiple}.}
The rest of this section is dedicated to proving \cref{thm:revelation-multiple}.
Suppose we are given a TFM 
$\Pi = (\beta, \bfI, \bfC, \bfP, \mu)$
where $\beta: \R_{\geq 0} \rightarrow \R^*$ 
denotes the user's honest bidding strategy, 
$\bfI$ denotes the inclusion rule,
$\bfC$ denotes the confirmation rule, 
$\bfP$ denotes the payment rule, 
and $\mu$ denotes the miner revenue rule.
We construct a \truthful{} mechanism
$\Pi' = (\bfI', \bfC', \bfP', \mu')$
as follows: 
\begin{itemize}
\item 
Inclusion rule $\bfI'$: 
given a bid vector $\bfb = (b_1,\dots,b_t)$, 
let $\beta(\bfb) = \left(\beta(b_1), \ldots, \beta(b_t)\right)$.
Create a block ${\bf B}'$ 
that includes  the following: 
\begin{itemize}
\item 
for each $j \in [t]$, if some bid in 
the vector $\beta(b_j)$
is selected by the original inclusion rule $\bfI(\beta(\bfb))$, 
then include $b_j$ in $\bfB'$; 
\item 
for each included bid $b_j$, attach the following extra annotation in $\bfB'$:
\[
{\sf info}(b_j) = 
\text{bids in $\beta(b_j)$ selected by $\bfI(\beta(\bfb))$}
\]
\end{itemize}
\ignore{
A {\it valid} block
is therefore 
of the form $\{b, {\bf s}\}_j$
where ${\bf s} \subseteq \beta(b)$.
}
A {\it valid block} $\bfB$ is one in which the annotation for each 
$b_j \in \bfB$ is a subset 
of the vector $\beta(b_j)$.
\item 
The confirmation rule $\bfC'$, the payment rule $\bfP'$, and the miner revenue rule $\mu'$:
If the input block is not valid, 
then no one is confirmed, and the miner revenue is zero.
Else, parse the block as $\{b'_j, {\bf s}'_j = {\sf info}(b'_j)\}_j$, 
form an imaginary block $\{{\bf s}'_j\}_j$, 
and run the original confirmation rule $\bfC$, payment rule $\bfP$
and miner revenue rule $\mu$ 
on $\{{\bf s}'_j\}_j$.
A bid $b'_j$ in the input block 
is considered confirmed 
if the primary bid in $\beta(b'_j)$
is contained in ${\bf s}'_j$ and is confirmed by $\bfC(\{{\bf s}'_j\}_j)$, 
and its payment
is the sum of all payments of ${\bf s}'_j$ in $\bfP(\{{\bf s}'_j\}_j)$.
The miner is paid $\mu(\{{\bf s}'_j\}_j)$.
 

\ignore{
\begin{enumerate}
\item 
Parse $\bfs = (s_1,\dots,s_k)$,
and let $\hat{\bfs} = (s_1,\dots, s_{|r|})$.
\item 
Run $\bfC(\hat{\bfs})$, and let $S \subseteq [k]$ be a set of indices that indicating the confirmed bids in $\bfC(\hat{\bfs})$.
Run $\bfP(\hat{\bfs})$, and let $p_i$ be the payment of the confirmed bid $s_i$ for all $i \in S$.
\item 
For all $j \in [k]$, if there exists $i \in S$ such that $r_i = j$ and $s_i$ is the primary bid of $a_j$, then $a_j$ is confirmed.
Otherwise, $a_j$ is unconfirmed.
\item 
Let $A_j$ be the set of all indices such that $i \in S$ and $r_i = j$.
For all confirmed bid $a_j$, the payment is set to be $\sum_{i \in A_j} p_i$.
\item 
The miner is paid $\mu(\hat{\bfs})$.
\end{enumerate}
}
\end{itemize}
\ignore{
Conceptually, the induced inclusion rule outputs $\left(\bfa, \bfs, \bfr\right)$, where $\bfa$ contains the actual bids from the users, $\bfs$ simulates the created block as if all users submit the bids according to the bidding rule $\beta$, and $\bfr$ indicates which actual bid is associated with each simulated bid in $\bfs$.
Then, the mechanism simulates the original mechanism by running $\bfC, \bfP, \mu$ upon the simulated bids $\bfs$.
The induced confirmation rule confirms the bid if its associated primary bid is confirmed in the simulation.
All the confirmed bid must pay for all its associated simulated bids.

Notice that $\beta$ is individually rational under $\Pi$,
so the payment associated with the bids $\beta(a)$ never exceeds $a$.
Thus, under the induced mechanism $\Pi'$, the payment never exceeds the actual bid, so the individual rationality of the payment rule is respected.
}

More intuitively, in the \truthful{} mechanism
$\Pi'$, the miner simulates
$\beta$ by applying $\beta$ to all 
bids it has received, resulting in $\beta(\bfb)$.
It then simulates the original inclusion rule
of the non-\truthful{} mechanism, and the result is $\bfI(\beta(\bfb))$, 
indicating which of $\beta(\bfb)$ would have been included by the original
$\bfI$.
The miner includes in the block each bid $b_j$ 
such that at least one of 
$\beta(b_j)$ would have been included  
by $\bfI(\beta(\bfb))$, and for each included bid $b_j$,
it attaches the extra information which bids among $\beta(b_j)$
would have been included by $\bfI(\beta(\bfb))$.
This way, the blockchain's confirmation, payment, and miner-revenue
rules can double-check the correct enforcement of $\beta$
again, before simulating the effect of the original confirmation, payment,
and miner revenue rules of the non-\truthful{} mechanism --- formally,
the checking is accomplished through a block validity check 
in our reduction.


\begin{lemma}
Assume the miner follows the mechanism honestly.
Then, the outcome when all users follow the bidding rule $\beta$ under the mechanism $\Pi$
is identically distributed as the outcome when all users bid truthfully under the mechanism $\Pi'$.
\label{lem:same-outcome}
\end{lemma}
\begin{proof}
Let $(v_1,\dots,v_t)$ be the true values of all users.
Under mechanism $\Pi$, if all users follow the bidding rule $\beta$, the created block will be $\bfB = \bfI(\beta(v_1),\dots,\beta(v_t))$.
For any user $i$, since the fake bids have no intrinsic values, user $i$'s bid is considered as confirmed if and only if its primary bid in $\beta(v)$ is confirmed.
User $i$'s payment is the sum of the payments of all its bids $\beta(v)$.
The miner revenue in this case is $\mu(\bfB)$.

Under mechanism $\Pi'$, if all users bid truthfully, the created block will be $\bfB'$ parsed as $\{b'_j, {\sf info}(b'_j)\}_j$.
Notice that $\{{\sf info}(b'_j)\}_j$ is identically distributed as $\bfB$.
The bid $b'_j$ under mechanism $\Pi'$ is considered confirmed if the primary bid in $\beta(b'_j)$
is contained in ${\sf info}(b'_j)$ and is confirmed by $\bfC(\{{\sf info}(b'_j)\}_j)$, 
and its payment is the sum of all payments of ${\bf s}'_j$ in $\bfP(\{{\sf info}(b'_j)\}_j)$.
The miner revenue in this case is $\mu(\{{\sf info}(b'_j)\}_j)$.

Because $\{{\sf info}(b'_j)\}_j$ is identically distributed as $\bfB$, the outcomes of the honest execution of two mechanisms are also identically distributed by checking all the syntax above.

\end{proof}

\begin{lemma}
The following statements hold.
\begin{itemize}
	\item If $\Pi$ is UIC, then $\Pi'$ is UIC.
	\item If $\Pi$ is MIC, then $\Pi'$ is MIC.
	\item For all $c$, if $\Pi$ is $c$-SCP, then $\Pi'$ is $c$-SCP.
\end{itemize}
\label{lem:revelation-reduction}
\end{lemma}
\begin{proof}
We will prove three properties individually.

\paragraph{$\Pi'$ is UIC.}
For the sake of contradiction, suppose $\Pi'$ is not UIC.
That is, under mechanism $\Pi'$, there exists a user $i$ with true value $v$ and a bid vector $\bfb_{-i}$ such that user $i$'s utility increases if it submits a vector $\bff$ (possibly includes some fake bids) instead of $v$ when other users' bids are $\bfb_{-i}$.
Denote the above as Scenario 1.
Then, consider another Scenario 2, where the mechanism is $\Pi$, and other users' bids are $\beta(\bfb_{-i})$.
The strategic user $i$ can submit a vector $\beta(\bff)$ instead of $\beta(v)$.

By \cref{lem:same-outcome}, the confirmation probability and the payment of the bids $\bff$ ($v$, resp.) in Scenario 1 are identically distributed as the confirmation probability and the payment of the bids $\beta(\bff)$ ($\beta(v)$, resp.) in Scenario 2.
Thus, the expected utility of user $i$ if it submits $\bff$ ($v$, resp.) in Scenario 1 is the same as the expected utility of user $i$ if it submits $\beta(\bff)$ ($\beta(v)$, resp.) in Scenario 2.
Thus, in Scenario 2, user $i$'s utility increases if it $\beta(\bff)$ instead of $\beta(v)$.
It violates UIC under $\Pi$, so $\Pi'$ must satisfy UIC.

\paragraph{$\Pi'$ is MIC.}
For the sake of contradiction, suppose $\Pi'$ is not MIC.
That is, there exists a bid vector $\bfb$ and a block $\tilde{\bfB} = \{b'_j, {\bf s}'_j = {\sf info}(b'_j)\}_j$ such that if the miner creates the block $\tilde{\bfB}$ instead of $\bfI'(\bfb)$, the miner's utility increases.
Denote the above as Scenario 1.
Then, consider another Scenario 2, where the mechanism is $\Pi$, and the incoming bid vector is $\beta(\bfb)$.
The strategic miner can create a block $\{{\bf s}'_j\}_j$ instead of $\bfI(\beta(\bfb))$.

Recall the miner's utility is the revenue minuses the cost of injecting the fake bids.
The miner revenue of the block $\tilde{\bfB}$ ($\bfI'(\bfb)$, resp.) in Scenario 1 is the same as the miner revenue of the block $\{{\bf s}'_j\}_j$ ($\bfI(\beta(\bfb))$, resp.) in Scenario 2.
If $\{{\bf s}'_j\}_j$ contains some fake bids, the expected payment of those fake bids in Scenario 2 is also the same as the expected payment of the fake bids when the created block is $\tilde{\bfB}$ in Scenario 1.
Therefore, in Scenario 2, if the miner creates $\{{\bf s}'_j\}_j$ instead of $\bfI(\beta(\bfb))$, the miner's utility increases.
It violates MIC under $\Pi$, so $\Pi'$ must satisfy MIC.

\paragraph{$\Pi'$ is $c$-SCP for all $c$.}
For the sake of contradiction, suppose $\Pi'$ is not $c$-SCP for some $c$.
That is, there exist a coalition $C$ formed by the miner and at most $c$ users, a bid vector $\bfb_{-C}$, and a block $\tilde{\bfB}$ such that when all non-colluding users' bids are $\bfb_{-C}$, $C$'s utility increases if the created block is $\tilde{\bfB}$ instead of $\bfI'(\bfb_{-C}, \bfb_{C})$, where $\bfb_C$ is the true values of all users in $C$.
Denote the above as Scenario 1.
Then, consider another Scenario 2, which is the same as scenario 1 except that the mechanism is $\Pi$, and the bid vector of non-colluding user is $\beta(\bfb_{-C})$.
Parse $\tilde{\bfB} = \{b'_j, {\bf s}'_j = {\sf info}(b'_j)\}_j$.
The coalition $C$ can create a block $\{{\bf s}'_j\}_j$ instead of $\bfI\left(\beta(\bfb_{-C}, \bfb_C)\right)$.

Following the same argument as UIC and MIC,
we conclude that the joint utility of $C$ of the block $\tilde{\bfB}$ ($\bfI'(\bfb_{-C}, \bfb_C)$, resp.) in Scenario 1 is the same as the joint utility of $C$ of the block $\{{\bf s}'_j\}_j$ ($\bfI(\beta(\bfb_{-C}, \bfb_C))$, resp.) in Scenario 2.
Thus, in Scenario 2, if the coalition creates $\{{\bf s}'_j\}_j$ instead of $\bfI\left(\beta(\bfb_{-C}, \bfb_C)\right)$, the joint utility of $C$ increases.
It violates $c$-SCP under $\Pi$, so $\Pi'$ must satisfy $c$-SCP.
\end{proof}

The proof of \cref{thm:revelation-multiple} directly follows from \cref{lem:same-outcome} and \cref{lem:revelation-reduction}.

\begin{remark}
	\label{remark:metadata}
	We assume the metadata appended by the inclusion rule do not occupy the block space.
	If the metadata occupy the block space, it is not hard to see that we can encode
	a valid block using a vector of length at most $3k$.
	For example, we can divide the block into three parts:
	the first part encodes users' original bids $b_j$;
	the second part encodes $\beta(b_j)$;
	and the third part encodes the mapping between the first part and the second part.
\end{remark}


\bibliographystyle{alpha}
\bibliography{refs,gametheory,tfm-refs}

\newcommand{\etalchar}[1]{$^{#1}$}
\begin{thebibliography}{KMSW22}

\bibitem[ACH11]{giladgtcrypto}
Gilad Asharov, Ran Canetti, and Carmit Hazay.
\newblock Towards a game theoretic view of secure computation.
\newblock In {\em Eurocrypt}, 2011.

\bibitem[ADGH06]{gtcrypto02}
Ittai Abraham, Danny Dolev, Rica Gonen, and Joseph Halpern.
\newblock Distributed computing meets game theory: Robust mechanisms for rational secret sharing and multiparty computation.
\newblock In {\em PODC}, 2006.

\bibitem[AL11]{giladutilityindjournal}
Gilad Asharov and Yehuda Lindell.
\newblock Utility dependence in correct and fair rational secret sharing.
\newblock {\em Journal of Cryptology}, 24(1), 2011.

\bibitem[BCD{\etalchar{+}}]{eip1559}
Vitalik Buterin, Eric Conner, Rick Dudley, Matthew Slipper, and Ian Norden.
\newblock Ethereum improvement proposal 1559: Fee market change for eth 1.0 chain.
\newblock \url{https://github.com/ethereum/EIPs/blob/master/EIPS/eip-1559.md}.

\bibitem[BEOS19]{functional-fee-market}
Soumya Basu, David~A. Easley, Maureen O'Hara, and Emin~G{\"{u}}n Sirer.
\newblock Towards a functional fee market for cryptocurrencies.
\newblock {\em CoRR}, abs/1901.06830, 2019.

\bibitem[BGR23]{active-miner-tfm}
Maryam Bahrani, Pranav Garimidi, and Tim Roughgarden.
\newblock Transaction fee mechanism design with active block producers.
\newblock {\em arXiv preprint arXiv:2307.01686}, 2023.

\bibitem[CCWS21]{gt-leader-shi}
Kai-Min Chung, T-H.~Hubert Chan, Ting Wen, and Elaine Shi.
\newblock Game-theoretic fairness meets multi-party protocols: The case of leader election.
\newblock In {\em CRYPTO}. Springer-Verlag, 2021.

\bibitem[CGL{\etalchar{+}}18]{gt-fair-cointoss}
Kai{-}Min Chung, Yue Guo, Wei{-}Kai Lin, Rafael Pass, and Elaine Shi.
\newblock Game theoretic notions of fairness in multi-party coin toss.
\newblock In {\em TCC}, volume 11239, pages 563--596, 2018.

\bibitem[CMW23]{optimal-base-fee}
Davide Crapis, Ciamac~C Moallemi, and Shouqiao Wang.
\newblock Optimal dynamic fees for blockchain resources.
\newblock {\em arXiv preprint arXiv:2309.12735}, 2023.

\bibitem[CS23]{foundation-tfm}
Hao Chung and Elaine Shi.
\newblock Foundations of transaction fee mechanism design.
\newblock In {\em Proceedings of the 2023 Annual ACM-SIAM Symposium on Discrete Algorithms (SODA)}, pages 3856--3899. SIAM, 2023.

\bibitem[DR07]{gtcrypto06}
Yevgeniy Dodis and Tal Rabin.
\newblock Cryptography and game theory.
\newblock In {\em AGT}, 2007.

\bibitem[EFW22]{credibleauction-comm01}
Meryem Essaidi, Matheus V.~X. Ferreira, and S.~Matthew Weinberg.
\newblock Credible, strategyproof, optimal, and bounded expected-round single-item auctions for all distributions.
\newblock In Mark Braverman, editor, {\em 13th Innovations in Theoretical Computer Science Conference, {ITCS} 2022, January 31 - February 3, 2022, Berkeley, CA, {USA}}, volume 215 of {\em LIPIcs}, pages 66:1--66:19, 2022.

\bibitem[FMPS21]{dynamicpostedprice}
Matheus V.~X. Ferreira, Daniel~J. Moroz, David~C. Parkes, and Mitchell Stern.
\newblock Dynamic posted-price mechanisms for the blockchain transaction-fee market.
\newblock {\em CoRR}, abs/2103.14144, 2021.

\bibitem[FW20]{credibleauction-comm00}
Matheus V.~X. Ferreira and S.~Matthew Weinberg.
\newblock Credible, truthful, and two-round (optimal) auctions via cryptographic commitments.
\newblock In P{\'{e}}ter Bir{\'{o}}, Jason~D. Hartline, Michael Ostrovsky, and Ariel~D. Procaccia, editors, {\em {EC} '20: The 21st {ACM} Conference on Economics and Computation, Virtual Event, Hungary, July 13-17, 2020}, pages 683--712. {ACM}, 2020.

\bibitem[GKM{\etalchar{+}}13]{rdp00}
Juan~A. Garay, Jonathan Katz, Ueli Maurer, Bj{\"{o}}rn Tackmann, and Vassilis Zikas.
\newblock Rational protocol design: Cryptography against incentive-driven adversaries.
\newblock In {\em FOCS}, 2013.

\bibitem[GKTZ15]{rdp01}
Juan Garay, Jonathan Katz, Bj\"{o}rn Tackmann, and Vassilis Zikas.
\newblock How fair is your protocol? a utility-based approach to protocol optimality.
\newblock In {\em PODC}, 2015.

\bibitem[GLR10]{seqrationalcrypto}
Ronen Gradwohl, Noam Livne, and Alon Rosen.
\newblock Sequential rationality in cryptographic protocols.
\newblock In {\em FOCS}, 2010.

\bibitem[GTZ15]{rdp02}
Juan~A. Garay, Bj{\"{o}}rn Tackmann, and Vassilis Zikas.
\newblock Fair distributed computation of reactive functions.
\newblock In {\em DISC}, volume 9363, pages 497--512, 2015.

\bibitem[GY22]{greedy-tfm}
Yotam Gafni and Aviv Yaish.
\newblock Greedy transaction fee mechanisms for (non-) myopic miners.
\newblock {\em arXiv preprint arXiv:2210.07793}, 2022.

\bibitem[GY24]{barrier-tfm}
Yotam Gafni and Aviv Yaish.
\newblock Barriers to collusion-resistant transaction fee mechanisms.
\newblock {\em arXiv preprint arXiv:2402.08564}, 2024.

\bibitem[Har]{myerson-lecture-hartline}
Jason Hartline.
\newblock Lectures on optimal mechanism design.
\newblock \url{http://users.eecs.northwestern.edu/~hartline/omd.pdf}.

\bibitem[HT04]{gtcrypto00}
Joseph Halpern and Vanessa Teague.
\newblock Rational secret sharing and multiparty computation.
\newblock In {\em STOC}, 2004.

\bibitem[Kat08]{katzgametheory}
Jonathan Katz.
\newblock Bridging game theory and cryptography: Recent results and future directions.
\newblock In {\em Theory of Cryptography Conference}, pages 251--272. Springer, 2008.

\bibitem[KKLP23]{tiered-tfm}
Aggelos Kiayias, Elias Koutsoupias, Philip Lazos, and Giorgos Panagiotakos.
\newblock Tiered mechanisms for blockchain transaction fees.
\newblock {\em arXiv preprint arXiv:2304.06014}, 2023.

\bibitem[KMSW22]{logstar-gt-leader}
Ilan Komargodski, Shin’ichiro Matsuo, Elaine Shi, and Ke~Wu.
\newblock log*-round game-theoretically-fair leader election.
\newblock In {\em CRYPTO}, 2022.

\bibitem[KN08]{gtcrypto01}
Gillat Kol and Moni Naor.
\newblock Cryptography and game theory: Designing protocols for exchanging information.
\newblock In {\em TCC}, 2008.

\bibitem[LRMP23]{optimal-base-fee2}
Stefanos Leonardos, Dani{\"e}l Reijsbergen, Barnab{\'e} Monnot, and Georgios Piliouras.
\newblock Optimality despite chaos in fee markets.
\newblock In {\em International Conference on Financial Cryptography and Data Security}, pages 346--362. Springer, 2023.

\bibitem[LSZ19]{zoharfeemech}
Ron Lavi, Or~Sattath, and Aviv Zohar.
\newblock Redesigning bitcoin's fee market.
\newblock In {\em The World Wide Web Conference, {WWW} 2019}, pages 2950--2956, 2019.

\bibitem[Mye81]{myerson}
Roger~B. Myerson.
\newblock Optimal auction design.
\newblock {\em Math. Oper. Res.}, 6(1), 1981.

\bibitem[Ndi23]{optimal-base-fee3}
Abdoulaye Ndiaye.
\newblock Blockchain price vs. quantity controls.
\newblock {\em Quantity Controls (July 27, 2023)}, 2023.

\bibitem[OPRV09]{gtcrypto03}
Shien~Jin Ong, David~C. Parkes, Alon Rosen, and Salil~P. Vadhan.
\newblock Fairness with an honest minority and a rational majority.
\newblock In {\em {TCC}}, 2009.

\bibitem[PS17]{fruitchain}
Rafael Pass and Elaine Shi.
\newblock Fruitchains: A fair blockchain.
\newblock In {\em PODC}, 2017.

\bibitem[Rou20]{roughgardeneip1559}
Tim Roughgarden.
\newblock Transaction fee mechanism design for the {Ethereum} blockchain: An economic analysis of {EIP}-1559.
\newblock Manuscript, \url{https://timroughgarden.org/papers/eip1559.pdf}, 2020.

\bibitem[Rou21]{roughgardeneip1559-ec}
Tim Roughgarden.
\newblock Transaction fee mechanism design.
\newblock In {\em EC}, 2021.

\bibitem[SCW23]{crypto-tfm}
Elaine Shi, Hao Chung, and Ke~Wu.
\newblock What can cryptography do for decentralized mechanism design?
\newblock In {\em {ITCS}}, volume 251 of {\em LIPIcs}, pages 97:1--97:22. Schloss Dagstuhl - Leibniz-Zentrum f{\"{u}}r Informatik, 2023.

\bibitem[TY23]{pos-tfm}
Wenpin Tang and David~D Yao.
\newblock Transaction fee mechanism for proof-of-stake protocol.
\newblock {\em arXiv preprint arXiv:2308.13881}, 2023.

\bibitem[WAS22]{gt-fair-coin-complete}
Ke~Wu, Gilad Asharov, and Elaine Shi.
\newblock A complete characterization of game-theoretically fair, multi-party coin toss.
\newblock In {\em Eurocrypt}, 2022.

\bibitem[WSC24]{LP-tfm}
Ke~Wu, Elaine Shi, and Hao Chung.
\newblock {Maximizing Miner Revenue in Transaction Fee Mechanism Design}.
\newblock In Venkatesan Guruswami, editor, {\em 15th Innovations in Theoretical Computer Science Conference (ITCS 2024)}, 2024.

\bibitem[XFP23]{sequencer}
Matheus~Venturyne Xavier~Ferreira and David~C Parkes.
\newblock Credible decentralized exchange design via verifiable sequencing rules.
\newblock In {\em Proceedings of the 55th Annual ACM Symposium on Theory of Computing}, pages 723--736, 2023.

\bibitem[Yao]{yaofeemech}
Andrew Chi-Chih Yao.
\newblock {An Incentive Analysis of Some Bitcoin Fee Designs (Invited Talk)}.
\newblock In {\em ICALP 2020}.

\bibitem[ZCZ22]{bayesian-tfm}
Zishuo Zhao, Xi~Chen, and Yuan Zhou.
\newblock Bayesian-nash-incentive-compatible mechanism for blockchain transaction fee allocation.
\newblock \url{https://arxiv.org/abs/2209.13099}, 2022.

\end{thebibliography}

\appendix
\section{Comparison of Collusion-Resilience Notions}
\label{sec:comparison}
In this work, we discuss three different definitions of collusion-resilience: $c$-SCP, global SCP, and OCA-proofness.
Both global SCP and OCA-proofness focus on the global coalition consisting of the miner and all users, while $c$-SCP considers a coalition formed by the miner and at most $c$ users.
Both $c$-SCP and global SCP require following the honest bidding rule is a dominant strategy,
so if a TFM satisfies $c$-SCP for all $c$, it must also satisfy global SCP.
On the other hand, an OCA-proof TFM allows each user to bid according to an individually rational bidding strategy $\sigma: \mathbb{R}_{\geq 0} \rightarrow \mathbb{R}$ as long as $\sigma$ does not depend on others' bids.
Hence, for any TFM whose bidding rule only outputs a single real, the TFM satisfies global SCP implies it also satisfies OCA-proofness by simply taking the honest bidding rule as the individually rational bidding strategy $\sigma$.
However, in \cref{sec:tfm-def}, the honest bidding rule can output multiple bids in general.
In this case, global SCP may not imply OCA-proofness.
Consider the following somewhat contrived mechanism.

\begin{mdframed}
	\begin{center}
		{\bf Even Auction}
	\end{center}

\noindent\textbf{Parameters:} the block size is infinite.

\vspace{2pt}
\noindent\textbf{Mechanism:}
	\begin{itemize}
	\item 
	{\bf Bidding rule.}
	Each user submits two bids $(v,0)$, where $v$ is the true value.
	\item 
	{\bf Inclusion rule.}
	Include all pending bids.
	\item 
	{\bf Confirmation rule.}
	If the number of the included bids is odd, no one is confirmed.
	Otherwise, if the number of the included bids is even, all included bids are confirmed.
	\item 
	{\bf Payment rule and miner revenue rule.}
	Everyone pays nothing, and the miner revenue is zero.
	\end{itemize}
\end{mdframed}
\begin{lemma}
\label{lem:global-not-OCA}
The even auction above satisfies global SCP but not OCA-proofness.
\end{lemma}
\begin{proof}
If everyone follows the bidding rule, the number of the included bids must be even, so everyone's bid is confirmed.
Thus, the social welfare is maximized, and the mechanism satisfies global SCP.
However, when there are an odd number of users, following any bidding strategy $\sigma: \mathbb{R}_{\geq 0} \rightarrow \mathbb{R}$ will lead to an odd number of bids, and no one is confirmed.
Thus, it does not satisfy OCA-proofness.
\end{proof}

Notice that whether the individually rational bidding strategy $\sigma$ outputs a single bid or multiple bids is important.
If we require $\sigma$ to only output a single real as we defined in \cref{def:OCAproof},
there is no non-trivial TFM that can satisfy UIC, MIC, and OCA-proofness as we proved in \cref{thm:oca-random-impossibility}.
However, if we allow $\sigma$ to output multiple bids, we obtain a feasibility in \cref{sec:counterexample}.

When we analyze whether a TFM satisfies $c$-SCP for some $c$, $1$-SCP is the weakest requirement since $c'$-SCP always implies $c$-SCP for any $c' > c$ by definition.
Interestingly, even so, $1$-SCP is incomparable with global SCP and OCA-proofness.
The relations between $1$-SCP, global SCP, and OCA-proofness is depicted in \cref{fig:ICgraph}.

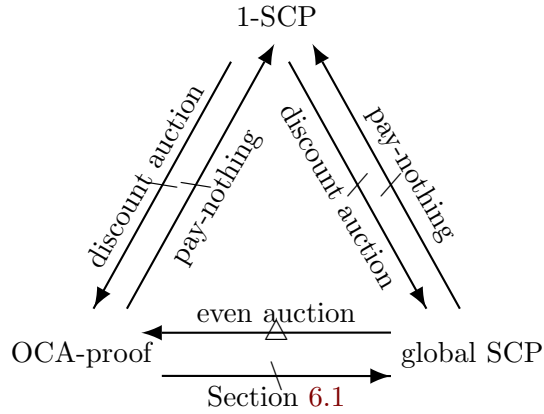
\begin{figure}[h]
    \centering
    \begin{tikzpicture}[
	squarednode/.style={rectangle, draw=white, minimum size=12mm},
	]
	\node[squarednode] (A) at (90:3) {$1$-SCP};
	\node[squarednode] (B) at (210:3) {OCA-proof};
	\node[squarednode] (C) at (330:3) {global SCP};
	\draw[decoration={markings,
	mark=at position 0.5 with \node[transform shape] (tempnode) {$\backslash$};,
	mark=at position 1 with {\arrow[scale=1.5,>={latex}]{>}}
	},postaction={decorate}, thick] 
	(A.225) -- node[rotate=60,above] {discount auction} (B.75);
	\draw[decoration={markings,
	mark=at position 0.5 with \node[transform shape] (tempnode) {$\backslash$};,
	mark=at position 0 with {\arrow[scale=1.5,>={latex}]{<}}
	},postaction={decorate}, thick] 
	(A.255) -- node[rotate=60,below] {pay-nothing} (B.45);
	\draw[decoration={markings,
	mark=at position 0.5 with \node[transform shape] (tempnode) {$\backslash$};,
	mark=at position 1 with {\arrow[scale=1.5,>={latex}]{>}}
	},postaction={decorate}, thick]
	(A.285) -- node[rotate=300,below] {discount auction} (C.135);
	\draw[decoration={markings,
	mark=at position 0.5 with \node[transform shape] (tempnode) {$\backslash$};,
	mark=at position 0 with {\arrow[scale=1.5,>={latex}]{<}}
	},postaction={decorate}, thick]
	(A.315) -- node[rotate=300,above] {pay-nothing} (C.105);
	\draw[decoration={markings,
	mark=at position 0.5 with \node[transform shape] (tempnode) {$\backslash$};,
	mark=at position 1 with {\arrow[scale=1.5,>={latex}]{>}}
	},postaction={decorate}, thick]
	(B.345) -- node[rotate=0,below] {\cref{sec:non-direct-oca-counterexample}} (C.195);
	\draw[decoration={markings,
	mark=at position 0.5 with \node[transform shape] (tempnode) {$\triangle$};,
	mark=at position 0 with {\arrow[scale=1.5,>={latex}]{<}}
	},postaction={decorate}, thick]
	(B.15) -- node[rotate=0,above] {even auction} (C.165);
\end{tikzpicture}
    \caption{Relationship between collusion-resilience notions.}
    \label{fig:ICgraph}
\end{figure}

We explain \cref{fig:ICgraph} in more detail below:
\begin{itemize}
\item 
global SCP $\overset{\triangle}{\implies}$ OCA-proofness:
As we discussed above, if the honest bidding rule only outputs a single real, global SCP implies OCA-proofness;
if the honest bidding rule outputs multiple reals, \cref{lem:global-not-OCA} states that global SCP does not imply OCA-proofness.
\item 
OCA-proofness $\centernot\implies$ global SCP:
In \cref{sec:non-direct-oca-counterexample}, we give a non-\truthful{} TFM that satisfies UIC, MIC and OCA-proofness.
However, in \cref{thm:random-scp-impossible}, we show that no non-trivial TFM can satisfy UIC, MIC and global SCP no matter it is \truthful{} or non-\truthful{}.
Combining these results, we conclude that OCA-proofness does not imply global SCP.
\item 
OCA-proofness $\centernot\implies$ $1$-SCP, global SCP $\centernot\implies$ $1$-SCP:
Conceptually, since a strategic miner can unilaterally decide which transactions are included, the miner can favor the colluding user's bid even if it harms the social welfare.
This is formally proven in \cref{lem:not-1-SCP}.
\item 
$1$-SCP $\centernot\implies$ OCA-proofness, $1$-SCP $\centernot\implies$ global SCP:
The counterexample is given in \cref{lem:discount}.
Intuitively, the attack relies on the fact that many colluding users ``amortize'' the cost of deviation.
\end{itemize}

Below, we introduce the (somewhat contrived) auctions for constructing the counterexamples.

\begin{mdframed}
	\begin{center}
		{\bf Pay-Nothing Auction}
	\end{center}

\noindent\textbf{Parameters:} the block size $k$.

\vspace{2pt}
\noindent\textbf{Mechanism:}
	\begin{itemize}
	\item 
	{\bf Bidding rule.}
	Each user submits the true value.
	\item 
	{\bf Inclusion rule.}
	Choose the highest up to $k$ bids.
	\item 
	{\bf Confirmation rule, payment rule and miner revenue rule.}
	All included bids are confirmed, everyone pays nothing, and the miner revenue is zero.
	\end{itemize}
\end{mdframed}
\begin{lemma}
\label{lem:not-1-SCP}
The pay-nothing auction above satisfies global SCP and OCA-proofness.
However, it does not satisfy $1$-SCP.
\end{lemma}
\begin{proof}
When everyone follows the mechanism, the users with the highest values will be confirmed, so the social welfare is maximized.
Thus, the mechanism satisfies global SCP and OCA-proofness.
However, if there are $k+1$ users with the true values $v_1 > \cdots > v_{k+1} > 0$.
When everyone follows the mechanism, the $(k+1)$-th user's bid will not be included in the block, so the user's utility is zero.
The platform can collude with the $(k+1)$-th user, and include its bid.
Since the payment is always zero, the coalition now has positive utility, which violates $1$-SCP.
\end{proof}

\begin{mdframed}
	\begin{center}
		{\bf Discount Auction}
	\end{center}

\noindent\textbf{Parameters:} the fixed price $r$, and the block size is infinite.

\vspace{2pt}
\noindent\textbf{Mechanism:}
	\begin{itemize}
	\item 
	{\bf Bidding rule.}
	Each user submits the true value.
	\item 
	{\bf Inclusion rule.}
	Include all pending bids.
	\item 
	{\bf Confirmation rule.}
	Let $S$ be the set of all included bids, and define $f$ as \[
	f(t) = \left\{\begin{matrix}
		r, &\text{if } t \leq 10; \\ 
		r/2, & \text{otherwise}.
		\end{matrix}\right.	
	\]
	Let $t^*$ be the largest integer such that there are $t^*$ bids in $S$ larger than $f(t^*)$.
	All bids larger than $f(t^*)$ are confirmed.
	\item 
	{\bf Payment rule.}
	Each confirmed bid pays $f(t^*)$.
	Unconfirmed bids pay nothing.
	\item 
	{\bf Miner revenue rule.}
	The miner revenue is zero, and all payments are burnt.
	\end{itemize}
\end{mdframed}
\begin{lemma}
\label{lem:discount}
The discount auction above satisfies $1$-SCP.
However, it does not satisfy global SCP or OCA-proofness.
\end{lemma}
\begin{proof}

We first show that the mechanism satisfies $1$-SCP.
Consider any coalition $C$ formed by the miner and user $i$.
Because the miner revenue is always zero, $C$'s joint utility can only come from user $i$.
Fix any bid vector $\bfb_{-i}$ from other users, user $i$'s confirmation and the payment satisfy the requirements in Myerson's lemma.
Thus, overbidding and underbidding cannot increase user $i$'s utility.
Next, we show that injecting fake bids does not help either.
To increase user $i$'s utility, the only way is to lower the payment.
However, regardless of the bids, the payment for each bid is always in $[r/2,r]$, so injecting fake bids can only lower the payment by at most $r/2$.
On the other hand, all confirmed bids must pay at least $r/2$, which means user $i$ has to pay at least $r/2$ for each fake bid.
Thus, injecting fake bids do not increase user $i$'s utility, so we conclude that the mechanism satisfies $1$-SCP.

Now, we show that the mechanism does not satisfy OCA-proofness.
Imagine a scenario where there are $10$ users each with the true value $r$.
Without injecting any fake bid, the lowest possible payment for each user is $f(10) = r$,
so the social welfare is at most $0$.
However, if the global coalition inject one fake bid $r$, everyone's payment will further lower to $f(11) = r/2$.
In this case, the social welfare becomes $10r - 11r/2 = 9r/2 > 0$.
Thus, the mechanism does not satisfy OCA-proofness.

The mechanism does not satisfy global SCP either since the bidding rule only outputs a single real.
If the mechanism satisfied global SCP, the global coalition could pick
the suggested bidding rule as the individually rational bidding
strategy $\sigma$ to maximize the social welfare.
Given that the mechanism does not satisfy OCA-proofness, it cannot satisfy global SCP either.
\end{proof}


\section{Alternative Proof of \cref{thm:random-scp-impossible}}

The proof in \cref{sec:randomglobalscp} relies on the fact that the miner can 
deviate from the inclusion rule and choose an arbitrary subset of transactions to include in the block.
In particular, \cref{lem:top-k} assumes that the global coalition 
can internally simulate the inclusion rule and strategically choose which transactions are included.
In this section, we present an alternative proof of \cref{thm:random-scp-impossible}
which does not rely on deviations from the inclusion rule.
As a result, this alternative proof rules out the possibility of achieving UIC, MIC,
and global SCP even when the inclusion rule is strictly enforced, for example using trusted hardware.

The two proofs, however, are incomparable.
The alternative proof in this section relies on the value domain being unbounded,
since it analyzes the behavior of payments and burning as user values tend to infinity (see the proof intuition below).
In contrast, the proof in \cref{sec:randomglobalscp} still holds
even when the value domain is bounded, say $[0, M]$,
assuming the confirmation probability under a single bid $a < M$ is non-zero.
This assumption is mild.
If it does not hold, i.e., if a single bid $a$ has positive confirmation probability
only when $a = M$,
then its payment must be $M$ by Myerson's lemma.
Consequently, the single user's utility is always zero.
Moreover, when the value domain is unbounded,
we can always find a sufficiently large value $M$ such that
confirmation probability is positive as shown in \cref{sec:randomglobalscp}.
Because each proof has its advantages,
we include both in this paper.


\paragraph{Proof intuition.}
Consider first a lonely world with only a single user submitting a true value~$v$.
By Myerson's lemma, the user's payment is given by $p(v) = v \cdot x(v) - \int_0^v x(t) \, dt$.
We show that $\lim_{v \rightarrow \infty} p(v)/v = 0$ (\cref{lemma:payment-burning-little}),
i.e., when the bid~$v$ is sufficiently large, the user's payment becomes negligible compared to~$v$.
Since the miner's revenue is upper bounded by the total payment, the miner's revenue is also negligible in the lonely world.

Now, consider a crowded world where $n$ users each submit the same bid~$v$.
We show that for any $n$, if the mechanism satisfies global SCP,
then $\lim_{v \rightarrow \infty} Q_n(v)/v = 0$, 
where $Q_n(v)$ denotes the total amount of burnt coin (\cref{lem:burning-myerson} and \cref{lemma:payment-burning-little}).
Conceptually, burnt coins represent the ``payment'' of the global coalition,
so the proof parallels the analysis of a single user's payment.

Next, we show that for any $n$, the miner revenue $\mu_n(v)$ in the crowded world 
must satisfy $\lim_{v \rightarrow \infty} \mu_n(v)/v = 0$ (\cref{lemma:miner-revenue-little}).
If this is not true, i.e., $\mu_n(v)$ is large with respect to $v$,
then there must be a user $i$ whose payment is at least $\mu_n(v)/n$,
which is also large with respect to $v$.
However, as we have seen, the miner revenue is negligible in the lonely world.
Thus, if there is only one user $i$ submitting $v$,
the miner can inject $n-1$ fake bids,
inflating user~$i$'s payment to $\mu_n(v)/n$.
Because the amount of burnt coins is negligible,
the miner can extract $\mu_n(v)/n$ for user $i$,
which is much larger than the honest case.

Finally, by global SCP, we show that adding an additional user to the system does not affect the joint utility of the original players (\cref{lem:add-user-expost}).
As we have shown, the miner's revenue is negligible,
so \cref{lem:add-user-expost} implies that each original user's utility almost remains unchanged when an additional user is added.
Thus, starting from the lonely world, where the only user has positive utility~$u_v$, the total social welfare in the crowded world becomes $n u_v$, growing linearly with~$n$.
However, because the block size is finite, this is impossible.

\begin{remark}[On the permissionless setting]
In this alternative proof, 
it is essential to consider the \emph{permissionless setting}, where the number of users is not fixed \emph{a priori} and strategic players may inject an arbitrary number of fake bids.  
This assumption is critical because the proof does not rely on any deviation from the inclusion rule.
Hence, it remains valid even if the miner is forced to follow the inclusion rule exactly.  

If, in contrast, we further assume that the number of users is fixed \emph{a priori}, the miner cannot inject any fake bids, and its strategy space becomes empty. 
In such a restricted setting, MIC would hold trivially. 
However, the second-price auction satisfies both UIC and global SCP,
which implies that the impossibility does not hold.
\end{remark}

The rest of this section is devoted to the formal proof.
For convenience, we will use the following notation.
Consider any number $n$ of users each with the true value $v$,
and let $\bfv_n = (\underbrace{v,\dots,v}_n)$.
Let $X_n(v) = \sum_i x_i(\bfv_n)$ denote the sum of all users' allocation probability,
$P_n(v) = \sum_i p_i(\bfv_n)$ denote the total payment,
$\mu_n(v) = \mu(\bfv_n)$ denote the miner revenue,
and $Q_n(v) = P_n(v) - \mu_n(v)$ denote the total burnt coin.

\begin{lemma}[Technical lemma implied by the proof of Myerson's Lemma~\cite{myerson,myerson-lecture-hartline}]
Let $f(z)$ be a non-decreasing function. 
Suppose that $y \cdot (f(z)-f(y)) \leq g(z) - g(y) \leq z \cdot (f(z)-f(y))  $ 
for any $z \geq y \geq 0$, and moreover, $g(0) = 0$.
Then, it must be that 
\[
g(z) = z \cdot f(z) - \int_0^{z} f(t) dt.
\]
\label[lemma]{lem:sandwich}
\end{lemma}

\begin{lemma}
\label[lemma]{lem:burning-myerson}
For any $n$,
if the mechanism satisfies global SCP, 
then $X_n(\cdot)$ must be non-decreasing,
and the function $Q_n(\cdot)$ satisfies
\[
	Q_n(v) = v \cdot X_n(v) - \int_0^{v} X_n(t) d t.
\]
\end{lemma}
\begin{proof}
If all users bid $b$ instead of their true value $v$,
the social welfare is $v \cdot X_n(b) - Q_n(b)$.
Consider any $0 \leq y \leq z$,
by global SCP, 
it must be $y \cdot X_n(y) - Q_n(y) \geq y \cdot X_n(z) - Q_n(z)$,
and $z \cdot X_n(z) - Q_n(z) \geq z \cdot X_n(y) - Q_n(y)$.
Thus, we obtain the following ``burning sandwich'' inequality: \[
y \cdot [X_n(z) - X_n(y)] \leq Q_n(z) - Q_n(y) \leq z \cdot [X_n(z) - X_n(y)],
\]
which implies that $X_n(\cdot)$ must be non-decreasing.
Because $p_i(0) = 0$ for any user $i$, we have $Q_n(0) = 0$.
By \cref{lem:sandwich},
we have \[
	Q_n(v) = v \cdot X_n(v) - \int_0^{v} X_n(t) d t.
\]
\end{proof}

\begin{lemma}
\label[lemma]{lemma:payment-burning-little}
Consider any natural number $k$ and any non-decreasing function $f(z): \mathbb{R}_{\geq 0} \rightarrow [0,k]$.
We define function $g(\cdot)$ as
\[
g(v) = v \cdot f(v) - \int_0^{v} f(t) dt.
\]
Then, it holds that $\lim_{v \rightarrow \infty} g(v)/v = 0$.
\end{lemma}
\begin{proof}
It is equivalent to show that \[
	\lim_{v \rightarrow \infty} \left[f(v) - \frac{1}{v} \int_0^{v} f(t) d t\right] = 0.
\]
It is well known that a bounded monotonic function must have a finite limit at infinity.
Thus, we have $\lim_{v \rightarrow \infty} f(v) = L$ for some $L \in [0, k]$.
	
Next, we will show that $\lim_{v \rightarrow \infty} \frac{1}{v} \int_0^{v} f(t) d t = L$.
In other words, we need to show that for any $\delta > 0$,
we can find a real number $v^*$ such that $\frac{1}{v} \int_0^{v} f(t) d t \in [L - \delta, L + \delta]$ for all $v > v^*$.
Because $\lim_{v \rightarrow \infty} f(v) = L$,
there exists a real number $v_0$ such that $f(v) \in [L - \delta, L + \delta]$ for all $v > v_0$,
so we have \[
	(L - \delta) \cdot (v - v_0) \leq \int_{v_0}^{v} f(t) d t 
	\leq (L + \delta) \cdot (v - v_0).
\]
Because $\lim_{v \rightarrow \infty} (v - v_0)/v = 1$,
we have \[
\lim_{v \rightarrow \infty}\frac{1}{v} \int_{v_0}^{v} f(t) d t \in [L - \delta, L + \delta].
\]
We can split the integral as \[
	\frac{1}{v}\int_0^{v} f(t) d t
	= \frac{1}{v}\int_0^{v_0} f(t) d t + \frac{1}{v}\int_{v_0}^{v} f(t) d t.
\]
The term $\int_0^{v_0} f(t) d t$ is a constant independent of $v$,
so we have $\lim_{v \rightarrow \infty} \frac{1}{v} \int_0^{v_0} f(t) d t = 0$.
Combining the argument above, we have
\begin{align*}
L - \delta \leq \lim_{v \rightarrow \infty} \frac{1}{v} \int_0^{v} f(t) d t 
= \lim_{v \rightarrow \infty}\frac{1}{v}\int_0^{v_0}  f(t) d t + 
\lim_{v \rightarrow \infty}\frac{1}{v} \int_{v_0}^{v} f(t) d t \leq L + \delta.
\end{align*}
Because the argument above holds for any $\delta > 0$, we have 
$\lim_{v \rightarrow \infty} \frac{1}{v} \int_0^{v} f(t) d t = L$.
Therefore, we have \[
\lim_{v \rightarrow \infty} \left[f(v) - \frac{1}{v} \int_0^{v} f(t) d t\right] = 0,\]
so we conclude the proof.
\end{proof}

\begin{lemma}
\label[lemma]{lemma:miner-revenue-little}
Suppose the mechanism satisfies UIC,
MIC,
and global SCP.
Then, for any $n$,
it must be $\lim_{v \rightarrow \infty} \mu_n(v)/v = 0$.
\end{lemma}
\begin{proof}

Consider any case where there are $n$ users all submitting $v$.
Because the total payment is $P_n(v)$,
there must exist a user whose payment is at least $P_n(v)/n$.

Now, imagine that there is only one user with identity $i$ submitting $v$,
so the miner's utility is $\mu_1(v)$ in the honest execution.
The miner randomly chooses $n-1$ fake identities $\mcal{I}$,
and checks whether $i$'s expected payment $p_i(\bfv_n)$ is at least $P_n(v)/n$
given $n$ identities $\mcal{I} \cup \{i\}$ all submitting $v$.
If not, the miner plays honestly without injecting any fake bids.
Otherwise, the miner injects $n-1$ fake bids $v$ on behalf of $\mcal{I}$.
In this case, $i$'s payment becomes $p_i(\bfv_n)$,
and the miner's utility is $p_i(\bfv_n) - Q_n(v)$.
By MIC, it must be 
\begin{equation}
\label{eq:MIC-1}
\mu_1(v) \geq p_i(\bfv_n) - Q_n(v) \geq \frac{P_n(v)}{n} - Q_n(v).
\end{equation}
Because the honest miner revenue is upper bounded by the payment,
we have $\mu_1(v) \leq P_1(v)$ and $\mu_n(v) \leq P_n(v)$.
Together with \cref{eq:MIC-1}, we have
\begin{equation}
	\label{eq:MIC-2}
	\mu_n(v) \leq n \cdot [P_1(v) + Q_n(v)].
\end{equation}

Next, by UIC and Myerson's lemma, the only user $i$'s payment is subject to 
$P_1(v) = v \cdot x(v) - \int_0^{v} x(v) d t$.
Moreover, by \cref{lem:burning-myerson}, we have $Q_n(v) = v \cdot X_n(v) - \int_0^{v} X_n(t) d t$.
Thus, by \cref{lemma:payment-burning-little}, we have 
$\lim_{v \rightarrow \infty} P_1(v)/v = 0$ and
$\lim_{v \rightarrow \infty} Q_n(v)/v = 0$.
Because the miner revenue must be non-negative,
fix any $n$,
\cref{eq:MIC-2} implies $\lim_{v \rightarrow \infty} \mu_n(v)/v = 0$.
\end{proof}

Next, we will show that for any user $i$,
no matter how user $i$ bids,
the joint utility of the rest of the world does not change.
Formally, given a bid vector $\bfb$ and a user $i$, 
let $\mathsf{SW}(\bfb)$ be social welfare,
i.e., the sum of the utilities of all users and the miner.
Let $\mathsf{SW}_{-i}(\bfb) = \mathsf{SW}(\bfb) - \mathsf{util}_i(\bfb)$,
i.e., the joint utility of the miner and all users except user $i$.

\begin{lemma}
\label[lemma]{lem:util-not-change-expost}
Consider any (possibly randomized) TFM satisfying UIC and global SCP.
Then, 
for any bid vector $\bfb_{-i}$ of an arbitrary length,
for any user $i$,
and for any $b_i, b'_i$, 
it holds that \[
	\mathsf{SW}_{-i}(\bfb_{-i}, b_i)
	=
	\mathsf{SW}_{-i}(\bfb_{-i}, b'_i).
\] 
\end{lemma}
\begin{proof}
Because the mechanism satisfies UIC,
user $i$'s payment is subject to Myerson's lemma.
Thus, given any bid vector $\bfb_{-i}$ from other users,
user's payment is $p_i(\bfb_{-i}, b_i) = b_i \cdot x_i(\bfb_{-i}, b_i) - \int_{0}^{b_i} x_i(\bfb_{-i}, t) dt$.
Suppose user $i$'s true value is $b_i$.
For any $b'_i$, if user $i$ bids $b'_i$ instead of $b_i$, 
user $i$'s loss can be represented as
\begin{align*}
	\mathsf{util}_i(\bfb_{-i}, b_i) - \mathsf{util}_i(\bfb_{-i}, b'_i)
	&= 	[b_i \cdot x_i(\bfb_{-i}, b_i) - p_i(\bfb_{-i}, b_i)] - [b_i \cdot x_i(\bfb_{-i}, b'_i) - p_i(\bfb_{-i}, b'_i)]\\
	&\leq
	[b'_i - b_i][x_i(\bfb_{-i}, b'_i) - x_i(\bfb_{-i}, b_i)].
\end{align*}
Next, By global SCP, for any bid vector $\bfb_{-i}$, it must be
\[
	\mathsf{SW}_{-i}(\bfb_{-i}, b_i) + \mathsf{util}_i(\bfb_{-i}, b_i)
	\geq
	\mathsf{SW}_{-i}(\bfb_{-i}, b'_i) + \mathsf{util}_i(\bfb_{-i}, b'_i).
\]
Notice that the argument above holds for any $b_i$ and $b'_i$.
Thus, for any $b_i$ and $b'_i$,
we have 
\begin{equation}
\label{eq:util-not-change-expost-1}	
	\mathsf{SW}_{-i}(\bfb_{-i}, b'_i) - \mathsf{SW}_{-i}(\bfb_{-i}, b_i) 
	\leq 
	[b'_i - b_i][x_i(\bfb_{-i}, b'_i) - x_i(\bfb_{-i}, b_i)].
\end{equation}
	
Now, 
we can divide the interval $[b_i, b'_i]$ into $L$ equally sized segments $b_i^{(0)},\dots,b_i^{(L)}$ 
where $b_i^{(0)} = b_i$ and $b_i^{(L)} = b'_i$.
By \cref{eq:util-not-change-expost-1}, we have 
\begin{align*}
	\mathsf{SW}_{-i}(\bfb_{-i}, b'_i) - \mathsf{SW}_{-i}(\bfb_{-i}, b_i)
	&= \sum_{j=0}^{L-1} 
	\mathsf{SW}_{-i}(\bfb_{-i}, b_i^{(j+1)}) - \mathsf{SW}_{-i}(\bfb_{-i}, b_i^{(j)})\\
	&\leq \sum_{j=0}^{L-1}(b_j^{(j+1)} - b_j^{(j)}) \cdot
	\left[x_i(\bfb_{-i}, b_i^{(j+1)}) - x_i(\bfb_{-i}, b_i^{(j)})\right]\\
	&=\frac{b'_i - b_i}{L} \cdot 
	\left[x_i(\bfb_{-i}, b'_i) - x_i(\bfb_{-i}, b_i)\right].
\end{align*}
By taking the limit for $L \rightarrow \infty$, we have 
$\mathsf{SW}_{-i}(\bfb_{-i}, b'_i) - \mathsf{SW}_{-i}(\bfb_{-i}, b_i) \leq 0$.
Because the argument above holds for any $b_i$ and $b'_i$,
it must be $\mathsf{SW}_{-i}(\bfb_{-i}, b_i) =	\mathsf{SW}_{-i}(\bfb_{-i}, b'_i)$.
\end{proof}

\begin{lemma}
\label[lemma]{lem:add-user-expost}
Consider any (possibly randomized) TFM satisfying UIC and global SCP.
Then, for any bid vector $\bfb_{-i}$ of an arbitrary length,
for any user $i$,
and for any $b_i$,
it holds that \[
	\mathsf{SW}_{-i}(\bfb_{-i})
	=
	\mathsf{SW}_{-i}(\bfb_{-i}, b_i).
\] 
\end{lemma}
\begin{proof}
First, we show that $\mathsf{SW}_{-i}(\bfb_{-i}) = \mathsf{SW}_{-i}(\bfb_{-i}, 0_i)$.
To see this, if $\mathsf{SW}_{-i}(\bfb_{-i})
<
\mathsf{SW}_{-i}(\bfb_{-i}, 0_i)$,
then consider a world where users' true values are $\bfb_{-i}$.
The global coalition can inject a fake bid $0_i$, and the social welfare increases.
On the other hand, if $\mathsf{SW}_{-i}(\bfb_{-i}) > \mathsf{SW}_{-i}(\bfb_{-i}, 0_i)$,
then consider a world where users' true values are $(\bfb_{-i}, 0_i)$.
Because user $i$'s utility is zero under an honest execution, 
we have $\mathsf{SW}_{-i}(\bfb_{-i}, 0_i) 
= \mathsf{SW}(\bfb_{-i}, 0_i)$.
If user $i$ drops out, the social welfare becomes $\mathsf{SW}_{-i}(\bfb_{-i})$,
which is larger than the honest execution.
Thus, it must be $\mathsf{SW}_{-i}(\bfb_{-i})
=
\mathsf{SW}_{-i}(\bfb_{-i}, 0_i)$.
The lemma directly follows from the argument above and \cref{lem:util-not-change-expost}.
\end{proof}

\begin{theorem}
\label{thm:expost-MIC-global-SCP}
No non-trivial, possibly randomized truthful TFM can simultaneously satisfy UIC,
MIC,
and global SCP (all in the ex post sense) when the block size is finite.
\end{theorem}
\begin{proof}
For the sake of contradiction, suppose the mechanism is non-trivial;
that is, there exists a scenario $\bfb = (b_1,\dots,b_t)$ such that 
some user $i$ has positive allocation probability $p_i(\bfb)$.
We first show that for a sufficiently large $v_0$, 
the allocation probability under a single bid $v_0$ is positive. 
Let $v_0$ be any real number such that $p_i(\bfb) \cdot v_0 > \sum_{j=1}^t b_j$.
Then, imagine that there is only one user with true value $v_0$.
If the allocation probability $x(v_0)$ is zero,
the global coalition can replace the user's primary bid with $b_i$, inject the fake bids  
$\bfb_{-i}$, and pretend that the bid vector is $\bfb$.  
This way, the user would get positive allocation probability $p_i(\bfb)$. 
Because its expected payment is at most $\sum_{j=1}^t b_j$, 
the user's expected utility (and thus the social welfare) is positive.
This violates global SCP,
so the allocation probability $x(v_0) $ must be positive.

Let $u_v$ denote the user's utility when there is only one user submitting its true value $v$.
By Myerson's lemma, it must be $u_v = \int_0^{v} x(t) d t$.
Because $x(v_0) > 0$, we can find a sufficiently large value $v_1$ such that 
$u_{v_1}/v_1 > x(v_0) / 2$.
Moreover, because $x(v)$ is non-decreasing,
we have $u_v / v \geq u_{v_1} / v_1$ for any $v \geq v_1$.

Next, by weak symmetry,
when $n$ users all bid $v$,
the social welfare is $\mathsf{SW}(\bfv_n)$ regardless of the users' identities.
Thus, by \cref{lem:add-user-expost},
for any user identity $i$, 
we have $\mathsf{SW}(\bfv_n) = \mathsf{SW}_{-i}(\bfv_{n+1})$.
Because $\mathsf{SW}(\bfv_{n+1}) = \mathsf{SW}_{-i}(\bfv_{n+1}) + \mathsf{util}_i(\bfv_{n+1})$,
we have $\mathsf{util}_i(\bfv_{n+1}) = \mathsf{SW}(\bfv_{n+1}) - \mathsf{SW}(\bfv_n)$.
In other words, given any $v$ and $n$,
every user's utility is the same regardless of the identity.
Given any $v$ and $n$, the sum of all users' utilities is $\mathsf{SW}(\bfv_n) - \mu(\bfv_n)$,
so a user's utility is $\mathsf{util}(\bfv_n) = \frac{1}{n}[\mathsf{SW}(\bfv_n) - \mu(\bfv_n)]$.
Therefore, we have
\begin{equation}
	\label{eq:SW-recursion}	
	\mathsf{SW}(\bfv_{n+1}) = \frac{n+1}{n}\cdot \mathsf{SW}(\bfv_n) - \frac{1}{n} \cdot \mu(\bfv_{n+1}).
\end{equation}

By \cref{lemma:miner-revenue-little}, for any natural number $N$,
we can find a sufficiently large value $v \geq v_1$ such that 
$\mu_{n+1}(v)/v < x(v_0)/4$ for all $n+1 \leq N$.
Because $u_v / v \geq u_{v_1} / v_1$ for any $v \geq v_1$,
we have \[
\frac{u_v}{v} \geq \frac{u_{v_1}}{v_1} > \frac{x(v_0)}{2} > \frac{2\mu_{n+1}(v)}{v}.
\]
Thus, we have $\mu_{n+1}(v) < u_{v}/2$ for all $n + 1 \leq N$,
and \cref{eq:SW-recursion} implies $\mathsf{SW}(\bfv_{n+1}) \geq \frac{n+1}{n}\cdot \mathsf{SW}(\bfv_n) - \frac{1}{2n} \cdot u_{v}$.
When there is only one user submitting its true value $v$,
the social welfare $\mathsf{SW}(\bfv_1)$ is at least $u_{v}$.
By the proof of induction, it is easy to see that $\mathsf{SW}(\bfv_n) \geq n\cdot u_{v}/2$.
Because the block size $k$ is finite, the social welfare $\mathsf{SW}(\bfv_n)$ must be upper bounded by $k v$.
Thus, we have $u_v/v \leq 2k/n$ for all $n \leq N$.
However, as we have shown, for a sufficiently large value $v_1$,
we have $u_v / v \geq u_{v_1} / v_1 > x(v_0) / 2$.
Consequently, if we choose $N$ such that $N > 4k /x(v_0)$,
we have $u_v / v \leq 2k/N < x(v_0)/2$, 
which leads to a contradiction.

\end{proof}

\end{document}